\pdfoutput=1
\documentclass[11pt]{article}
\usepackage{custom}
\usepackage{macros}

\allowdisplaybreaks[1]

\title{An Automated Approach to the Collatz Conjecture \\[1ex]
  \large (Explorations with Rewriting Systems, Matrix Interpretations, and SAT)}
\author{
  Emre Yolcu\thanks{\texttt{eyolcu@cs.cmu.edu}, \texttt{marijn@cmu.edu}. Computer Science Department, Carnegie Mellon University.} \and
  Scott Aaronson\thanks{\texttt{scott@scottaaronson.com}. Department of Computer Science, University of Texas at Austin.} \and
  Marijn J.\,H. Heule\footnotemark[1]}
\date{}

\begin{document}

\maketitle

\begin{abstract}
  We explore the Collatz conjecture and its variants through the lens of termination of string rewriting. We construct a rewriting system that simulates the iterated application of the Collatz function on strings corresponding to mixed binary--ternary representations of positive integers. We prove that the termination of this rewriting system is equivalent to the Collatz conjecture. We also prove that a previously studied rewriting system that simulates the Collatz function using unary representations does not admit termination proofs via natural matrix interpretations, even when used in conjunction with dependency pairs. To show the feasibility of our approach in proving mathematically interesting statements, we implement a minimal termination prover that uses natural/arctic matrix interpretations and we find automated proofs of nontrivial weakenings of the Collatz conjecture. Although we do not succeed in proving the Collatz conjecture, we believe that the ideas here represent an interesting new approach.
\end{abstract}


\hspace{0pt}
\vfill
\begin{figure}
  \includegraphics[width=\textwidth,center]{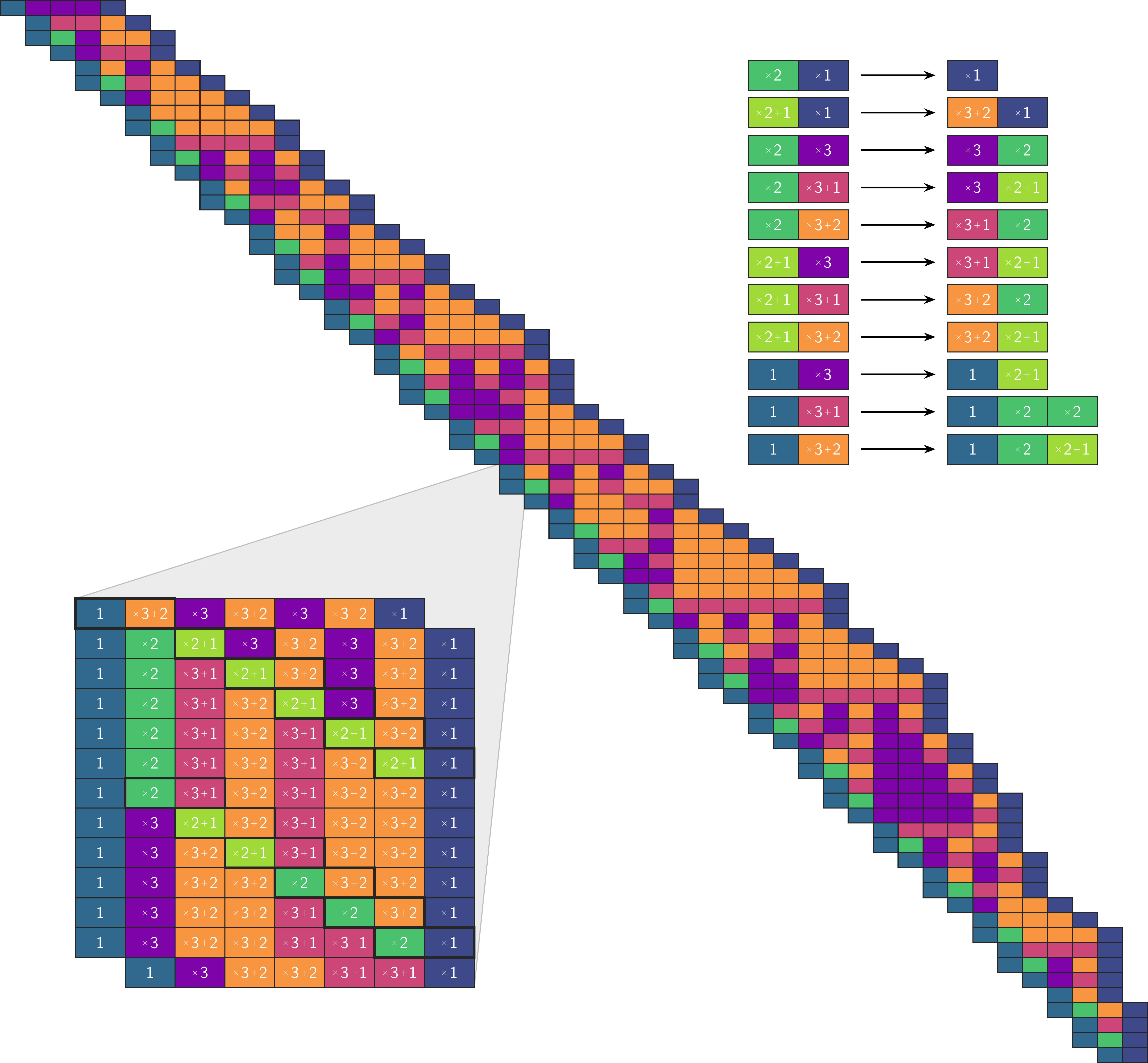}
  \caption*{Diagrammatic view of the Collatz trajectory of 27. Each row in the cascade is a mixed binary--ternary representation of an iterate, with each box corresponding to a digit. Rewrite rules are shown at the top right corner. The zoomed-in slice at the bottom left corner highlights the intermediate rewrite steps performed to change the least significant digit into binary while preserving the numeric value of the representation.}
\end{figure}
\vfill
\hspace{0pt}

\newpage
\tableofcontents

\newpage
\section{Introduction}
\label{sec:introduction}
Let $\N = \{0,1,2,\dots\}$ denote the natural numbers and $\Np = \{1,2,3,\dots\}$ denote the positive integers. We define the \emph{Collatz function} $C \colon \Np \to \Np$ as
\begin{equation*}
  C(n) =
  \begin{cases}
    n/2 & \text{if } n \equiv 0 \pmod 2 \\
    3n + 1 & \text{if } n \equiv 1 \pmod 2.
  \end{cases}
\end{equation*}
Given a function $f$ and a number $k \in \N$, the function $f^k$ denotes the \emph{$k$th iterate of $f$}, defined as $f^k \coloneqq \underbrace{f \circ f \circ \dots \circ f}_{k \text{ times}}$. The well-known \emph{Collatz conjecture} is the following:
\begin{conjecture}\label{conj:collatz}
  For all $n \in \Np$, there exists some $k \in \N$ such that $C^k(n) = 1$.
\end{conjecture}
This is a longstanding open problem and there is a vast literature dedicated to its study. For its history and comprehensive surveys on the problem, we refer the reader to Lagarias~\cite{Lag10,Lag11,Lag12}. See also the work by Tao~\cite{Tao20} for dramatic recent progress on the problem.

\begin{definition}[Convergent function]
  Consider a function $f \colon X \to X$ and a designated element $z \in X$. Given $x \in X$, the sequence of iterates $f^{\N}(x) \coloneqq (x, \allowbreak f(x), \allowbreak f^2(x), \dots)$ is called the \emph{$f$-trajectory of $x$}. If the trajectory $f^{\N}(x)$ reaches a cycle containing $z$, the trajectory is \emph{convergent to $z$ for $x$}. If for all $x \in X$ the trajectory $f^{\N}(x)$ is convergent to $z$, the function $f$ is \emph{convergent to $z$}.\footnote{For the functions we define throughout this paper, the designated element $z$ is clear from the context, so we refer to convergence without explicitly mentioning $z$. It is also apparent in all cases that $z$ is in a cycle, and we simply call a trajectory convergent if it contains $z$.}
\end{definition}

As of 2020, it has been checked by computation that for all $n \leq 2^{68}$ the Collatz trajectory of $n$ is convergent~\cite{Bar21}.\footnote{See \url{https://pcbarina.fit.vutbr.cz} for David Barina's ongoing verification effort.} This fact, along with the heuristic argument that on average the iterates in a trajectory tend to decrease~\cite{Lag85}, is often used as evidence to support the belief that the conjecture is true. On a different note, Conway proved that a generalization of the problem is undecidable~\cite{Con72}, which Kurtz and Simon later extended to a variant even closer to the Collatz conjecture~\cite{KS07}. While the prevailing opinion is that the Collatz conjecture holds, and the heuristic arguments suggest that an average trajectory is convergent, the undecidability results help explain why this seems to be so hard to show.

In this paper, we describe an approach based on termination of string rewriting to automatically search for a proof of the Collatz conjecture. Although trying to prove the Collatz conjecture via automated deduction is clearly a moonshot goal, there are two recent technological advances that provide reasons for optimism that at least some interesting variants of the problem might be solvable. First, the invention of the method of matrix interpretations~\cite{HW06,EWZ08} and its variants such as arctic interpretations~\cite{KW09,ST14} turns the quest of finding a ranking function to witness termination into a problem that is suitable for systematic search. Second, the progress in satisfiability (SAT) solving makes it possible to solve many seemingly difficult combinatorial problems efficiently in practice. Their combination, i.e., using SAT solvers to find interpretations, has so far been an effective strategy in solving challenging termination problems.

We make the following contributions:
\begin{itemize}
\item We show how a generalized Collatz function can be expressed as a rewriting system that is terminating if and only if the function is convergent.
\item We prove that no termination proof via natural matrix interpretations exists for a certain system by Zantema~\cite{Zan05} that simulates the Collatz function using unary representations of numbers. This result holds even when natural matrix interpretations are used in conjunction with the dependency pair transformation~\cite{AG00}.
\item We show that translations into rewriting systems that use non-unary representations of numbers are more amenable to automated methods, compared with the previously and commonly studied unary representations.
\item We automatically prove various weakenings of the Collatz conjecture. We observe that, for some of these weakenings, the only matrix interpretations that our termination tool was able to find involved relatively large matrices (of dimension $5$). Existing termination tools often limit their default strategies to search for small interpretations as they are tailored for the setting where the task is to quickly solve a large quantity of relatively easy problems. We make the point that, given more resources, the method of matrix interpretations has the potential to scale.
\item We observe that the phase-saving heuristic used by default in modern SAT solvers can degrade the performance of CDCL (conflict-driven clause learning) solvers~\cite{MS99} on formulas encoding the existence of matrix interpretations, whereas using negative branching can improve solver performance.
\end{itemize}

In \cref{sec:more-rewriting}, we present adaptations of our rewriting system that allow reformulating several more open problems in mathematics as termination problems of small size.

\section{Preliminaries}
\label{sec:preliminaries}

\subsection{String Rewriting Systems}
\label{sec:string-rewriting-systems}

We briefly review the notions related to string rewriting that are used in the rest of the paper. We adopt the notation of Zantema~\cite{Zan05} and Endrullis~et~al.~\cite{EWZ08}. For a comprehensive introduction to string or term rewriting, see Book~and~Otto~\cite{BO93} or Baader~and~Nipkow~\cite{BN98}, respectively.

\begin{definition}[String rewriting system]\label{defn:srs}
  Let $\Sigma$ be an alphabet, i.e., a set of symbols. A \emph{string rewriting system} (SRS) over $\Sigma$ is a relation $R \subseteq \Sigma^* \times \Sigma^*$. Elements $(\ell, r) \in R$ are called \emph{rewrite rules} and are usually written as $\ell \to r$. The system $R$ induces a \emph{rewrite relation} $\rew_R \coloneqq \{(s \ell t, s r t) \mid s, t \in \Sigma^*,\; \ell \to r \in R\}$ on the set $\Sigma^*$ of strings.
\end{definition}

\begin{definition}[Termination]\label{defn:termination}
  A relation $\rew$ on $A$ is \emph{terminating} if there is no infinite sequence $s_0, s_1, \dotsc \in A$ such that $s_i \to s_{i+1}$ for all $i \geq 0$. A terminating relation is also commonly called Noetherian, well-founded, or strongly normalizing. We write $\SN(\rew)$ to denote that $\rew$ is terminating.
\end{definition}

We often conflate an SRS $R$ with the rewrite relation it induces when discussing termination. In particular, we write ``$R$ is terminating'' instead of ``$\rew_R$ is terminating''. We also often do not specify the alphabet and take it to be the set of all the symbols that appear in a given rewriting system.

\begin{example}\label{ex:termination}
  \hfill
  \begin{enumerate}
  \item $R_1 = \{aa \to a\}$ is terminating, which is proved easily by observing that each rewrite decreases the length of a string by $1$.
  \item $R_2 = \{a \to aa\}$ is not terminating, since even a single occurrence of $a$ allows indefinitely producing more of the symbol $a$.
  \item $R_3 = \{ab \to ba\}$ is terminating, since all rewrite sequences eventually convert a given string into the form $b^* a^*$ and stop. An example sequence is shown below, with the underlines indicating the substrings that the rules are applied to.
    \begin{equation*}
      \underline{ab}bab \to_{R_3} b\underline{ab}ab \to_{R_3} bba\underline{ab} \to_{R_3} bb\underline{ab}a \to_{R_3} bbbaa
    \end{equation*}
  \end{enumerate}
\end{example}

A useful generalization of termination is relative termination:
\begin{definition}[Relative termination]\label{defn:relative-termination}
  Let $\rew_1, \rew_2$ be relations on $A$. Then $\rew_1$ is said to be \emph{terminating relative to} $\rew_2$ if there is no infinite sequence $s_0, s_1, \dotsc \in A$ such that
  \begin{itemize}
  \item $s_i \to_1 s_{i+1}$ for infinitely many values of $i$,
  \item $s_i \to_2 s_{i+1}$ for all other values of $i$.
  \end{itemize}
  We write $\SN(\rew_1 \rel \rew_2)$ to denote that $\rew_1$ is terminating relative to $\rew_2$.
\end{definition}
Equivalently, $R$ is said to be terminating relative to $S$ if every sequence of rewrites for the system $R \cup S$ applies the rules from $R$ at most finitely many times. Note that we have $\SN(R \rel \varnothing) = \SN(R)$ by the definition, so relative termination is indeed a generalization.

\begin{example}\label{ex:relative-termination}
  Let $R = \{aa \to aba\}$ and $S = \{b \to bb\}$. The system $R \cup S$ is not terminating since $S$ gives way to infinite rewrite sequences; however, $R$ is terminating relative to $S$ because $R$ by itself is terminating and applying the rule $b \to bb$ does not facilitate further applications of the rule $aa \to aba$.
\end{example}

Relative termination allows proofs of termination to be broken into steps as codified by the following theorem.
\begin{theorem}[Rule~removal~{\cite[Theorem~1]{Zan05}}]\label{thm:relative-termination}
  Let $R$ be an SRS\@. If there exists a subset $T \subseteq R$ such that $\SN(T \rel R)$ and $\SN(R \setminus T)$, then $\SN(R)$.
\end{theorem}
This theorem allows us to ``remove rules'' in the following way. When proving $\SN(R)$, if we succeed at finding a subset $T$ satisfying $\SN(T \rel R)$, the proof obligation becomes weakened to $\SN(R \setminus T)$, where the rules of $T$ are no longer present. This removal of rules can be repeated until no rules remain, thus producing a stepwise proof of termination.

Another useful technique is reversal:
\begin{definition}[Reversal]\label{defn:reversal}
  Let $s = s_1 \dots s_n \in \Sigma^*$ be a string and $R$ be an SRS\@. Letting $\rev{s} \coloneqq s_n \dots s_1$, the \emph{reversal of $R$} is defined as $\rev{R} \coloneqq \{\rev{\ell} \to \rev{r} \mid \ell \to r \in R\}$.
\end{definition}
\begin{lemma}[Rule~reversal~{\cite[Lemma~2]{Zan05}}]\label{lem:reversal}
  Let $R,S$ be SRSs. Then, $\SN(R \rel S)$ if and only if $\SN(\rev{R} \rel \rev{S})$.
\end{lemma}
Reversal is of interest because methods for proving termination are not necessarily invariant under reversal, that is, a given technique may fail to show the termination of a system $R$ while succeeding for its reversal $\rev{R}$.

Yet another important notion is top termination:
\begin{definition}[Top termination]\label{defn:top-termination}
  Let $R$ be an SRS over $\Sigma$. The \emph{top rewrite relation} induced by $R$ is defined as $\rew_{\topt{R}} \coloneqq \{(\ell s, r s ) \mid s \in \Sigma^*,\; \ell \to r \in R\}$. If $\rew_{\topt{R}}$ is terminating, $R$ is said to be \emph{top terminating}.
\end{definition}
In plain language, top termination allows rewrites to be performed only at the leftmost end of a string, so it requires a condition less strict than that of termination. As we will see in the next section (\cref{thm:monotone-algebra}), top termination problems can admit proofs of a more relaxed form compared with termination.

Relative top termination, i.e., proving $\SN(\topt{R} \rel S)$ for SRSs $R$ and $S$, is a crucial component in the dependency pair approach~\cite{AG00}, which reduces a termination problem to a relative top termination problem that often ends up being easier to solve via automated methods compared with the original problem. In order to avoid requiring familiarity with the dependency pair approach, we omit its discussion and instead prove a self-contained result (\cref{lem:T-top-termination}) that encapsulates dependency pairs in a more elementary manner for the specific rewriting systems that we consider in this paper.

\subsection{Interpretation Method}
\label{sec:interpretation-method}

Let $R$ be an SRS over the alphabet $\Sigma$. The main idea when proving termination is to show that there exists a well-founded (strict) order $>$ on $\Sigma^*$ such that, for all strings $s,t \in \Sigma^*$, if $s \to_R t$ then $s > t$. When this is the case, $R$ is terminating since any rewrite causes a ``strict decrease'' with respect to $>$ and by its well-foundedness there cannot exist an infinitely decreasing sequence. The condition over infinitely many strings can be replaced by another one that involves only the rules of $R$ by employing a \emph{reduction order}~\citetext{\citealp[Section~5.2]{BN98}; \citealp[Section~2.2]{BO93}}.
\begin{definition}[Reduction order]
  Let $>$ be a well-founded order on $\Sigma^*$. If for all $s,t,p,q \in \Sigma^*$, $s > t$ implies $psq > ptq$, then $>$ is called a \emph{reduction order}.
\end{definition}
\begin{theorem}[{\citetext{\citealp[Theorem~5.2.3]{BN98}; \citealp[Theorem~2.2.4]{BO93}}}]\label{thm:termination-reduction-order}
  Let $R$ be an SRS over the alphabet $\Sigma$. The system $R$ is terminating if and only if there exists a reduction order $>$ on $\Sigma^*$ such that $\ell > r$ holds for all $\ell \to r \in R$.
\end{theorem}

The interpretation method is a particular way of defining a reduction order. Instead of considering the strings directly, we interpret each symbol $\sigma \in \Sigma$ as a function $[\sigma] \colon A \to A$, where $A$ is some carrier set, for example of natural numbers or vectors, and we extend this interpretation from symbols to strings $s = s_1 \dots s_n \in \Sigma^*$ as $[s] \coloneqq [s_1] \circ \dots \circ [s_n]$. When the domain $A$ is already equipped with a well-founded order $>$, we can define a well-founded order $>_\sA$ on $\Sigma^*$ by
\begin{equation}\label{eq:interpretation-order}
  s >_\sA t \quad \text{ if and only if } \quad [s](x) > [t](x) \text{ for all } x \in A,
\end{equation}
where $\sA = (A, [\cdot]_\Sigma, >)$ is a structure with $[\cdot]_\Sigma \coloneqq \{[\sigma] \mid \sigma \in \Sigma\}$ denoting the collection of interpretations. For appropriately chosen interpretations, $>_\sA$ ends up being a reduction order:
\begin{theorem}[{\cite[Theorem~5.3.3]{BN98}}]\label{thm:interpretation-order-monotone}
  Let $\Sigma$ be an alphabet and let $>$ be a well-founded order on $A$. If for all symbols $\sigma \in \Sigma$ the interpretation $[\sigma] \colon A \to A$ is monotone with respect to $>$, then $>_\sA$ as defined in~\cref{eq:interpretation-order} is a reduction order on $\Sigma^*$.
\end{theorem}

Thus, after fixing $(A, >)$, proving termination reduces to finding suitable (monotone) interpretations of the symbols as functions~\cite[for details see][Section~5.3]{BN98}.\footnote{Termination is in fact equivalent to the existence of such interpretations~\cite[Proposition~1]{Zan94}.} As a concrete example, here is a simple proof of termination via the interpretation method.
\begin{example}
  Recall $R_3 = \{ab \to ba\}$ from \cref{ex:termination} with $\Sigma = \{a,b\}$. Fix $(\Np, >)$ with the usual order on positive integers, which is well-founded by the well-ordering principle.

  Let $[a](x) = x^2$ and $[b](x) = x + 1$. Denote $\sN = (\Np, \{[a], [b]\} , >)$. Both functions are monotone with respect to $>$, so by \cref{thm:interpretation-order-monotone} the relation $\mathord{>}_\sN = \{(s,t) \in \Sigma^* \times \Sigma^* \mid [s](x) > [t](x) \text{ for all } x \in \Np\}$ is a reduction order. For the rule $ab \to ba$, since
  \begin{equation*}
    ([a] \circ [b])(x) = (x + 1)^2 >  x^2 + 1 = ([b] \circ [a])(x)
  \end{equation*}
  holds for all $x \in \Np$, we have $ab >_\sN ba$. Thus, by \cref{thm:termination-reduction-order} the SRS $R_3$ is terminating.
\end{example}

In order to use the interpretation method for relative termination (resp.\ top termination), the above ideas can be generalized in the form of extended (resp.\ weakly) monotone algebras~\cite{EWZ08}.
\begin{definition}[Extended/weakly monotone algebra]\label{defn:monotone-algebra}
  Let
  \begin{itemize}
  \item $\Sigma$ be an alphabet,
  \item $A$ a set,
  \item $[\sigma] \colon A \to A$ an interpretation for every $\sigma \in \Sigma$,
  \item $>$ and $\gtrsim$ order relations over $A$ such that $>$ is well-founded and $\gtrsim$ satisfies $\mathord{> \cdot \gtrsim} \subseteq \mathord{>}$.
  \end{itemize}
  Letting $[\cdot]_\Sigma = \{[\sigma] \mid \sigma \in \Sigma\}$, the structure $(A, [\cdot]_\Sigma, >, \gtrsim)$ is a \emph{weakly monotone $\Sigma$-algebra} if for every $\sigma \in \Sigma$ the interpretation $[\sigma]$ is monotone with respect to $\gtrsim$. It is an \emph{extended monotone $\Sigma$-algebra} if, additionally, for every $\sigma \in \Sigma$ the interpretation $[\sigma]$ is monotone with respect to $>$.
\end{definition}

Relative termination (resp.\ top termination) is characterized as the existence of extended (resp.\ weakly) monotone algebras.
\begin{theorem}[{\cite[Theorem 2]{EWZ08}}]\label{thm:monotone-algebra}
  Let $R$ and $S$ be SRSs over the alphabet $\Sigma$. We have $\SN(R \rel S)$ (resp.\ $\SN(\topt{R} \rel S)$) if and only if there exists an extended (resp.\ weakly) monotone $\Sigma$-algebra $(A, [\cdot]_\Sigma, >, \gtrsim)$ such that
  \begin{itemize}
  \item for each rule $\ell \to r \in R$ we have $[\ell](x) > [r](x)$ for all $x \in A$,
  \item for each rule $\ell \to r \in S$ we have $[\ell](x) \gtrsim [r](x)$ for all $x \in A$.
  \end{itemize}
\end{theorem}

An effective way to automatically prove relative termination (or top termination) is to try to satisfy the conditions of the above theorem by fixing $(A, >, \gtrsim)$ and algorithmically searching for appropriate interpretations of the symbols. Matrix interpretations is a specific instance of this generic method. In the following section we describe (at a high level) the method of matrix interpretations. For more details and background, we refer the reader to existing work~\cite{HW06,EWZ08,KW09,ST14}.

It may seem an odd choice to focus on matrix interpretations and implement our own prover given that there is a wealth of other techniques for proving termination that are already implemented in the existing tools (e.g., \AProVE~\cite{GAB+17}, \Matchbox~\cite{Wal04}, \TTTT~\cite{KSZM09}). We did experiment with other methods; however, at least for the rewriting systems of the specific form in this paper it appears that the method of matrix interpretations is the most effective. For instance, one of the most difficult instances in this paper (in the sense of the smallest proof of termination that we could find) appears to require arctic matrix interpretations in all the proofs that we (and the authors of \AProVE{} and \Matchbox{}) could find. (There remain some versions of matrix interpretations that we have not yet experimented with~\cite{CGP10,NM11}.)

Additionally, we chose to implement our own minimal prover to have finer control over the SAT solving aspect of the search for an interpretation. We give a summary of our findings regarding this topic in \cref{sec:sat-solving}.

\subsection{Natural/Arctic Matrix Interpretations}
\label{sec:natural-arctic-matrix-interpretations}

\subsubsection{Natural Numbers}
\label{sec:natural-interpretations}
Let $R$ be an SRS over the alphabet $\Sigma$. Fix a dimension $d \in \Np$. In the method of natural matrix interpretations~\cite{EWZ08} we set $A = \N^d$ and define, for vectors $\vx, \vy \in \N^d$,
\begin{align*}
  \vx > \vy &\iff x_1 > y_1 \ \wedge \ x_i \geq y_i \text { for all } i \in \{2,\dots,d\}, \\
  \vx \gtrsim \vy &\iff x_i \geq y_i \text{ for all } i \in \{1,\dots,d\}.
\end{align*}
For interpreting each symbol $\sigma \in \Sigma$, we consider an affine function $[\sigma] \colon \N^d \to \N^d$:
\begin{equation*}\label{eq:affine-interpretation}
  [\sigma](\vx) = \mM_\sigma \vx + \vv_\sigma
\end{equation*}
In this way, the structure ($\N^d, [\cdot]_\Sigma, >, \gtrsim$) satisfies the requirements of \cref{defn:monotone-algebra} for a weakly monotone algebra. Additionally setting $(M_\sigma)_{1,1} = 1$ satisfies the requirements for an extended monotone algebra.

Furthermore, since the composition of affine functions is affine, for all $s \in \Sigma^*$ there exist $\mM_s \in \N^{d \times d}$, $\vv_s \in \N^d$ such that $[s](\vx) = \mM_s \vx + \vv_s$. As a result, given a collection $[\cdot]_\Sigma$ of natural matrix interpretations, the conditions of \cref{thm:monotone-algebra} can be checked by computing for each rule $\ell \to r \in R$ the corresponding matrices $\mM_\ell$, $\mM_r$ and vectors $\vv_\ell$, $\vv_r$ and comparing them elementwise~\cite[Lemma~4]{EWZ08}. Extend $\gtrsim$ to matrices $\mM, \mN \in \N^{d\times d}$ as
\begin{equation*}
  \mM \gtrsim \mN \iff M_{i,j} \geq N_{i,j} \text{ for all } i,j \in \{1,\dots,d\}.
\end{equation*}
Then, we have
\begin{equation}\label{eq:matrix-check}
  \begin{split}
    \mM_\ell \gtrsim \mM_r \ \wedge \ \vv_\ell > \vv_r &\iff [\ell](\vx) > [r](\vx) \text{ for all } \vx \in \N^d, \\
    \mM_\ell \gtrsim \mM_r \ \wedge \ \vv_\ell \gtrsim \vv_r &\iff [\ell](\vx) \gtrsim [r](\vx) \text{ for all } \vx \in \N^d.
  \end{split}
\end{equation}
This shows that it is decidable to check whether a given collection of interpretations of the form considered here constitutes a proof of termination.

\begin{example}
  Recall $R = \{aa \to aba\}$ and $S = \{b \to bb\}$ from \cref{ex:relative-termination} with $\Sigma = \{a,b\}$. The following functions constitute a matrix interpretations proof that shows $\SN(R \rel S)$.
  \begin{gather*}
    [a](\vx) =
    \begin{bmatrix}
      1 & 1 \\
      0 & 0
    \end{bmatrix}
    \vx +
    \begin{bmatrix}
      0 \\
      1
    \end{bmatrix}
    \qquad
    [b](\vx) =
    \begin{bmatrix}
      1 & 0 \\
      0 & 0
    \end{bmatrix}
    \vx +
    \begin{bmatrix}
      0 \\
      0
    \end{bmatrix}
  \end{gather*}
  It is easy to check that the above interpretations give an extended monotone algebra and that they satisfy the below relations for all $\vx \in \N^2$, which implies $\SN(R \rel S)$ via \cref{thm:monotone-algebra}.
  \begin{gather*}
    [aa](\vx) =
    \begin{bmatrix}
      1 & 1 \\
      0 & 0
    \end{bmatrix}
    \vx +
    \begin{bmatrix}
      1 \\
      1
    \end{bmatrix}
    >
    \begin{bmatrix}
      1 & 1 \\
      0 & 0
    \end{bmatrix}
    \vx +
    \begin{bmatrix}
      0 \\
      1
    \end{bmatrix} = [aba](\vx)
    \\[0.5em]
    [b](\vx) =
    \begin{bmatrix}
      1 & 0 \\
      0 & 0
    \end{bmatrix}
    \vx +
    \begin{bmatrix}
      0 \\
      0
    \end{bmatrix}
    \gtrsim
    \begin{bmatrix}
      1 & 0 \\
      0 & 0
    \end{bmatrix}
    \vx +
    \begin{bmatrix}
      0 \\
      0
    \end{bmatrix} = [bb](\vx)
  \end{gather*}
\end{example}

In order to automate the search for matrix interpretations given a rewriting system $R$, an approach that turns out to be effective is to encode all of the aforementioned constraints as a propositional formula in conjunctive normal form and use a SAT solver to look for a satisfying assignment. This additionally involves fixing a finite domain for the coefficients that can occur in the interpretations and encoding arithmetic over the chosen finite domain using propositional variables. (We discuss our choice of encoding in \cref{sec:sat-solving}.)

\subsubsection{Arctic Integers}
\label{sec:arctic-interpretations}
Natural matrix interpretations can also be adapted to the max--plus algebra of arctic\footnote{The name ``arctic'' was chosen to contrast with the min--plus algebra of the ``tropical'' semiring.} natural numbers $\A_\N = \N \cup \{-\infty\}$ or arctic integers $\A_\Z = \Z \cup \{- \infty\}$ as coefficients. This adaptation enables termination proofs in some cases where either there is no other automated method that leads to a proof or only relatively complicated proofs are known~\cite{KW09}.

Let $R$ be an SRS over the alphabet $\Sigma$, and let $\A$ be either one of $\A_\N$ or $\A_\Z$. In the method of arctic matrix interpretations~\cite{KW09,ST14} we use the arctic semiring $(\A, \oplus, \otimes)$, with the following operations:
\begin{align*}
  x \oplus y &= \max\{x, y\} \\
  x \otimes y &= x + y
\end{align*}
Let $>_\Z$ and $\geq_\Z$ denote the usual order relations over integers. Adopting the conventions that $x >_\Z -\infty$ for all $x \in \Z$ and that $-\infty \geq_\Z -\infty$, we define, for $x, y \in \A$,
\begin{align*}
  x \gg y &\iff x >_\Z y \; \vee \; x = y = -\infty, \\
  x \geq y &\iff x \geq_\Z y.
\end{align*}
We fix a dimension $d \in \Np$ and we extend both $\gg$ and $\geq$ to arctic vectors and arctic matrices elementwise. As these definitions allow $-\infty \gg -\infty$, the order $\gg$ is not well-founded. To ensure well-foundedness of $\gg$ we restrict the domain to $\N \times \A^{d-1}$. Then, for interpreting each symbol $\sigma \in \Sigma$, we consider an affine function $[\sigma] \colon \A^d \to \A^d$, written as
\begin{equation*}
  [\sigma](\vx) = \mM_\sigma \otimes \vx \oplus \vv_\sigma,
\end{equation*}
where we require $(M_\sigma)_{1,1} \geq 0$ or $(v_\sigma)_1 \geq 0$ to ensure that $[\sigma](\vz) \in \N \times \A^{d-1}$ for $\vz \in \N \times \A^{d-1}$. In this way, the structure $(\N \times \A^{d-1}, [\cdot]_\Sigma, \gg, \geq)$ satisfies the requirements of \cref{defn:monotone-algebra} for a weakly monotone algebra. Additionally choosing $\A = \A_\N$ and setting each $\vv_\sigma$ to be the $-\infty$ vector (consequently enforcing $(M_\sigma)_{1,1} \geq 0$) satisfies the requirements for an extended monotone algebra. This means that arctic integers (as opposed to arctic natural numbers) and vectors that contain finite elements (i.e., not all $-\infty$) are used only when proving top termination.

As the composition of affine functions over arctic vectors is affine, given a collection $[\cdot]_\Sigma$ of arctic matrix interpretations we can compute for each rule $\ell \to r \in R$ the corresponding arctic matrices $\mM_\ell$, $\mM_r$ and arctic vectors $\vv_\ell$, $\vv_r$. Elementwise comparisons of these matrices and vectors give sufficient conditions for the inequalities in \cref{thm:monotone-algebra}. In particular, due to Koprowski~and~Waldmann~\cite[Lemma~6.5]{KW09}, we have
\begin{equation*}
  \begin{split}
    \mM_\ell \gg \mM_r \ \wedge \ \vv_\ell \gg \vv_r &\implies [\ell](\vx) \gg [r](\vx) \text{ for all } \vx \in \A^d, \\
    \mM_\ell \geq \mM_r \ \wedge \ \vv_\ell \geq \vv_r &\implies [\ell](\vx) \geq [r](\vx) \text{ for all } \vx \in \A^d.
  \end{split}
\end{equation*}
Thus, the search for arctic matrix interpretations can be automated in a manner similar to natural matrix interpretations by encoding the aforementioned constraints as a SAT instance.

\begin{example}
  Let $\rP$ denote the following SRS\footnote{This SRS simulates the Abelian~sandpile~model~\cite{BTW87} over an infinite path. The symbols $\btL$, $\ttX$, $\btR$ represent the vertices on the path. At each instant, the weight of a vertex is identified by the number of $\ttA$s and $\ttB$s to its left until the next occurrence of a vertex symbol.}.
  \begin{equation*}
    \begin{array}{rcl}
      \btL \ttB & \to & \btL \ttA \ttX \\
      \ttA \ttA \ttX & \to & \ttB \ttX \ttA \\
      \ttA \ttB & \to & \ttB \ttA \\
      \ttX \ttB & \to & \ttA \ttX \\
      \ttA \btR & \to & \ttA \ttX \btR
    \end{array}
  \end{equation*}
  As an example of the different capabilities of arctic versus natural matrix interpretations, consider trying to prove the relative termination statement $\SN(\{\btL \ttB \to \btL \ttA \ttX\} \rel \rP)$. To our knowledge, there is no direct proof of this statement via natural matrix interpretations although there is one via the $4$-dimensional arctic matrix interpretations below. In the interpretations, $-$ indicates $-\infty$, and the multiplication operator $\otimes$ as well as the vectors of all $-\infty$s are omitted.
  {\setlength{\jot}{1em}
    \begin{NiceMatrixBlock}[auto-columns-width]
      \begin{gather*}
        [\ttA](\vx) =
        \begin{bNiceMatrix}
          0 & \fademinfty & \fademinfty & \fademinfty \\
          1 & 1 & \fademinfty & 0 \\
          \fademinfty & 0 & \fademinfty & \fademinfty \\
          \fademinfty & \fademinfty & 0 & \fademinfty
        \end{bNiceMatrix}
        \vx
        \qquad
        [\ttB](\vx) =
        \begin{bNiceMatrix}
          0 & \fademinfty & \fademinfty & \fademinfty \\
          1 & 1 & \fademinfty & 0 \\
          1 & 0 & 1 & \fademinfty \\
          \fademinfty & \fademinfty & 0 & \fademinfty
        \end{bNiceMatrix}
        \vx
        \qquad
        [\ttX](\vx) =
        \begin{bNiceMatrix}
          0 & \fademinfty & \fademinfty & \fademinfty \\
          \fademinfty & \fademinfty & 0 & \fademinfty \\
          \fademinfty & \fademinfty & \fademinfty & 0 \\
          \fademinfty & 0 & \fademinfty & \fademinfty
        \end{bNiceMatrix}
        \vx
        \\
        [\btL](\vx) =
        \begin{bNiceMatrix}
          0 & \fademinfty & 0 & \fademinfty \\
          \fademinfty & \fademinfty & \fademinfty & \fademinfty \\
          \fademinfty & \fademinfty & \fademinfty & \fademinfty \\
          \fademinfty & \fademinfty & \fademinfty & \fademinfty
        \end{bNiceMatrix}
        \vx
        \qquad
        [\btR](\vx) =
        \begin{bNiceMatrix}
          0 & \fademinfty & \fademinfty & \fademinfty \\
          \fademinfty & \fademinfty & \fademinfty & \fademinfty \\
          \fademinfty & \fademinfty & \fademinfty & \fademinfty \\
          \fademinfty & \fademinfty & \fademinfty & \fademinfty
        \end{bNiceMatrix}
        \vx
      \end{gather*}
    \end{NiceMatrixBlock}%
  }%
  With the above interpretations, the rules of $\rP$ satisfy the following relations, which, by \cref{thm:monotone-algebra}, proves that $\{\btL \ttB \to \btL \ttA \ttX\}$ is terminating relative to $\rP$.
  \begin{NiceMatrixBlock}[auto-columns-width]
    \begin{longtable*}{IJK}
      {[}\btL \ttB](\vx) =
      \begin{bNiceMatrix}
        1 & 0 & 1 & \fademinfty \\
        \fademinfty & \fademinfty & \fademinfty & \fademinfty \\
        \fademinfty & \fademinfty & \fademinfty & \fademinfty \\
        \fademinfty & \fademinfty & \fademinfty & \fademinfty
      \end{bNiceMatrix}
      \vx
      & \gg &
      \begin{bNiceMatrix}
        0 & \fademinfty & 0 & \fademinfty \\
        \fademinfty & \fademinfty & \fademinfty & \fademinfty \\
        \fademinfty & \fademinfty & \fademinfty & \fademinfty \\
        \fademinfty & \fademinfty & \fademinfty & \fademinfty
      \end{bNiceMatrix}
      \vx
      = {[}\btL \ttA \ttX](\vx)
      \\[3.075em]
      {[}\ttA \ttA \ttX](\vx) =
      \begin{bNiceMatrix}
        0 & \fademinfty & \fademinfty & \fademinfty \\
        2 & 1 & 2 & 0 \\
        1 & 0 & 1 & \fademinfty \\
        \fademinfty & \fademinfty & 0 & \fademinfty
      \end{bNiceMatrix}
      \vx
      & \geq &
      \begin{bNiceMatrix}
        0 & \fademinfty & \fademinfty & \fademinfty \\
        1 & 1 & \fademinfty & 0 \\
        1 & 0 & 1 & \fademinfty \\
        \fademinfty & \fademinfty & 0 & \fademinfty
      \end{bNiceMatrix}
      \vx
      = {[}\ttB \ttX \ttA](\vx)
      \\[3.075em]
      {[}\ttA \ttB](\vx) =
      \begin{bNiceMatrix}
        0 & \fademinfty & \fademinfty & \fademinfty \\
        2 & 2 & 0 & 1 \\
        1 & 1 & \fademinfty & 0 \\
        1 & 0 & 1 & \fademinfty
      \end{bNiceMatrix}
      \vx
      & \geq &
      \begin{bNiceMatrix}
        0 & \fademinfty & \fademinfty & \fademinfty \\
        2 & 2 & 0 & 1 \\
        1 & 1 & \fademinfty & 0 \\
        \fademinfty & 0 & \fademinfty & \fademinfty
      \end{bNiceMatrix}
      \vx
      = {[}\ttB \ttA](\vx)
      \\[3.075em]
      {[}\ttX \ttB](\vx) =
      \begin{bNiceMatrix}
        0 & \fademinfty & \fademinfty & \fademinfty \\
        1 & 0 & 1 & \fademinfty \\
        \fademinfty & \fademinfty & 0 & \fademinfty \\
        1 & 1 & \fademinfty & 0
      \end{bNiceMatrix}
      \vx
      & \geq &
      \begin{bNiceMatrix}
        0 & \fademinfty & \fademinfty & \fademinfty \\
        1 & 0 & 1 & \fademinfty \\
        \fademinfty & \fademinfty & 0 & \fademinfty \\
        \fademinfty & \fademinfty & \fademinfty & 0
      \end{bNiceMatrix}
      \vx
      = {[}\ttA \ttX](\vx)
      \\[3.075em]
      {[}\ttA \btR](\vx) =
      \begin{bNiceMatrix}
        0 & \fademinfty & \fademinfty & \fademinfty \\
        1 & \fademinfty & \fademinfty & \fademinfty \\
        \fademinfty & \fademinfty & \fademinfty & \fademinfty \\
        \fademinfty & \fademinfty & \fademinfty & \fademinfty
      \end{bNiceMatrix}
      \vx
      & \geq &
      \begin{bNiceMatrix}
        0 & \fademinfty & \fademinfty & \fademinfty \\
        1 & \fademinfty & \fademinfty & \fademinfty \\
        \fademinfty & \fademinfty & \fademinfty & \fademinfty \\
        \fademinfty & \fademinfty & \fademinfty & \fademinfty
      \end{bNiceMatrix}
      \vx
      = {[}\ttA \ttX \btR](\vx)
    \end{longtable*}
  \end{NiceMatrixBlock}
\end{example}

\subsection{Generalized Collatz Functions}
\label{sec:collatz-functions}

We consider instances of the following generalization of the Collatz function. Its variants have commonly appeared in the literature~\cite{Lag85,Raw85,Wag85,BM90,Kas92,Koh07,Mic15}.
\begin{definition}[Generalized Collatz function]\label{defn:generalized-collatz}
  Let $X$ be one of $\N$, $\Np$, or $\Z$ and define $X_\bot \coloneqq X \cup \{\bot\}$. A function $f \colon X_\bot \to X_\bot$ is a \emph{generalized Collatz function} if $f(\bot) = \bot$ and there exist an integer $d \geq 2$ and rational numbers $q_0, \dots, q_{d-1}, r_0, \dots, r_{d-1}$ such that for all $0 \leq i \leq d-1$ and all $n \in X$, we have a definition
  \begin{equation*}
    \begin{aligned}
      \text{either of the form} \quad & f(n) = q_i n + r_i && \text{if } n \equiv i \pmod d, \\
      \text{or of the form} \quad & f(n) = \bot && \text{if } n \equiv i \pmod d.
    \end{aligned}
  \end{equation*}
\end{definition}
In the above definition, we allow the representation of a partially defined function by mapping to $\bot$ in the undefined cases. We call a partial $f$ convergent if all $f$-trajectories contain $\bot$.

Note that the Collatz function corresponds to the case $d = 2$, $q_0 = 1/2$, $r_0 = 0$, $q_1 = 3$, $r_1 = 1$. Although the Collatz function is by far the most widely studied case, there are several other concrete examples of generalized Collatz functions the convergence of which is worth studying due to their connections to open problems in number theory and computability theory. We discuss these cases in \cref{sec:more-rewriting}.

\section{Rewriting the Collatz Function}
\label{sec:rewriting-collatz}

We start our discussion with systems that use unary representations and then demonstrate via examples that mixed-base representations can be more suitable for use with automated methods.

\subsection{Rewriting in Unary}
\label{sec:rewriting-unary}

The following system of Zantema~\cite{Zan05} simulates the iterated application of the Collatz function to a number represented in unary, and it terminates upon reaching $1$.
\begin{example}\label{ex:zantema-collatz}
  $\rZ$ denotes the following SRS, consisting of 5 symbols and 7 rules.
  \begin{equation*}
    \begin{array}[t]{rcl}
      \half \un \un & \to & \un \half
    \end{array}
    \qquad
    \begin{array}[t]{rcl}
      \un \un \half \spc & \to & \un \un \shift \spc \\
      \un \shift & \to & \shift \un \\
      \spc \shift & \to & \spc \half
    \end{array}
    \qquad
    \begin{array}[t]{rcl}
      \half \un \spc & \to & \triple \un \un \spc \\
      \un \triple & \to & \triple \un \un \un \\
      \spc \triple & \to & \spc \half
    \end{array}
  \end{equation*}
\end{example}
This system can be seen as encoding the execution of a Turing machine with cells that can be contracted/expanded. The symbols $\un$ and $\spc$ (blank) form the tape alphabet, while the symbols $\half$ (half), $\shift$ (shift), $\triple$ (triple) indicate the head along with the state of the machine. Through the following result, the Collatz conjecture can be reformulated as termination of string rewriting.
\begin{theorem}[{\cite[Theorem~16]{Zan05}}]
  $\rZ$ is terminating if and only if the Collatz conjecture holds.
\end{theorem}
While the forward direction of the above theorem is easy to see (since $\spc \half \un^{2n} \spc \to_{\rZ}^* \spc \half \un^{n} \spc$ for $n > 1$ and $\spc \half \un^{2n+1} \spc \to_{\rZ}^* \spc \half \un^{3n+2} \spc$ for $n \geq 0$), the backward direction is far from obvious because not every string corresponds to a valid configuration of the underlying machine.

As another example, consider the system below, which follows a theme similar to \cref{ex:zantema-collatz}.
\begin{example}\label{ex:waldmann-unary}
  $\rW$ denotes the following SRS, consisting of 4 symbols and 4 rules.
  \begin{equation*}
    \begin{array}{rcl}
      \half \un \un & \to & \un \half \\
      \un \half \spc & \to & \un \triple \spc \\
      \un \triple & \to & \triple \un \un \un \\
      \spc \triple & \to & \spc \half
    \end{array}
  \end{equation*}
\end{example}
Johannes Waldmann shared the above system (originally due to Zantema\footnote{\url{https://www.lri.fr/~marche/tpdb/tpdb-2.0/SRS/Zantema/z079.srs}}) with us, mentioning that its termination has yet to be proved via automated methods. Nevertheless, there is a simple reason for its termination: It simulates the iterated application of a partial generalized Collatz function $W \colon \Np_\bot \to \Np_\bot$ defined as follows, which is easily seen to be convergent.
\begin{equation*}\label{eq:waldmann-fn}
  W(n) =
  \begin{cases}
    3n/2 & \text{if } n \equiv 0 \pmod 2 \\
    \bot & \text{if } n \equiv 1 \pmod 2
  \end{cases}
\end{equation*}

\subsubsection{Nonexistence of Proofs via Natural Matrix Interpretations}
\label{sec:nonexistence-proofs-natural-matrix}
Natural matrix interpretations cannot be used to directly remove any of the rules from the above kind of unary rewriting systems that simulate certain maps, in particular the Collatz function.
\begin{definition}[$\N$-rational sequence]
  A sequence $x_1, x_2, \dotsc \in \N$ is called \emph{$\N$-rational} if there exists a matrix $\mM \in \N^{d\times d}$ and vectors $\vv,\vw \in \N^d$ of some dimension $d$ such that, for all $n$, we have $x_n = \vv^{\tr} \mM^n \vw$.
\end{definition}

The following is a consequence of Berstel's theorem~\cite[Theorem~8.1.1]{RB10}. (See also the remark by Soittola~\cite[page~318]{Soi76}.)

\begin{theorem}\label{thm:berstel}
  Let $x_1, x_2, \dots$ be an $\N$-rational sequence. Then there exist $m, p \in \N$ such that, for all $j \in \{0,1,\dots,p-1\}$, each of the subsequences $\{x_{m+kp+j}\}_{k=0}^\infty$ has its $n$th element equal to $(1 + o(1))P(n)\alpha^n$ for some nonzero polynomial $P$ and a constant $\alpha > 0$. In particular, each of these subsequences is monotonically nondecreasing after some finite point.
\end{theorem}

\begin{corollary}\label{cor:N-rational-collatz}
  There exists no $\N$-rational sequence $x_1, x_2, \dots$ satisfying $x_{8n+1} > x_{9n+2}$ for all $n$.
\end{corollary}
\begin{proof}[Proof\/\protect\footnotemark]
  \footnotetext{We thank Gjergji Zaimi for this argument (see \url{https://mathoverflow.net/a/270372}).}
  Let $p$ be as in \cref{thm:berstel}, and consider $n$ of the form $qp - 1$ for some large enough $q$. Then $(9n+2) - 8n+1 = n+1 = qp$, so $x_{8n+1}$ and $x_{9n+2}$ belong to the same subsequence. This means that for all sufficiently large $n$ of this form, we must have $x_{8n+1} \leq x_{9n+2}$.
\end{proof}

To make use of the above facts in ruling out the existence of natural matrix interpretations of the form described in \cref{sec:natural-interpretations}, we also need the following result proving that we can work with interpretations that are purely linear (instead of affine) by moving up to a higher dimension.
\begin{theorem}[{\cite[Theorem~6]{EWZ08}}]\label{thm:affine-linear}
  Let $R$ be an SRS over the alphabet $\Sigma$, and let $[\sigma](\vx) = \mM_\sigma \vx + \vv_\sigma$ denote a $d$-dimensional affine interpretation for each $\sigma \in \Sigma$. Define a $(d+1)$-dimensional linear interpretation for each $\sigma \in \Sigma$ as
  \begin{equation*}
    [\sigma]^+(\vx) \coloneqq \mD_\sigma \vx =
    \begin{bmatrix}
      \mM_\sigma & \vv_\sigma \\
      \vzero^{\tr} & 1
    \end{bmatrix}
    \vx,
  \end{equation*}
  where $\vzero$ denotes the $d$-dimensional zero vector. For each rule $\ell \to r \in R$, let $\mM_\ell$, $\mM_r$, $\vv_\ell$, $\vv_r$, $\mD_\ell$, $\mD_r$ denote the matrices and the vectors occurring in the corresponding interpretations (as obtained by composing the interpretations of the symbols). If we have $\mM_\ell \gtrsim \mM_r$ and $\vv_\ell \gtrsim \vv_r$, then $\mD_\ell \gtrsim \mD_r$. Additionally, if $\vv_\ell > \vv_r$, then $(D_\ell)_{1,d+1} > (D_r)_{1,d+1}$.
\end{theorem}
In informal terms, the above theorem states that if we have a collection of affine interpretations to remove a rule from a system (due to the strict decrease across the first elements of the corresponding two vectors), then there is also a collection of purely linear interpretations such that the top-right elements decrease strictly across the two matrices corresponding to this rule. With this fact at hand, we can prove the following.
\begin{theorem}\label{thm:Z-no-direct-natural-matrix}
  Let $\Sigma = \{\un, \spc, \half, \shift, \triple\}$. There exists no collection $[\cdot]_\Sigma$ of natural matrix interpretations of any dimension $d$ satisfying the requirements for an extended monotone $\Sigma$-algebra such that
  \begin{itemize}
  \item for at least one rule $\ell \to r \in \rZ$ we have $[\ell](\vx) > [r](\vx)$ for all $\vx \in \N^d$, and
  \item for the remaining $\ell' \to r' \in \rZ$ we have $[\ell'](\vx) \gtrsim [r'](\vx)$ for all $\vx \in \N^d$.
  \end{itemize}
\end{theorem}
\begin{proof}
  Assume for a contradiction that such a collection of natural matrix interpretations exists for some dimension $d-1$, with each symbol $\sigma \in \Sigma$ interpreted as $[\sigma](\vx) = \mM_\sigma \vx + \vv_\sigma$ and each matrix $\mD_\sigma$ defined as in \cref{thm:affine-linear}. Let $\ell \to r \in \rZ$ be a rule for which $[\ell](\vx) > [r](\vx)$ for all $\vx \in \N^{d-1}$, and denote $\rho = \{\ell \to r\}$. As the interpretations satisfy the requirements for extended monotonicity, for all $s,t \in \Sigma^*$, if $s \to_\rho t$, then $[s](\vx) > [t](\vx) \text{ for all } \vx \in \N^{d-1}$. This conclusion implies by~\cref{eq:matrix-check} that $\mM_s \gtrsim \mM_t$ and $\vv_s > \vv_t$, which, by \cref{thm:affine-linear}, further imply $\mD_s \gtrsim \mD_t$ and $(D_s)_{1,d} > (D_t)_{1,d}$. Thus, for all $s,t \in \Sigma^*$,
  \begin{equation}\label{eq:decrease-rho}
    s \to_\rho t \quad \implies \quad \mD_s \gtrsim \mD_t \text{ and } (D_s)_{1,d} > (D_t)_{1,d}.
  \end{equation}

  In Zantema's system $\rZ$, we can represent an arbitrary integer $k$ by the string $\spc \half \un^k \spc$, interpreted as
  \begin{equation*}
    \big[\spc \half \un^k \spc\big]^+(\vx) = \mD_{\spc\! \half} \mD_{\un}^k \mD_{\spc} \vx.
  \end{equation*}
  Let $\ve_i \in \N^d$ denote the $i$th standard basis vector, i.e., with $1$ at the $i$th position and $0$s elsewhere. Define
  \begin{align}\label{eq:N-rational-interpretation}
    \begin{split}
      f(k) &\coloneqq \ve_1^{\tr} \big[\spc \half \un^k \spc\big]^+(\ve_d) \\
           &= \ve_1^{\tr} (\mD_{\spc\! \half} \mD_{\un}^k \mD_{\spc}) \ve_d \\
           &= (\ve_1^{\tr} \mD_{\spc\! \half}) \mD_{\un}^k (\mD_{\spc} \ve_d) \\
           &= \vv^\tr \mD_{\un}^k \vw,
    \end{split}
  \end{align}
  where $\vv = \mD_{\spc\! \half}^\tr \ve_1$ and $\vw = \mD_{\spc} \ve_d$ are fixed, so $f(1), f(2), \dots$ is an $\N$-rational sequence.

  Now, since the system $\rZ$ simulates the Collatz function (with the odd case incorporating an additional division by $2$), for all numbers of the form $8n+1$ it can simulate the sequence of mappings $8n+1 \mapsto 12n+2 \mapsto 6n+1 \mapsto 9n+2$ that results from applying the Collatz map three times. Thus, for all $n$, we have
  \begin{equation*}
    \spc \half \un^{8n+1} \spc \; \to_{\rZ}^* \; \spc \half \un^{9n+2} \spc,
  \end{equation*}
  which crucially requires the use of every single rule in $\rZ$ (as it involves applying both $m \mapsto m/2$ and $m \mapsto (3m+1)/2$). This implies in particular that the rule $\ell \to r$ is used in the above derivation. Recall also that the interpretations of all rules in $\rZ$ are at least nonstrictly decreasing by assumption. This means along with~\cref{eq:decrease-rho} that, letting $s = \spc \half \un^{8n+1} \spc$ and $t = \spc \half \un^{9n+2} \spc$, we have $(D_s)_{1,d} > (D_t)_{1,d}$. From the definition of $f$, for all $n$,
  \begin{equation*}
    f(8n+1) = \ve_1^{\tr} \mD_s \ve_d = (D_s)_{1,d} > (D_t)_{1,d} = \ve_1^{\tr} \mD_t \ve_d = f(9n+2),
  \end{equation*}
  but we already argued in \cref{cor:N-rational-collatz} that no $\N$-rational sequence can satisfy this inequality for all $n$, contradiction.
\end{proof}

Natural matrix interpretations fail to be useful for the unary rewriting systems even when used in conjunction with dependency pairs. For the reader familiar with dependency pairs, we show how to adapt the above proof to go through in the setting where the dependency pair transformation is applied before searching for natural matrix interpretations. Recall that, for an SRS $R$ over $\Sigma$, letting $\DP(R)$ denote the SRS consisting of all dependency pairs of $R$ as defined by Arts and Giesl~\cite[Definition~3]{AG00}, we have $\SN(R)$ if and only if $\SN(\topt{\DP(R)} \rel R)$~\cite[Theorem~6]{AG00}. For a defined symbol $s \in \Sigma$, we write $\tup{s}$ to denote the corresponding tuple symbol.

Let $\rI = \{\tup{\spc} \shift \to \tup{\spc} \half,\ \tup{\spc} \triple \to \tup{\spc} \half \}$, and note that $\rI \subseteq \DP(\rZ)$. With an automated termination prover, the following is straightforward to establish for Zantema's system $\rZ$ after applying the dependency pair transformation and performing several steps of rule removal (e.g., using natural matrix interpretations).

\begin{lemma}\label{lem:Z-dependency-pairs}
  $\SN(\topt{\DP(\rZ)} \rel \rZ)$ if and only if $\SN(\topt{\rI} \rel \rZ)$.
\end{lemma}

We prove below that neither rule in $\rI$ can be removed using natural matrix interpretations. The proof is similar to that of \cref{thm:Z-no-direct-natural-matrix}, so we keep it relatively brief.

\begin{theorem}\label{thm:Z-no-dp-natural-matrix}
  Let $\Sigma = \{\un, \spc, \tup{\spc}, \half, \shift, \triple\}$. There exists no collection $[\cdot]_\Sigma$ of natural matrix interpretations of any dimension $d$ satisfying the requirements for a weakly monotone $\Sigma$-algebra such that
  \begin{itemize}
  \item for at least one rule $\ell \to r \in \rI$ we have $[\ell](\vx) > [r](\vx)$ for all $\vx \in \N^d$, and
  \item for every rule $\ell' \to r' \in \rZ$ we have $[\ell'](\vx) \gtrsim [r'](\vx)$ for all $\vx \in \N^d$.
  \end{itemize}
\end{theorem}
\begin{proof}
  Assume for a contradiction that such a collection of natural matrix interpretations exists for some dimension $d-1$, with each symbol $\sigma \in \Sigma$ interpreted as $[\sigma](\vx) = \mM_\sigma \vx + \vv_\sigma$ and each matrix $\mD_\sigma$ defined as in \cref{thm:affine-linear}. Let $\ell \to r \in \rI$ be a rule for which $[\ell](\vx) > [r](\vx)$ for all $\vx \in \N^{d-1}$, and denote $\rho = \{\ell \to r\}$. As the interpretations satisfy the requirements for only weak monotonicity, for all $s,t \in \Sigma^*$, if $s \to_{\topt{\rho}} t$, then $[s](\vx) > [t](\vx) \text{ for all } \vx \in \N^{d-1}$. Thus, for all $s,t \in \Sigma^*$,
  \begin{equation}\label{eq:decrease-rho-top}
    s \to_{\topt{\rho}} t \quad \implies \quad \mD_s \gtrsim \mD_t \text{ and } (D_s)_{1,d} > (D_t)_{1,d}.
  \end{equation}

  We represent an arbitrary integer $k$ by the string $\tup{\spc} \half \un^k \spc$, interpreted as
  \begin{equation*}
    \big[\tup{\spc} \half \un^k \spc\big]^+(\vx) = \mD_{\tup{\spc}\! \half} \mD_{\un}^k \mD_{\spc} \vx.
  \end{equation*}

  Similar to~\cref{eq:N-rational-interpretation}, we let $\ve_i \in \N^d$ denote the $i$th standard basis vector and define $g(k) \coloneqq \vu^\tr \mD_{\un}^k \vw$, where $\vu = \mD_{\tup{\spc}\! \half}^\tr \ve_1$ and $\vw = \mD_{\spc} \ve_d$ are fixed, so $g(1), g(2), \dots$ is an $\N$-rational sequence.

  For all $n$, we have
  \begin{equation*}
    \tup{\spc} \half \un^{8n+1} \spc \; \to_{\rI \cup \rZ}^* \; \tup{\spc} \half \un^{9n+2} \spc,
  \end{equation*}
  which crucially requires the use of both rules in $\rI$, and, moreover, these rules are applied only at the leftmost end of a string. In particular, the rule $\ell \to r$ is used in the above derivation to perform a top rewrite. Thus, by~\cref{eq:decrease-rho-top}, for all $n$, we have $g(8n+1) > g(9n+2)$, which contradicts \cref{cor:N-rational-collatz}.
\end{proof}

Analogues of \cref{thm:Z-no-direct-natural-matrix,thm:Z-no-dp-natural-matrix} also hold for the system $\rev{\rZ}$, with its rules reversed.\footnote{Specifically, \cref{thm:Z-no-dp-natural-matrix} holds for $\rev{\rZ}$ with $\rJ = \{\tup{\spc} \half \un \un \to \tup{\spc} \shift \un \un,\ \tup{\spc} \un \half \to \tup{\spc} \un \un \triple\}$ in place of $\rI$. The corresponding version of \cref{lem:Z-dependency-pairs} is again straightforward to establish using an automated termination prover.} Moreover, by arguments similar to above, we can show that natural matrix interpretations fail to prove the termination of $\rW$. It is natural to expect that if a proof of the Collatz conjecture is to be produced by some automated method that relies on rewriting, then that method better be able to prove a statement as simple as the convergence of $W$. With this in mind, we describe an alternative rewriting system that simulates the Collatz function and terminates upon reaching $1$. We then provide examples where the alternative system is more suitable for use with termination tools (for instance allowing a matrix interpretations proof of the convergence of $W$). The arguments in this section that prove \cref{thm:Z-no-direct-natural-matrix,thm:Z-no-dp-natural-matrix} also appear not to apply to the alternative system, because the connection between rewrite sequences and $\N$-rational sequences gets lost when not using unary representations. More specifically, the proofs crucially use the fact that, since the system $\rZ$ represents numbers in unary, it is possible to form a sequence such that the $n$th element is represented by some product $\vv^\tr \mM^n \vw$, where $\vv$ and $\vw$ are vectors and $\mM$ is a matrix. There is no single such matrix that one can extract when using mixed binary--ternary, or even just binary, representations. Thus, the same kind of analysis, which uses the results known about $\N$-rational sequences, does not go through for the alternative system.

\subsection{Rewriting in Mixed Base}
\label{sec:rewriting-mixed-base}

In the mixed-base scheme, the overall idea is as follows. Given a number $n \in \Np$, we write a mixed binary--ternary representation for it (noting that this representation is not unique). With this representation, as long as the least significant digit is binary, the parity of the number can be recognized by checking only this digit, as opposed to scanning the entire string when working in unary. This allows us to easily determine the correct case when applying the Collatz function. If the least significant digit is ternary, then the representation is rewritten (while preserving its value) to make this digit binary. Afterwards, since computing $n/2$ corresponds to erasing a trailing binary $0$ and computing $3n + 1$ corresponds to inserting a trailing ternary $1$, applying the Collatz function takes a single rewrite step.

Intuitively, the rewriting system we present can be seen as applying the Collatz map to natural numbers written in binary, with some secondary computation performed using auxiliary symbols to facilitate the application of the map $n \mapsto 3n + 1$. However, since termination of rewriting is concerned with sequences that arise from \emph{any initial string} and \emph{any possible application of the rules}, we need to consider the behavior of the system over any string that may be encountered, which includes even the strings that could not result from computing in binary with a prescribed strategy for rule application. The mixed-base scheme is simply a way to give an arithmetical meaning to all strings and all possible rewrite steps.

More formally, a mixed-base numeral system~\cite{Can69} is a numeral system where the base changes across positions, which we define as follows.
\begin{definition}[Mixed-base representation]\label{defn:mixed-base}
  Fix $B \subseteq \N$ to be a set of bases. Let $n_1, \dots, n_k \in \N$ and let $b_1, \dots, b_k \in B$. If we have for each $2 \leq i \leq k$ that $n_i < b_i$, then the matrix
  \begin{equation*}
    N = \begin{bmatrix} n_1 & n_2 & \dots & n_k \\ b_1 & b_2 & \dots & b_k \end{bmatrix}
  \end{equation*}
  is called a \emph{mixed $B$-ary representation}. To reduce clutter, we identify the matrix $N$ with a string $(n_1)_{(b_1)} (n_2)_{(b_2)} \dots (n_k)_{(b_k)}$, where each $(n_i)_{(b_i)}$ is viewed as a single symbol denoting the \emph{$b_i$-ary~digit~$n_i$}.
\end{definition}
The string $N$ in the above definition represents the number
\begin{equation}\label{eq:value-string}
  \Val(N) \coloneqq \sum_{i=1}^k n_i \prod_{j=i+1}^{k} b_j.
\end{equation}
Observing that the addition of leading zeros (symbols $0_b$ for any $b$) to a string does not change its value, we may assume without loss of generality that $n_1 > 0$. Furthermore, $b_1$ does not affect the value of the string, so we replace it by zero (noting that by \cref{defn:mixed-base} no other base is zero).

\begin{example}\label{ex:mixed-base}
  \hfill
  \begin{enumerate}
  \item With $B = \{2\}$, we have the binary numeral system. In the notation of \cref{defn:mixed-base} we write $1_2 0_2 1_2 1_2$ to represent the number $11$.
  \item With $B = \{2,3\}$, we obtain the mixed binary--ternary system that we will use in the rest of this paper. In this system we may write, for instance, $1_0 2_3 1_2$ or $2_0 1_2 1_2$ or $1_0 1_2 2_3$ to represent the number $11$.
  \end{enumerate}
\end{example}

Now, define $\beta_b^{n}(x) \coloneqq bx + n$. After rearranging~\cref{eq:value-string}, we see that for a mixed $B$-ary string $N = (n_1)_{(b_1)} (n_2)_{(b_2)} \dots (n_k)_{(b_k)}$, its value $\Val(N)$ is also given by evaluating the composite function
\begin{equation}\label{eq:comp-string}
  \Comp_N(x) \coloneqq (\beta_{b_k}^{n_k} \circ \beta_{b_{k-1}}^{n_{k-1}} \circ \dots \circ \beta_{b_1}^{n_1})(x)
\end{equation}
at any value (because $b_1 = 0$ implies that the innermost function $\beta_0^{n_1}(x) = n_1$ is constant). This gives us a string and a function view of the same representation, and we will switch between them as appropriate. In doing so, we also conflate the symbols and the corresponding functions, referring to $\beta_b^n$ as $n_b$.

As the last ingredient before describing the rewriting system, we observe that we can write $(\beta_b^{n} \circ \beta_c^{m})(x) = bcx+bm+n$ equivalently as another composition $(\beta_c^{m'} \circ \beta_b^{n'})(x) = cbx+cn'+m'$ for some suitable $0 \leq n' < b$ and $0 \leq m' < c$. This allows us to swap the bases of adjacent positions while preserving the value of the string.

From this point on, we constrain ourselves to the mixed $\{2,3\}$-ary (binary--ternary) representations as we shift our focus to simulating the Collatz function (noting that it is possible to adapt the resulting rewriting system for other instances of the general case). More precisely, we simulate the following redefinition of the Collatz function, where the odd case incorporates an additional division by $2$.
\begin{equation*}\label{eq:shortcut-collatz-fn}
  T(n) =
  \begin{cases}
    \frac{n}{2} & \text{if } n \equiv 0 \pmod 2 \\
    \frac{3n + 1}{2} & \text{if } n \equiv 1 \pmod 2
  \end{cases}
\end{equation*}

We will describe an SRS $\rT$ over the symbols $\{\bZ, \bO, \tZ, \tO, \tT, \btL, \btR\}$ that simulates the iterated application of the Collatz function and terminates upon reaching $1$. The symbols $\bZ, \bO$ correspond to binary digits $0_2, 1_2$; and $\tZ, \tO, \tT$ to ternary digits $0_3, 1_3, 2_3$. The symbol $\btL$ marks the beginning of a string while also standing for the most significant digit (without loss of generality assumed to be $1_0$) and $\btR$ marks the end of a string while also standing for the redundant trailing digit $0_1$. Consider the functional view of these symbols:
\begin{equation}\label{eq:interp}
  \begin{array}{lcl}
    \bZ(x) & = & 2x \\
    \bO(x) & = & 2x + 1
  \end{array}
  \qquad
  \begin{array}{lcl}
    \tZ(x) & = & 3x \\
    \tO(x) & = & 3x + 1 \\
    \tT(x) & = & 3x + 2
  \end{array}
  \qquad
  \begin{array}{lcl}
    \btL(x) & = & 1 \\
    \btR(x) & = & x
  \end{array}
\end{equation}
Each positive integer can be expressed as some composition of these functions, which corresponds to a string as per our previous discussion.

\begin{example}\label{ex:sample-rewrite-T}
  Viewing the expression $\btL(x)$ as the constant $1$, we can write
  \begin{equation*}
    19 = \Val(\btL \tZ \bZ \tO \btR) = \btR(\tO(\bZ(\tZ(\btL(x))))).
  \end{equation*}
  The string representation ends with a ternary symbol, so we will rewrite it.

  With the function view, we have
  \begin{equation*}
    \tO(\bZ(x)) = 3(2x)+1 = 6x+1 = 2(3x)+1 = \bO(\tZ(x)).
  \end{equation*}
  This shows that we could also write $19 = \Val(\btL \tZ \tZ \bO \btR)$, which now ends with the binary digit $1_2$. This gives us the rewrite rule $\bZ \tO \to \tZ \bO$. We can now apply the Collatz function to this representation by rewriting only the rightmost two symbols of the string since
  \begin{equation*}
    T(\btR(\bO(x))) = \frac{3(2x+1)+1}{2} = \frac{6x+4}{2} = 3x+2 = (\btR(\tT(x))).
  \end{equation*}
  This gives us the rewrite rule $\bO \btR \to \tT \btR$. After applying this rule to the string $\btL \tZ \tZ \bO \btR$, we indeed obtain $T(19) = 29 = \Val(\btL \tZ \tZ \tT \btR)$.
\end{example}

In the manner of the above example, we compute all the necessary transformations and obtain the following 11-rule SRS $\rT$.
\begin{equation*}
  \rD_T =
  \left\{
    \begin{array}[c]{rcl}
      \bZ \btR & \to & \btR \\
      \bO \btR & \to & \tT \btR
    \end{array}
  \right\}
  \qquad
  \rA =
  \left\{
    \begin{array}[c]{rcl}
      \bZ \tZ & \to & \tZ \bZ \\
      \bZ \tO & \to & \tZ \bO \\
      \bZ \tT & \to & \tO \bZ
    \end{array}
    \qquad
    \begin{array}[c]{rcl}
      \bO \tZ & \to & \tO \bO \\
      \bO \tO & \to & \tT \bZ \\
      \bO \tT & \to & \tT \bO
    \end{array}
  \right\}
  \qquad
  \rB =
  \left\{
    \begin{array}[c]{rcl}
      \btL \tZ & \to & \btL \bO \\
      \btL \tO & \to & \btL \bZ \bZ \\
      \btL \tT & \to & \btL \bZ \bO
    \end{array}
  \right\}
\end{equation*}
This SRS is split into subsystems $\rD_T$ (dynamic rules for $T$) and $\rX = \rA \cup \rB$ (auxiliary rules). The two rules in $\rD_T$ encode the application of the Collatz function $T$, while the rules in $\rX$ serve to push binary symbols towards the rightmost end of the string by swapping the bases of adjacent positions without changing the represented value.

\begin{example}[Rewrite sequence of $\rT$]\label{ex:rewrite-sequence-T}
  Consider the string $s = \btL \bZ \bZ \tZ \btR$, which represents the number $12$. Below is a possible rewrite sequence of $\rT$ that starts from $s$, with the corresponding values (under the interpretations from~\cref{eq:interp}) displayed above the strings. Underlines indicate the parts of the strings where the rules are applied.
  \begin{equation*}
    \begin{array}[c]{*{16}{@{\,}c@{\,}}}
      & 12 && 12 && 6 && 6 && 3 && 3 && 5 && 5 \\[0.25em]
      & \btL \bZ \underline{\bZ \tZ} \btR & \to_\rA & \btL \bZ \tZ \underline{\bZ \btR} & \to_{\rD_T} & \btL \underline{\bZ \tZ} \btR & \to_\rA & \btL \tZ \underline{\bZ \btR} & \to_{\rD_T} & \underline{\btL \tZ} \btR & \to_\rB & \btL \underline{\bO \btR} & \to_{\rD_T} & \underline{\btL \tT} \btR & \to_\rB & \btL \bZ \underline{\bO \btR} \\[1em]
      & 8 && 8 && 8 && 4 && 2 && 1 \\[0.25em]
      \to_{\rD_T} & \btL \underline{\bZ \tT} \btR & \to_\rA & \underline{\btL \tO} \bZ \btR & \to_\rB & \btL \bZ \bZ \underline{\bZ \btR} & \to_{\rD_T} & \btL \bZ \underline{\bZ \btR} & \to_{\rD_T} & \btL \underline{\bZ \btR} & \to_{\rD_T} & \btL \btR
    \end{array}
  \end{equation*}
\end{example}
The trajectory of $T$ would continue upon reaching $1$; however, in order to be able to formulate the Collatz conjecture as a termination problem, $\rT$ is made in such a way that its rewrite sequences stop upon reaching the string representation $\btL \btR$ of $1$ since no rule is applicable.

Termination of the subsystems of $\rT$ with $\rB$ or $\rD_T$ removed is easily seen. However, since we have matrix interpretations at our disposal, let us give a compact and formal proof.
\begin{lemma}\label{lem:A-termination}
  $\SN(\rT \setminus \rB)$ and $\SN(\rT \setminus \rD_T)$.
\end{lemma}
\begin{proof}
  It is easily checked that the interpretations below show $\SN(\rev{(\rT \setminus \rB)})$, which implies $\SN(\rT \setminus \rB)$ by \cref{lem:reversal}.
  \begin{equation*}
    [\bZ](x) = [\bO](x) = 2x + 1 \qquad [\btR] = x \qquad [\tZ](x) = [\tO](x) = [\tT](x) = 2 x
  \end{equation*}

  Similarly, the below interpretations show $\SN(\rev{(\rT \setminus \rD_T)})$, which implies $\SN(\rT \setminus \rD_T)$ by \cref{lem:reversal}.
  \begin{equation*}
    [\bZ](x) = [\bO](x) = [\btL](x) = x + 1 \qquad
    [\tZ](x) = [\tO](x) = [\tT](x) = 4 x \qedhere
  \end{equation*}
\end{proof}

When considering the termination of $\rT$ (on all initial strings), it suffices to limit the discussion to initial strings of a specific form that we have been working with so far, e.g., in \cref{ex:sample-rewrite-T,ex:rewrite-sequence-T}.
\begin{restatable}{lemma}{blockwise}\label{lem:blockwise-termination}
  If $\rT$ is terminating on all initial strings of the \emph{canonical form} $\btL (\bZ \vert \bO \vert \tZ \vert \tO \vert \tT)^* \btR$, then $\rT$ is terminating (on all initial strings).
\end{restatable}
For the proof of \cref{lem:blockwise-termination}, we will use \emph{type introduction}, which is a technique for proving termination by switching to typed rewriting. We refer the reader to Sabel~and~Zantema~\cite[Section~3.1]{SZ17} for background. We also offer a more elementary proof in \cref{sec:alternative-proof-blockwise-termination}.
\begin{proof}[Proof of \cref{lem:blockwise-termination}]
  Let $\mathbb{T} = \{\rho, \sigma, \tau\}$ be a set of types. It is straightforward to see that, with respect to the typing
  \begin{equation*}
    \btL : \sigma \to \tau, \qquad \bZ, \bO, \tZ, \tO, \tT : \sigma \to \sigma, \qquad \btR : \rho \to \sigma,
  \end{equation*}
  the SRS $\rT$ is well-typed, so it is terminating (in the untyped setting) if and only if it is terminating in the typed setting~\cite[Corollary~3.4]{SZ17}. Due to \cref{lem:A-termination}, a well-typed string with $\sigma$ as its source or target type cannot admit an infinite rewrite sequence for $\rT$. Thus, a well-typed string admits an infinite rewrite sequence only if it is of type $\rho \to \tau$, which implies that it is of the canonical form.
\end{proof}
Note that the converse of the above lemma is obvious, since if $\rT$ is terminating then it is terminating on all initial strings of any form, in particular $\btL (\bZ \vert \bO \vert \tZ \vert \tO \vert \tT)^* \btR$.

As a whole, the rewriting system $\rT$ simulates the iterated application of $T$ (except at $1$).
\begin{theorem}\label{thm:simulation}
  $\rT$ is terminating if and only if $T$ is convergent (i.e., the Collatz conjecture holds).
\end{theorem}
\begin{proof}
  Without loss of generality (due to \cref{lem:blockwise-termination}), consider only the strings of the canonical form $\btL (\bZ \vert \bO \vert \tZ \vert \tO \vert \tT)^* \btR$ throughout this proof. We will show that, given such a string $N$, all rewrite sequences for $\rT$ that start from $N$ simulate the $T$-trajectory of $\Val(N)$, except when $\Val(N) = 1$. More specifically, we prove the following and use them to deduce the theorem statement.
  \begin{claim}\label{claim:equiv-rewrite-num}
    For any pair $N, N'$ of strings, if $N \to_{\rX} N'$ then $\Val(N') = \Val(N)$, and if $N \to_{\rD_T} N'$ then $\Val(N') = T(\Val(N))$.
  \end{claim}
  \begin{claim}\label{claim:equiv-1}
    For any string $N$ with $\Val(N) = 1$, there exists no string $N'$ satisfying $N \to_\rT N'$.
  \end{claim}
  \begin{claim}\label{claim:equiv-num-rewrite}
    For any integer $\nu > 1$ and any string $N$ with $\Val(N) = \nu$, there exists a string $N'$ such that $N \to_\rT^* N'$ and $\Val(N') = T(\nu)$.
  \end{claim}
  \begin{proof}[Proof of \cref{claim:equiv-rewrite-num}]
    Assume $N \to_\rX N'$. Then there exist strings $X,Y$ and some rule $\ell \to r \in \rX$ such that $N = X \ell Y$ and $N' = X r Y$. Recalling the definition of $\Comp$ from~\cref{eq:comp-string}, we have $\Val(N') = \Val(N)$ if and only if $\Comp_Y \circ \Comp_r \circ \Comp_X = \Comp_Y \circ \Comp_\ell \circ \Comp_X$ for all $\ell \to r \in \rX$. This holds if and only if $\Comp_r = \Comp_\ell$ for all $\ell \to r \in \rX$. With the functional views of the symbols as in~\cref{eq:interp}, we indeed have:
    \begin{equation*}
      \begin{array}{lclcccrcr}
        \Comp_{\tZ\! \bZ}(x) & = & \bZ(\tZ(x)) & = & 6x & = & \tZ(\bZ(x)) & = & \Comp_{\bZ\! \tZ}(x) \\[0.25em]
        \Comp_{\tZ\! \bO}(x) & = & \bO(\tZ(x)) & = & 6x+1 & = & \tO(\bZ(x)) & = & \Comp_{\bZ\! \tO}(x) \\[0.25em]
        \Comp_{\tO\! \bZ}(x) & = & \bZ(\tO(x)) & = & 6x+2 & = & \tT(\bZ(x)) & = & \Comp_{\bZ\! \tT}(x) \\[0.25em]
        \Comp_{\tO\! \bO}(x) & = & \bO(\tO(x)) & = & 6x+3 & = & \tZ(\bO(x)) & = & \Comp_{\bO\! \tZ}(x) \\[0.25em]
        \Comp_{\tT\! \bZ}(x) & = & \bZ(\tT(x)) & = & 6x+4 & = & \tO(\bO(x)) & = & \Comp_{\bO\! \tO}(x) \\[0.25em]
        \Comp_{\tT\! \bO}(x) & = & \bO(\tT(x)) & = & 6x+5 & = & \tT(\bO(x)) & = & \Comp_{\bO\! \tT}(x) \\[0.25em]
        \Comp_{\btL\! \bO}(x) & = & \bO(\btL(x)) & = & 3 & = & \tZ(\btL(x)) & = & \Comp_{\btL\! \tZ}(x) \\[0.25em]
        \Comp_{\btL\! \bZ\! \bZ}(x) & = & \bZ(\bZ(\btL(x))) & = & 4 & = & \tO(\btL(x)) & = & \Comp_{\btL\! \tO}(x) \\[0.25em]
        \Comp_{\btL\! \bZ\! \bO}(x) & = & \bO(\bZ(\btL(x))) & = & 5 & = & \tT(\btL(x)) & = & \Comp_{\btL\! \tT}(x)
      \end{array}
    \end{equation*}
    This proves the first part of the claim.

    Now assume $N \to_{\rD_T} N'$. Again, there exist a string $X$ and some rule $\ell \to r \in \rD_T$ such that $N = X \ell$ and $N' = X r$ (since the rules in $\rD_T$ can be applied only at the rightmost end of the string). We have $\Val(N') = T(\Val(N))$ if and only if $\Comp_r \circ \Comp_X = T \circ \Comp_\ell \circ \Comp_X$ for all $\ell \to r \in \rD_T$. This holds if and only if $\Comp_r = T \circ \Comp_\ell$ for all $\ell \to r \in \rD_T$. There are only two rules to check, and indeed we have:
    \begin{equation*}
      \begin{array}{lclcccccrcr}
        \Comp_{\btR}(x) & = & \btR(x) & = & x & = & T(2x) & = & T(\btR(\bZ(x))) & = & T(\Comp_{\bZ\! \btR}(x)) \\[0.25em]
        \Comp_{\tT\! \btR}(x) & = & \btR(\tT(x)) & = & 3x+2 & = & T(2x+1) & = & T(\btR(\bO(x))) & = & T(\Comp_{\bO\! \btR}(x))
      \end{array}
    \end{equation*}
    This proves the second part of the claim.
  \end{proof}
  \begin{proof}[Proof of \cref{claim:equiv-1}]
    The string $\btL \btR$ has $\Val(\btL \btR) = 1$, and inserting any symbol from $\{\bZ, \bO, \tZ, \tO, \tT\}$ in between the delimiters increases the value of the string, so $\btL \btR$ is the unique string representing $1$. None of the rules in $\rT$ apply to $\btL \btR$, which proves the claim.
  \end{proof}
  \begin{proof}[Proof of \cref{claim:equiv-num-rewrite}]
    Let $\nu > 1$ be an integer, and let $N$ be a string with $\Val(N) = \nu$. We will describe a rewrite sequence $N \to_{\rX}^* N' \to_{\rD_T} N^{\prime\prime}$, which implies by \cref{claim:equiv-rewrite-num} that $\Val(N^{\prime\prime}) = T(\nu)$. For the rest of the proof let $b_i$ and $t_i$ stand for binary and ternary symbols, respectively. Assume that $N$ contains some binary symbol ($\bZ$ or $\bO$), since otherwise it is of the form $\btL (\tZ \vert \tO \vert \tT)^* \btR$ and we can apply some rule from $\rB$ to produce a binary symbol. Then the string contains some substring $b_1 t_2 t_3 \dots t_k \btR$. We rewrite this substring into $t_1' b_2 t_3 \dots t_k \btR$ by applying some rule from $\rX$. For each $1 < i < k$, assuming the current string contains some substring $b_i t_{i+1} \dots t_k \btR$, we rewrite it into $t_i' b_{i+1} t_{i+2} \dots t_k \btR$. In the end we obtain a string $N'$ with $\Val(N') = \nu$ that contains some $b_k \btR$, so we apply some rule from $\rD_T$ to obtain $N^{\prime\prime}$ such that $\Val(N^{\prime\prime}) = T(\nu)$.
  \end{proof}

  Assume that the SRS $\rT$ is not terminating. Let $(N_i)_{i=0}^\infty$ be an infinite rewrite sequence for $\rT$, and consider the sequence $(\nu_i)_{i=0}^\infty$, where $\nu_i = \Val(N_i)$. By \cref{claim:equiv-rewrite-num}, for all $i \in \N$, we have either $\nu_{i+1} = \nu_i$ or $\nu_{i+1} = T(\nu_i)$ depending on the rewrite rule used. Let $I = \{i \in \N \mid \nu_{i+1} = T(\nu_i)\}$, which is the same as the set of indices at which some rule from $\rD_T$ is applied in the rewrite sequence. By \cref{lem:A-termination}, the subsystem $\rX = \rT \setminus \rD_T$ is terminating, so the rewrite $\rew_{\rD_T}$ is performed infinitely many times in the sequence. This implies that $I$ is infinite, proving that the sequence $(\nu_i)_{i \in I}$ is a Collatz trajectory. Moreover, this trajectory is nonconvergent, i.e., it does not contain $1$, since otherwise the rewrite sequence could not be infinite (because of \cref{claim:equiv-1}). This proves the backward implication in the theorem statement.

  Assume that the function $T$ is not convergent. Let $(\nu_i)_{i=0}^\infty$ be a nonconvergent trajectory, and note that $\nu_i > 1$ for all $i \in \N$. We will inductively define a sequence $(N_i)_{i=0}^\infty$ such that $\Val(N_i) = \nu_i$ for all $i \in \N$. Let $N_0$ be any string with $\Val(N_0) = \nu_0$. For each $i \in \Np$, assuming $\Val(N_{i-1}) = \nu_{i-1}$, let $N_i$ be a string such that $N_{i-1} \to_{\rT}^* N_i$ and $\Val(N_i) = T(\nu_{i-1}) = \nu_i$, which exists by \cref{claim:equiv-num-rewrite}. Defined in this manner, the sequence $(N_i)_{i=0}^\infty$ satisfies $N_i \to_\rT^* N_{i+1}$ for all $i \in \N$, meaning that $\rT$ is not terminating. This proves the forward implication in the theorem statement, which concludes the proof.
\end{proof}

When trying to remove a rule in $\rD_T$ or $\rB$ it suffices to show relative top termination, allowing us to use weakly (instead of extended) monotone algebras when applying \cref{thm:monotone-algebra} and take advantage of the more relaxed constraints when searching for matrix interpretations. As we mentioned back in \cref{sec:string-rewriting-systems}, the lemma below encapsulates dependency pairs, and it can in fact be automatically proved by the existing termination tools via the dependency pair transformation followed by an application of the dependency pair framework~\cite{GTSF06}.\footnote{We thank an anonymous reviewer of CADE for pointing out this fact.}
\begin{lemma}\label{lem:T-top-termination}
  \hfill
  \begin{enumerate}
  \item For each subset $\rR \subseteq \rB$, if $\SN(\rR_\mathrm{top} \rel \rT)$ then $\SN(\rR \rel \rT)$.
  \item For each subset $\rQ \subseteq \rD_T$, if $\SN(\rev{\rQ}_\mathrm{top} \rel \rev{\rT})$ then $\SN(\rev{\rQ} \rel \rev{\rT})$.
  \end{enumerate}
\end{lemma}
\begin{proof}
  Without loss of generality (due to \cref{lem:blockwise-termination}), consider only the rewrite sequences that start with some string of the canonical form $\btL (\bZ \vert \bO \vert \tZ \vert \tO \vert \tT)^* \btR$ (resp.\ its reversal).

  Let $\rR \subseteq \rB$ (resp.\ $\rQ \subseteq \rD_T$). The rules in $\rR$ (resp.\ $\rev{\rQ}$) can be applied only at the leftmost end as they are all of the form $\btL \mathmakebox[.57em]{s} \to \btL \mathmakebox[.57em]{t}$ for some $s$, $t$ not containing $\btL$ (resp.\ $\btR \mathmakebox[.67em]{s'} \to \btR \mathmakebox[.67em]{t'}$ for some $s'$, $t'$ not containing $\btR$).

  Assume $\neg \SN(\rR \rel \rT)$ (resp.\ $\neg \SN(\rev{\rQ} \rel \rev{\rT})$). Then there exists an infinite rewrite sequence for $\rT$ (resp. $\rev{\rT}$) where the rules from $\rR$ (resp.\ $\rev{\rQ}$) are applied infinitely many times. As these rules can be applied only at the leftmost end, this implies relative top nontermination, i.e., $\neg \SN(\rR_\mathrm{top} \rel \rT)$ (resp.\ $\neg \SN(\rev{\rQ}_\mathrm{top} \rel \rev{\rT})$).
\end{proof}

\subsection{Rewriting with a Hybrid System}
\label{sec:rewriting-hybrid}

As another alternative to the rewriting system $\rZ$, instead of fully discarding unary in favor of a positional numeral system we may adopt a hybrid between the two approaches and design a system to increment/decrement by $1$ or multiply/divide by $2$ or $3$ in a single rewrite. We will describe an SRS over the symbols $\{\inc, \bZ, \tZ, \btL, \btR\}$ viewed as follows.
\begin{equation*}
  \begin{array}{lcl}
    \inc(x) & = & x + 1 \\
    \bZ(x) & = & 2x \\
    \tZ(x) & = & 3x \\
    \btL(x) & = & 1 \\
    \btR(x) & = & x
  \end{array}
\end{equation*}
In essentially the same manner as we did for the system $\rT$, we have a set of dynamic rules that simulate the application of the Collatz function and a set of auxiliary rules that update the string representation so that the parity of the corresponding number can be recognized in a single step. We have the following 7-rule SRS $\rL$ (designed by Luke Schaeffer).
\begin{equation*}
  \rD_T' =
  \left\{
    \begin{array}[c]{rcl}
      \bZ \btR & \to & \btR \\
      \bZ \inc \btR & \to & \tZ \inc \inc \btR
    \end{array}
  \right\}
  \quad
  \rA' =
  \left\{
    \begin{array}[c]{rcl}
      \bZ \inc \inc & \to & \inc \bZ \\
      \inc \tZ & \to & \tZ \inc \inc \inc \\
      \bZ \tZ & \to & \tZ \bZ
    \end{array}
  \right\}
  \quad
  \rB' =
  \left\{
    \begin{array}[c]{rcl}
      \btL \inc & \to & \btL \bZ \\
      \btL \tZ & \to & \btL \bZ \inc \\
    \end{array}
  \right\}
\end{equation*}
As long as the string ends with $\bZ \btR$ or $\bZ \inc \btR$, the parity of the represented number is apparent without needing to scan the whole string, which allows the Collatz function to be applied in a single rewrite. After the application of a dynamic rule, the string may no longer end with $\bZ \btR$ or $\bZ \inc \btR$, so the auxiliary rules in $\rX' = \rA' \cup \rB'$ push the symbol $\bZ$ towards the rightmost end of the string while preserving its value. Thus, $\rL$ is yet another rewriting system that simulates the iterated application of the Collatz function.
\begin{theorem}
  $\rL$ is terminating if and only if $T$ is convergent.
\end{theorem}
The proof of the above theorem is essentially the same as that of \cref{thm:simulation} once it is checked that the dynamic rules do indeed apply $T$ and that the auxiliary rules do indeed preserve the value of a string under the functional view of the symbols.

\section{Automated Proofs}
\label{sec:results}

We adapt the rewriting system $\rT$ for different generalized Collatz functions to explore the effectiveness of the mixed-base scheme on weakened variants of the Collatz conjecture.

All of the rewriting systems, scripts to reproduce the proofs, and our minimal implementation of a termination prover are available at \url{https://github.com/emreyolcu/rewriting-collatz}. In particular, any omitted relative termination proof that we refer to in this section is available at \url{https://github.com/emreyolcu/rewriting-collatz/tree/main/proofs}. We ran the large-scale experiments in this section on AWS EC2 \textsf{m5.metal} instances.

\subsection{A Simple Example}
\label{sec:convergence-w}

Earlier we mentioned a generalized Collatz function $W$ as a simple example that could serve as a sanity check for an automated method aiming to solve Collatz-like problems. With the mixed binary--ternary scheme, this function can be seen to be simulated by the system $\rW' = \{\bZ \btR \to \tZ \btR\} \cup \rX$. A small matrix interpretations proof is found for this system in less than a second, in contrast to its variant $\rW$ that uses unary representations for which no automated proof is known. Additionally, the function $W$ can be simulated by the rewriting system $\rW^{\prime\prime} = \{\bZ \btR \to \tZ \btR\} \cup \rX'$ adapted from the hybrid system $\rL$, although we were unable to find an automated proof of termination for $\rW^{\prime\prime}$.

\begin{theorem}
  $\SN(\rW')$.
\end{theorem}
\begin{proof}
  The interpretations below prove $\SN(\{\btR \bZ \to \btR \tZ\} \rel \rev{\rX})$:
  {\setlength{\jot}{1em}
    \begin{gather*}
      [\bZ](\vx) =
      \begin{bmatrix}
        1 & 0 \\
        0 & 1
      \end{bmatrix}
      \vx +
      \begin{bmatrix}
        1 \\
        1
      \end{bmatrix}
      \qquad
      [\bO](\vx) =
      \begin{bmatrix}
        1 & 0 \\
        0 & 0
      \end{bmatrix}
      \vx +
      \begin{bmatrix}
        1 \\
        0
      \end{bmatrix} \\
      [\btL](\vx) =
      \begin{bmatrix}
        1 & 0 \\
        0 & 0
      \end{bmatrix}
      \vx
      \qquad
      [\btR](\vx) =
      \begin{bmatrix}
        1 & 2 \\
        0 & 0
      \end{bmatrix}
      \vx
      \\
      [\tZ](\vx) =
      \begin{bmatrix}
        1 & 0 \\
        0 & 1
      \end{bmatrix}
      \vx +
      \begin{bmatrix}
        2 \\
        0
      \end{bmatrix}
      \qquad
      [\tO](\vx) =
      \begin{bmatrix}
        1 & 0 \\
        1 & 0
      \end{bmatrix}
      \vx +
      \begin{bmatrix}
        2 \\
        2
      \end{bmatrix}
      \qquad
      [\tT](\vx) =
      \begin{bmatrix}
        1 & 0 \\
        1 & 0
      \end{bmatrix}
      \vx +
      \begin{bmatrix}
        2 \\
        2
      \end{bmatrix}
    \end{gather*}
  }%
  With the above interpretations of the symbols, the rules of $\rev{\rW'}$ satisfy the following relations, which, by \cref{thm:monotone-algebra}, proves that $\{\btR \bZ \to \btR \tZ\}$ is terminating relative to $\rev{\rX}$.
  \begin{longtable*}{IJK}
    {[}\btR \bZ](\vx) =
    \begin{bmatrix}
      1 & 2 \\
      0 & 0
    \end{bmatrix}
    \vx +
    \begin{bmatrix}
      3 \\
      0
    \end{bmatrix}
    & > &
    \begin{bmatrix}
      1 & 2 \\
      0 & 0
    \end{bmatrix}
    \vx +
    \begin{bmatrix}
      2 \\
      0
    \end{bmatrix}
    = {[}\btR \tZ](\vx)
    \\[1.825em]
    {[}\tZ \bZ](\vx) =
    \begin{bmatrix}
      1 & 0 \\
      0 & 1
    \end{bmatrix}
    \vx +
    \begin{bmatrix}
      3 \\
      1
    \end{bmatrix}
    & \gtrsim &
    \begin{bmatrix}
      1 & 0 \\
      0 & 1
    \end{bmatrix}
    \vx +
    \begin{bmatrix}
      3 \\
      1
    \end{bmatrix}
    = {[}\bZ \tZ](\vx)
    \\[1.825em]
    {[}\tO \bZ](\vx) =
    \begin{bmatrix}
      1 & 0 \\
      1 & 0
    \end{bmatrix}
    \vx +
    \begin{bmatrix}
      3 \\
      3
    \end{bmatrix}
    & \gtrsim &
    \begin{bmatrix}
      1 & 0 \\
      0 & 0
    \end{bmatrix}
    \vx +
    \begin{bmatrix}
      3 \\
      0
    \end{bmatrix}
    = {[}\bO \tZ](\vx)
    \\[1.825em]
    {[}\tT \bZ](\vx) =
    \begin{bmatrix}
      1 & 0 \\
      1 & 0
    \end{bmatrix}
    \vx +
    \begin{bmatrix}
      3 \\
      3
    \end{bmatrix}
    & \gtrsim &
    \begin{bmatrix}
      1 & 0 \\
      1 & 0
    \end{bmatrix}
    \vx +
    \begin{bmatrix}
      3 \\
      3
    \end{bmatrix}
    = {[}\bZ \tO](\vx)
    \\[1.825em]
    {[}\tZ \bO](\vx) =
    \begin{bmatrix}
      1 & 0 \\
      0 & 0
    \end{bmatrix}
    \vx +
    \begin{bmatrix}
      3 \\
      0
    \end{bmatrix}
    & \gtrsim &
    \begin{bmatrix}
      1 & 0 \\
      0 & 0
    \end{bmatrix}
    \vx +
    \begin{bmatrix}
      3 \\
      0
    \end{bmatrix}
    = {[}\bO \tO](\vx)
    \\[1.825em]
    {[}\tO \bO](\vx) =
    \begin{bmatrix}
      1 & 0 \\
      1 & 0
    \end{bmatrix}
    \vx +
    \begin{bmatrix}
      3 \\
      3
    \end{bmatrix}
    & \gtrsim &
    \begin{bmatrix}
      1 & 0 \\
      1 & 0
    \end{bmatrix}
    \vx +
    \begin{bmatrix}
      3 \\
      3
    \end{bmatrix}
    = {[}\bZ \tT](\vx)
    \\[1.825em]
    {[}\tT \bO](\vx) =
    \begin{bmatrix}
      1 & 0 \\
      1 & 0
    \end{bmatrix}
    \vx +
    \begin{bmatrix}
      3 \\
      3
    \end{bmatrix}
    & \gtrsim &
    \begin{bmatrix}
      1 & 0 \\
      0 & 0
    \end{bmatrix}
    \vx +
    \begin{bmatrix}
      3 \\
      0
    \end{bmatrix}
    = {[}\bO \tT](\vx)
    \\[1.825em]
    {[}\tZ \btL](\vx) =
    \begin{bmatrix}
      1 & 0 \\
      0 & 0
    \end{bmatrix}
    \vx +
    \begin{bmatrix}
      2 \\
      0
    \end{bmatrix}
    & > &
    \begin{bmatrix}
      1 & 0 \\
      0 & 0
    \end{bmatrix}
    \vx +
    \begin{bmatrix}
      1 \\
      0
    \end{bmatrix}
    = {[}\bO \btL](\vx)
    \\[1.825em]
    {[}\tO \btL](\vx) =
    \begin{bmatrix}
      1 & 0 \\
      1 & 0
    \end{bmatrix}
    \vx +
    \begin{bmatrix}
      2 \\
      2
    \end{bmatrix}
    & \gtrsim &
    \begin{bmatrix}
      1 & 0 \\
      0 & 0
    \end{bmatrix}
    \vx +
    \begin{bmatrix}
      2 \\
      2
    \end{bmatrix}
    = {[}\bZ \bZ \btL](\vx)
    \\[1.825em]
    {[}\tT \btL](\vx) =
    \begin{bmatrix}
      1 & 0 \\
      1 & 0
    \end{bmatrix}
    \vx +
    \begin{bmatrix}
      2 \\
      2
    \end{bmatrix}
    & \gtrsim &
    \begin{bmatrix}
      1 & 0 \\
      0 & 0
    \end{bmatrix}
    \vx +
    \begin{bmatrix}
      2 \\
      0
    \end{bmatrix}
    = {[}\bO \bZ \btL](\vx)
  \end{longtable*}
  \noindent By \cref{lem:A-termination,lem:reversal}, $\rev{\rX}$ is terminating. As a result, $\rev{\rW'}$ is terminating, which by \cref{lem:reversal} implies that $\rW'$ is terminating.
\end{proof}

\subsection{Farkas' Variant}
\label{sec:farkas-variant}

Let $\NO = \{1,3,5,\dots\}$ denote the odd natural numbers. Farkas~\cite{Far05} studied a slight modification $F' \colon \NO \to \NO$ of the Collatz function for which it becomes possible to prove convergence via induction. We consider automatically proving the convergence of this function as another test case for the mixed-base scheme that is easier than the Collatz conjecture without being entirely trivial. Farkas defined this function as
\begin{equation*}
  F'(n) =
  \begin{cases}
    \frac{n}{3} & \text{if } n \equiv 0 \pmod 3 \\
    \frac{n+1}{2} & \text{if } n \not\equiv 0 \pmod 3 \text{ and } n \equiv 1 \pmod 4 \\
    \frac{3n+1}{2} & \text{if } n \not\equiv 0 \pmod 3 \text{ and } n \equiv 3 \pmod 4.
  \end{cases}
\end{equation*}

In this paper, we define another function $F \colon \N \to \N$ that resembles the Collatz function more closely than Farkas' $F'$ (with respect to the definitions of the cases) while being equivalent to $F'$ in terms of convergence. It is obtained by introducing an additional case in the Collatz function for $n \equiv 1 \pmod 3$ and applying $T$ otherwise:
\begin{equation*}
  F(n) =
  \begin{cases}
    \frac{n-1}{3} & \text{if } n \equiv 1 \pmod 3 \\
    \frac{n}{2} & \text{if } n \equiv 0 \text{ or } n \equiv 2 \pmod 6 \\
    \frac{3n+1}{2} & \text{if } n \equiv 3 \text{ or } n \equiv 5 \pmod 6
  \end{cases}
\end{equation*}
Given an $F'$-trajectory, the bijection $n \mapsto (n-1)/2$ applied to each of its iterates maps the trajectory to the corresponding $F$-trajectory. Thus, an $F'$-trajectory reaches $1$ if and only if the corresponding $F$-trajectory reaches $0$. Furthermore, an $F$-trajectory reaches $0$ if and only if it reaches $1$, which is a convenient fact since in our string representations we do not represent $0$, so in order to prove that $F$ is convergent it suffices to show that for all $n \in \Np$ the trajectory $F^{\N}(n)$ contains $1$.

Note that $F$ can also be viewed as a generalized Collatz function with 5 cases, which makes it straightforward to translate the cases of $F$ into a set $\rD_F$ of dynamic rules (shown below) to obtain a rewriting system that simulates the iterated application of $F$ on positive integers.
\begin{equation*}
  \begin{array}{rcl}
    \tO \btR & \to & \btR \\
    \tZ \bZ \btR & \to & \tZ \btR \\
    \tO \bZ \btR & \to & \tO \btR \\
    \tO \bO \btR & \to & \tO \tT \btR \\
    \tT \bO \btR & \to & \tT \tT \btR
  \end{array}
\end{equation*}
Termination of the rewriting system $\rF = \rD_F \cup \rX$ is equivalent to the convergence of $F$. The proof of the equivalence is essentially the same as that of \cref{thm:simulation} except for the step where we construct a nonterminating rewrite sequence from a nonconvergent trajectory. To construct such a rewrite sequence for $\rF$, we take any nonconvergent $F$-trajectory and write the first number in the trajectory in \emph{ternary} (instead of binary). If the leading digit turns out to be $2_3$ we replace it by $1_0 0_2$ since the leading digit is always assumed to be $1_0$ in our mixed-base representations. Then repeatedly performing the rightmost possible rewrite results in a rewrite sequence that simulates the $F$-trajectory.

Farkas gave an inductive proof of convergence for $F'$ via case analysis. We found an automated proof that $\rF$ is terminating via arctic matrix interpretations. Below, we present both a manual proof of convergence for $F$ and the automated proof of $\SN(\rF)$. It is worth mentioning that the default configurations of the existing termination tools (e.g., \AProVE{}, \Matchbox{}) are too conservative to prove the termination of this system, but after their authors tweaked the strategies they were also able to find automated proofs via arctic matrix interpretations. It also remains an interesting question whether the automated proof somehow captures the argument in the inductive proof.

\begin{theorem}
  For all $n \in \Np$, the trajectory $F^{\N}(n)$ contains $1$.
\end{theorem}
\begin{proof}[Proof (in the manner of Farkas~{\cite[Theorem~1]{Far05}})]
  We proceed by induction. For $n = 1$ the result holds since $F^{\N}(1) \ = (1,0,0,\dots)$. Assuming it holds for all positive integers less than $n$, we will show that it holds for $n$. In particular, we will show that the trajectory of $n$ reaches a number strictly smaller than $n$, which implies by the induction hypothesis that the trajectory contains $1$. We split into three cases.
  \begin{enumerate}[(i)]
  \item Assume $n \equiv 1 \pmod 3$. Then $F(n) = (n-1)/3$ is smaller than $n$.
  \item Assume $n \equiv 0$ or $n \equiv 2 \pmod 6$. Then $F(n) = n/2$ is smaller than $n$.
  \item Assume $n \equiv 3$ or $n \equiv 5 \pmod 6$. Denote $n$ by $N_1$, and consider the trajectory $F^{\N}(N_1) = (N_1, N_2, N_3, \dots)$. Let
    \begin{equation*}
      k = \sup \left\{i \in \Np \mathrel{\bigg\vert} N_j = \frac{3 N_{j-1} + 1}{2} \text{ for all } 1 < j \leq i\right\}.
    \end{equation*}
    We cannot have $k = \infty$, i.e., $F$ cannot apply $N \mapsto (3N+1)/2$ indefinitely, since for such a sequence we have
    \begin{equation*}
      N_j + 1 = \left(\frac{3}{2}\right)(N_{j-1}+1) = \left(\frac{3}{2}\right)^{j-1}(N_1 + 1).
    \end{equation*}
    As all the elements in the trajectory have to be integers, $k-1$ cannot exceed the number of times that $N_1+1$ can be divided by $2$, so $k$ is finite. Consider
    \begin{equation*}
      N_{k+1} = \left(\frac{3}{2}\right)^k (N_1+1)-1.
    \end{equation*}
    It is of the form $3^k \cdot L - 1$ for some positive integer $L$, so we know that either $N_{k+1} \equiv 2 \pmod 6$ or $N_{k+1} \equiv 5 \pmod 6$. Due to the way $k$ is defined, $N_{k+1}$ cannot be congruent to $5 \pmod 6$, so we have $N_{k+1} \equiv 2 \pmod 6$. In particular, $N_{k+1}$ is even, so $\left(\frac{3}{2}\right)^k (N_1+1)$ is odd, which in turn implies that $\frac{N_1 + 1}{2^k}$ is odd, i.e., we have $N_1 = 2^k \cdot (2M+1) - 1$ for some natural number $M$. Then, $N_{k+1}$ satisfies
    \begin{equation*}
      N_{k+1} = 3^k \cdot (2M+1) - 1 = 3^k \cdot 2M + 3^k - 1.
    \end{equation*}
    Since $N_{k+1} \equiv 2 \pmod 6$, we can deduce that
    \begin{equation*}
      F(N_{k+1}) = \frac{N_{k+1}}{2} = \frac{3^k \cdot 2M + 3^k - 1}{2} \equiv 1 \pmod 3,
    \end{equation*}
    and furthermore,
    \begin{equation*}
      F^2(N_{k+1}) = \frac{\frac{3^k \cdot 2M + 3^k - 1}{2} - 1}{3} = 3^{k-1} \cdot 2M + 3^{k-1} - 1 \equiv 2 \pmod 6.
    \end{equation*}
    Repeatedly applying $F$ in the above manner gives $F^{2k}(N_{k+1}) = 2M$. Since $k \geq 1$, we have $2M < 2^k \cdot (2M + 1) - 1 = N_1 = n$, and so the trajectory $F^{\N}(n)$ reaches a number strictly smaller than $n$. \qedhere
  \end{enumerate}
\end{proof}
\begin{proof}[Proof (via arctic matrix interpretations)]
  We will show $\SN(\rF)$. By \cref{lem:A-termination,lem:reversal}, we have $\SN(\rev{\rX})$. Consider the arctic matrix interpretations below. Recall that $-$ indicates $-\infty$ in the matrices. Also recall that arctic addition is $\max$ and arctic multiplication is $+$.
  {\setlength{\jot}{1em}
    \begin{NiceMatrixBlock}[auto-columns-width]
      \begin{equation}\label{eq:farkas-interpretations}
        \begin{gathered}
          [\bZ](\vx) =
          \begin{bNiceMatrix}
            \fademinfty & \fademinfty & \fademinfty & 2 & \fademinfty \\
            \fademinfty & 2 & 0 & \fademinfty & \fademinfty \\
            2 & \fademinfty & \fademinfty & \fademinfty & \fademinfty \\
            \fademinfty & \fademinfty & \fademinfty & \fademinfty & \fademinfty \\
            \fademinfty & \fademinfty & \fademinfty & \fademinfty & \fademinfty
          \end{bNiceMatrix}
          \vx \oplus
          \begin{bNiceMatrix}
            0 \\
            \fademinfty \\
            \fademinfty \\
            \fademinfty \\
            \fademinfty
          \end{bNiceMatrix}
          \qquad
          [\bO](\vx) =
          \begin{bNiceMatrix}
            \fademinfty & \fademinfty & \fademinfty & \fademinfty & 2 \\
            0 & 2 & 0 & \fademinfty & 0 \\
            2 & \fademinfty & 2 & \fademinfty & \fademinfty \\
            \fademinfty & \fademinfty & \fademinfty & \fademinfty & \fademinfty \\
            \fademinfty & \fademinfty & \fademinfty & \fademinfty & \fademinfty
          \end{bNiceMatrix}
          \vx \oplus
          \begin{bNiceMatrix}
            0 \\
            \fademinfty \\
            \fademinfty \\
            \fademinfty \\
            \fademinfty
          \end{bNiceMatrix}
          \\
          [\btL](\vx) =
          \begin{bNiceMatrix}
            0 \\
            2 \\
            \fademinfty \\
            \fademinfty \\
            4
          \end{bNiceMatrix}
          \qquad
          [\btR](\vx) =
          \begin{bNiceMatrix}
            0 & \fademinfty & \fademinfty & \fademinfty & \fademinfty \\
            \fademinfty & \fademinfty & \fademinfty & \fademinfty & \fademinfty \\
            \fademinfty & \fademinfty & \fademinfty & \fademinfty & \fademinfty \\
            \fademinfty & \fademinfty & \fademinfty & \fademinfty & \fademinfty \\
            \fademinfty & \fademinfty & \fademinfty & \fademinfty & \fademinfty
          \end{bNiceMatrix}
          \vx
          \\
          [\tZ](\vx) =
          \begin{bNiceMatrix}
            0 & 4 & 0 & \fademinfty & \fademinfty \\
            \fademinfty & 4 & \fademinfty & \fademinfty & \fademinfty \\
            \fademinfty & 4 & 0 & \fademinfty & \fademinfty \\
            0 & 3 & 0 & \fademinfty & \fademinfty \\
            \fademinfty & \fademinfty & \fademinfty & \fademinfty & \fademinfty
          \end{bNiceMatrix}
          \vx
          \qquad
          [\tO](\vx) =
          \begin{bNiceMatrix}
            1 & \fademinfty & \fademinfty & \fademinfty & \fademinfty \\
            \fademinfty & 4 & 0 & \fademinfty & \fademinfty \\
            \fademinfty & 4 & 0 & \fademinfty & \fademinfty \\
            0 & \fademinfty & \fademinfty & \fademinfty & \fademinfty \\
            0 & 3 & 0 & \fademinfty & \fademinfty
          \end{bNiceMatrix}
          \vx
          \qquad
          [\tT](\vx) =
          \begin{bNiceMatrix}
            0 & \fademinfty & 0 & \fademinfty & \fademinfty \\
            \fademinfty & 4 & \fademinfty & \fademinfty & \fademinfty \\
            0 & \fademinfty & 1 & \fademinfty & 0 \\
            \fademinfty & \fademinfty & \fademinfty & \fademinfty & \fademinfty \\
            0 & \fademinfty & 0 & \fademinfty & 0
          \end{bNiceMatrix}
          \vx
        \end{gathered}
      \end{equation}
    \end{NiceMatrixBlock}%
  }%
  Above interpretations prove $\SN(\rev{\rD_F}_\mathrm{top} \rel \rev{\rX})$ by \cref{thm:monotone-algebra}. (See \cref{sec:farkas-interpretations} for the corresponding interpretations of the rules of $\rev{\rF}$.) We then get $\SN(\rev{\rD_F} \rel \rev{\rX})$ by \cref{lem:T-top-termination}. As we know that $\rev{\rX}$ is terminating, by \cref{thm:relative-termination} we can conclude $\SN(\rev{\rD_F} \cup \rev{\rX})$, which implies $\SN(\rF)$ via \cref{lem:reversal}.
\end{proof}

\subsection{Subsets of the Collatz System}
\label{sec:subsystems-T}

It is also interesting to consider whether we can automatically prove termination of proper subsets of $\rT$. Specifically, we considered the 11 subsystems obtained by leaving out a single rewriting rule from $\rT$, and we found termination proofs via matrix interpretations for all of the 11 subproblems.

The reason for our interest in these problems is threefold:
\begin{enumerate}
\item Termination of $\rT$ implies termination of all of its subsystems, so proving its termination is at least as difficult a task as proving termination of the 11 subsystems. Therefore, the subproblems serve as additional sanity checks that an automated approach aspiring to succeed for the Collatz conjecture ought to be able to pass. Furthermore, they form an additional source of weakened variants that give us a heuristic idea of the suitable choice of parameters to use when searching for interpretations for the full problem.
\item When proving termination in a stepwise manner, we solve a sequence of relative termination problems. Having proved termination of all 11 subsystems is a partial solution to the full problem, since it implies that for any single rule $\ell \to r \in \rT$, proving $\SN(\{\ell \to r\} \rel \rT)$ settles the Collatz conjecture.
\item After the removal of a rule, the termination of the remaining system still encodes a valid mathematical question about the Collatz trajectories. The question of the termination of a proper subset is equivalent to asking if every corresponding Collatz trajectory that does not require the use of the left-out rule is convergent. For instance, removing $\{\bO \btR \to \tT \btR\}$ deletes the odd case of the Collatz function, which gives a clearly convergent function. Removing the auxiliary rules in $\rX$ lead to potentially more interesting questions: For instance, the termination of $\rT \setminus \{\bO \tO \to \tT \bZ\}$ implies that the trajectories where the pattern $\bO \tO$ is never encountered during the simulation are all convergent.

  Note that this type of inquiry is less meaningful for the unary system $\rZ$ as its simulation of the Collatz function depends crucially on each rule being present. Leaving out some subset of the rules from $\rZ$ causes the system to become terminating for a trivial reason, i.e., the computation it expresses no longer corresponds to the iterated application of some function.
\end{enumerate}

\begin{example}
  As an instance of leaving out a rule, consider the subsystem $\rT \setminus \{\bZ \tO \to \tZ \bO\}$. There is a single-step natural matrix interpretations proof that this system is terminating:
  {\setlength{\jot}{1em}
    \begin{gather*}
      [\bZ](\vx) =
      \begin{bmatrix}
        1 & 1 \\
        1 & 0
      \end{bmatrix}
      \vx
      \qquad
      [\bO](\vx) =
      \begin{bmatrix}
        1 & 3 \\
        3 & 4
      \end{bmatrix}
      \vx +
      \begin{bmatrix}
        1 \\
        1
      \end{bmatrix} \\
      [\btL](\vx) =
      \begin{bmatrix}
        1 & 5 \\
        0 & 0
      \end{bmatrix}
      \vx
      \qquad
      [\btR](\vx) =
      \begin{bmatrix}
        1 & 0 \\
        1 & 0
      \end{bmatrix}
      \vx +
      \begin{bmatrix}
        1 \\
        1
      \end{bmatrix} \\
      [\tZ](\vx) =
      \begin{bmatrix}
        7 & 2 \\
        2 & 5
      \end{bmatrix}
      \vx +
      \begin{bmatrix}
        2 \\
        1
      \end{bmatrix}
      \qquad
      [\tO](\vx) =
      \begin{bmatrix}
        2 & 1 \\
        1 & 1
      \end{bmatrix}
      \vx +
      \begin{bmatrix}
        1 \\
        0
      \end{bmatrix}
      \qquad
      [\tT](\vx) =
      \begin{bmatrix}
        2 & 2 \\
        2 & 4
      \end{bmatrix}
      \vx +
      \begin{bmatrix}
        0 \\
        2
      \end{bmatrix}
    \end{gather*}
  }%
  It can be checked that under these interpretations the rules of $\rT \setminus \{\bZ \tO \to \tZ \bO\}$ satisfy the following relations:
  \begin{NiceMatrixBlock}[auto-columns-width]
    \begin{longtable*}{IJK}
      {[}\bZ \btR](\vx) =
      \begin{bNiceMatrix}
        2 & 0 \\
        1 & 0
      \end{bNiceMatrix}
      \vx +
      \begin{bNiceMatrix}
        2 \\
        1
      \end{bNiceMatrix}
      & > &
      \begin{bNiceMatrix}
        1 & 0 \\
        1 & 0
      \end{bNiceMatrix}
      \vx +
      \begin{bNiceMatrix}
        1 \\
        1
      \end{bNiceMatrix}
      = {[}\btR](\vx)
      \\[1.825em]
      {[}\bO \btR](\vx) =
      \begin{bNiceMatrix}
        4 & 0 \\
        7 & 0
      \end{bNiceMatrix}
      \vx +
      \begin{bNiceMatrix}
        5 \\
        8
      \end{bNiceMatrix}
      & > &
      \begin{bNiceMatrix}
        4 & 0 \\
        6 & 0
      \end{bNiceMatrix}
      \vx +
      \begin{bNiceMatrix}
        4 \\
        8
      \end{bNiceMatrix}
      = {[}\tT \btR](\vx)
      \\[1.825em]
      {[}\bZ \tZ](\vx) =
      \begin{bNiceMatrix}
        9 &  7 \\
        7 &  2
      \end{bNiceMatrix}
      \vx +
      \begin{bNiceMatrix}
        3 \\
        2
      \end{bNiceMatrix}
      & > &
      \begin{bNiceMatrix}
        9 &  7 \\
        7 &  2
      \end{bNiceMatrix}
      \vx +
      \begin{bNiceMatrix}
        2 \\
        1
      \end{bNiceMatrix}
      = {[}\tZ \bZ](\vx)
      \\[1.825em]
      {[}\bZ \tT](\vx) =
      \begin{bNiceMatrix}
        4 &  6 \\
        2 &  2
      \end{bNiceMatrix}
      \vx +
      \begin{bNiceMatrix}
        2 \\
        0
      \end{bNiceMatrix}
      & > &
      \begin{bNiceMatrix}
        3 &  2 \\
        2 &  1
      \end{bNiceMatrix}
      \vx +
      \begin{bNiceMatrix}
        1 \\
        0
      \end{bNiceMatrix}
      = {[}\tO \bZ](\vx)
      \\[1.825em]
      {[}\bO \tZ](\vx) =
      \begin{bNiceMatrix}
        13 & 17 \\
        29 & 26
      \end{bNiceMatrix}
      \vx +
      \begin{bNiceMatrix}
        6 \\
        11
      \end{bNiceMatrix}
      & > &
      \begin{bNiceMatrix}
        5 & 10 \\
        4 &  7
      \end{bNiceMatrix}
      \vx +
      \begin{bNiceMatrix}
        4 \\
        2
      \end{bNiceMatrix}
      = {[}\tO \bO](\vx)
      \\[1.825em]
      {[}\bO \tO](\vx) =
      \begin{bNiceMatrix}
        5 &  4 \\
        10 &  7
      \end{bNiceMatrix}
      \vx +
      \begin{bNiceMatrix}
        2 \\
        4
      \end{bNiceMatrix}
      & > &
      \begin{bNiceMatrix}
        4 &  2 \\
        6 &  2
      \end{bNiceMatrix}
      \vx +
      \begin{bNiceMatrix}
        0 \\
        2
      \end{bNiceMatrix}
      = {[}\tT \bZ](\vx)
      \\[1.825em]
      {[}\bO \tT](\vx) =
      \begin{bNiceMatrix}
        8 & 14 \\
        14 & 22
      \end{bNiceMatrix}
      \vx +
      \begin{bNiceMatrix}
        7 \\
        9
      \end{bNiceMatrix}
      & > &
      \begin{bNiceMatrix}
        8 & 14 \\
        14 & 22
      \end{bNiceMatrix}
      \vx +
      \begin{bNiceMatrix}
        4 \\
        8
      \end{bNiceMatrix}
      = {[}\tT \bO](\vx)
      \\[1.825em]
      {[}\btL \tZ](\vx) =
      \begin{bNiceMatrix}
        17 & 27 \\
        0 & 0
      \end{bNiceMatrix}
      \vx +
      \begin{bNiceMatrix}
        7 \\
        0
      \end{bNiceMatrix}
      & > &
      \begin{bNiceMatrix}
        16 & 23 \\
        0 & 0
      \end{bNiceMatrix}
      \vx +
      \begin{bNiceMatrix}
        6 \\
        0
      \end{bNiceMatrix}
      = {[}\btL \bO](\vx)
      \\[1.825em]
      {[}\btL \tO](\vx) =
      \begin{bNiceMatrix}
        7 &  6 \\
        0 & 0
      \end{bNiceMatrix}
      \vx +
      \begin{bNiceMatrix}
        1 \\
        0
      \end{bNiceMatrix}
      & > &
      \begin{bNiceMatrix}
        7 &  6 \\
        0 & 0
      \end{bNiceMatrix}
      \vx +
      \begin{bNiceMatrix}
        0 \\
        0
      \end{bNiceMatrix}
      = {[}\btL \bZ \bZ](\vx)
      \\[1.825em]
      {[}\btL \tT](\vx) =
      \begin{bNiceMatrix}
        12 & 22 \\
        0 & 0
      \end{bNiceMatrix}
      \vx +
      \begin{bNiceMatrix}
        10 \\
        0
      \end{bNiceMatrix}
      & > &
      \begin{bNiceMatrix}
        9 & 22 \\
        0 & 0
      \end{bNiceMatrix}
      \vx +
      \begin{bNiceMatrix}
        7 \\
        0
      \end{bNiceMatrix}
      = {[}\btL \bZ \bO](\vx)
    \end{longtable*}

    Note that the left-out rule $\bZ \tO \to \tZ \bO$ ends up receiving a nondecreasing interpretation as there is no constraint to enforce otherwise:
    \begin{longtable*}{IJK}
      {[}\bZ \tO](\vx) =
      \begin{bNiceMatrix}
        3 & 2 \\
        2 & 1
      \end{bNiceMatrix}
      \vx +
      \begin{bNiceMatrix}
        1 \\
        1
      \end{bNiceMatrix}
      & \not> &
      \begin{bNiceMatrix}
        13 & 29 \\
        17 & 26
      \end{bNiceMatrix}
      \vx +
      \begin{bNiceMatrix}
        11 \\
        8
      \end{bNiceMatrix}
      = {[}\tZ \bO](\vx)
    \end{longtable*}
  \end{NiceMatrixBlock}

  The above interpretations witness for instance that the Collatz trajectory starting at $3$ (represented as $\btL \bO \btR$) is convergent, because the missing rule is not used in any derivation of $1$ ($\btL \btR$) from $3$. Below is an example derivation along with the values each string represents and a vector value of each string under the interpretations above (setting $\vx = (0, 0)$ for the purpose of demonstration). We omit the subscripts from the rewrite relations and simply write $\rew$.
  \begin{equation*}
    \begin{array}[c]{*{17}{@{\,}c@{\,}}}
      3 && 5 && 5 && 8 && 8 && 8 && 4 && 2 && 1 \\[0.25em]
      \btL \underline{\bO \btR} & \to & \underline{\btL \tT} \btR & \to & \btL \bZ \underline{\bO \btR} & \to & \btL \underline{\bZ \tT} \btR & \to & \underline{\btL \tO} \bZ \btR & \to & \btL \bZ \bZ \underline{\bZ \btR} & \to & \btL \bZ \underline{\bZ \btR} & \to & \btL \underline{\bZ \btR} & \to & \btL \btR \\[0.5em]
      \begin{bmatrix} 79 \\ 0 \end{bmatrix} & > & \begin{bmatrix} 78 \\ 0 \end{bmatrix} & > & \begin{bmatrix} 68 \\ 0 \end{bmatrix} & > & \begin{bmatrix} 62 \\ 0 \end{bmatrix} & > & \begin{bmatrix} 41 \\ 0 \end{bmatrix} & > & \begin{bmatrix} 40 \\ 0 \end{bmatrix} & > & \begin{bmatrix} 26 \\ 0 \end{bmatrix} & > & \begin{bmatrix} 14 \\ 0 \end{bmatrix} & > & \begin{bmatrix} 12 \\ 0 \end{bmatrix}
    \end{array}
  \end{equation*}
\end{example}

\cref{table:subsystems} shows the parameters of the matrix interpretations proofs that we found for the termination of each subsystem. For each rule $\ell \to r$ that is left out, we searched for a stepwise proof to show that $\rB \setminus \{\ell \to r\}$ is terminating relative to $\rT \setminus \{\ell \to r\}$ (freely utilizing weakly monotone algebras due to \cref{lem:T-top-termination}). Such a proof requires at most three steps since there are at most three rules in $\rB \setminus \{\ell \to r\}$. In the table, we report the smallest parameters (in terms of matrix dimension) that work for all of these steps. As we already know that $\SN(\rT \setminus \rB)$ holds (by \cref{lem:A-termination}), the interpretations found allow us to conclude the termination of each subsystem. This is not the only way to prove termination of the subsystems; however, we chose this uniform strategy for the sake of comparison. In the experiments we searched for matrices of up to $7$ dimensions, with the coefficients taking at most $8$ different values.

\begin{table}[t]
  \centering
  \caption{Smallest proofs found for termination of subsystems of $\rT$ in under 30 seconds. The columns show the matrix dimension $D$ and the maximum number $V$ of distinct coefficients that appear in the matrices, along with the median time to find an entire termination proof across 25 repetitions for the fixed $D$ and $V$.}
  \begin{tabular}{l rrQ{3.2em} rrQ{3.2em}}
    \toprule
    & \multicolumn{3}{c}{Natural} & \multicolumn{3}{c}{Arctic} \\
    \cmidrule(lr){2-4} \cmidrule(lr){5-7}
    Rule removed & $D$ & $V$ & Time & $D$ & $V$ & Time \\
    \midrule
    $\bZ \btR \to \btR$ & 3 & 4 & 1.42s & 3 & 5 & 15.95s \\
    $\bO \btR \to \tT \btR$ & 1 & 2 & 0.27s & 1 & 3 & 0.28s \\
    \midrule
    $\bZ \tZ \to \tZ \bZ$ & 4 & 2 & 0.92s & 3 & 4 & 2.46s \\
    $\bZ \tO \to \tZ \bO$ & 1 & 3 & 0.50s & 1 & 4 & 0.51s \\
    $\bZ \tT \to \tO \bZ$ & 1 & 2 & 0.38s & 1 & 3 & 0.39s \\
    \midrule
    $\bO \tZ \to \tO \bO$ & 4 & 3 & 1.20s & 3 & 4 & 0.87s \\
    $\bO \tO \to \tT \bZ$ & 5 & 2 & 0.89s & 4 & 3 & 0.84s \\
    $\bO \tT \to \tT \bO$ & 4 & 4 & 10.00s & 2 & 5 & 0.62s \\
    \midrule
    $\btL \tZ \to \btL \bO$ & 2 & 2 & 0.40s & 2 & 3 & 0.42s \\
    $\btL \tO \to \btL \bZ \bZ$ & 3 & 3 & 0.53s & 3 & 4 & 0.57s \\
    $\btL \tT \to \btL \bZ \bO$ & 4 & 4 & 7.51s & 4 & 3 & 4.04s \\
    \bottomrule
  \end{tabular}
  \label{table:subsystems}
\end{table}

\subsection{Odd Trajectories}
\label{sec:odd-trajectories}

In the originally defined Collatz function $C$, applying $2n+1 \mapsto 6n+4$ produces an even number, so we incorporate a single division by $2$ into the definition of the odd case and obtain the function $T$ with the same overall dynamics as $C$. Taking this idea further by performing as many divisions by $2$ as possible leads to the so-called Syracuse function $\Syr \colon \NO \to \NO$, defined as
\begin{equation*}
  \Syr(n) = \frac{3n+1}{2^k} \text{, where } k = \max \{k \in \Np \mid 2^k \text{ divides } 3n+1\}.
\end{equation*}
The Syracuse function maps each odd number to the next odd number in its Collatz trajectory. In this sense, it has the same overall dynamics as the Collatz functions $C$ or $T$ we have previously defined, i.e., $\Syr$
is convergent if and only if $C$ or $T$ is convergent. The function $T$ takes fewer steps than $C$ to reach $1$ and can be thought of as an accelerated Collatz function, making the Syracuse function a further accelerated version. Reducing the number of steps towards convergence is of interest mainly because the existence of various kinds of termination proofs is dependent on the derivational complexity\footnote{Derivational complexity is defined as the maximal length of rewrite sequences in the system as a function of the length of the initial string in the sequence.} of the rewriting system staying below certain thresholds~\cite{HL89}. For instance, given a rewriting system, if its termination can be proved in a single-step via natural matrix interpretations then its derivational complexity is at most exponential~\cite{EWZ08}. Similarly, a single-step termination proof via arctic matrix interpretations implies a linear derivational complexity~\cite{KW09}. Thus, it can be preferable to work with an alternative rewriting system with smaller derivational complexity that still simulates the dynamics of the Collatz function.

Expressing the Syracuse function as a generalized Collatz function would require infinitely many cases to account for all of the possible appearances of $2^k$ as the denominator with different values of $k$. As a result, we are unable to simulate it with a finite rewriting system. Nevertheless, we may compromise and accelerate the Collatz function by a smaller amount. We first observe that if $n \equiv 1 \pmod 8$ then $\Syr(n) = \frac{3n+1}{4}$ and if $n \equiv 3 \pmod 4$ then $\Syr(n) = \frac{3n+1}{2}$. Furthermore, for any $n \in \N$ we have $\Syr(8n+5) = \Syr(2n+1)$ since $3(8n+5)+1 = 24n+16 = 4(6n+4) = 4(3(2n+1)+1)$. Putting these observations together, we can define a generalized Collatz function $S \colon \NO \to \NO$ as follows.
\begin{equation*}
  S(n) =
  \begin{cases}
    \frac{3n+1}{4} & \text{if } n \equiv 1 \pmod 8 \\
    \frac{n-1}{4} & \text{if } n \equiv 5 \pmod 8 \\
    \frac{3n+1}{2} & \text{if } n \equiv 3 \pmod 4
  \end{cases}
\end{equation*}
\begin{example}
  The three trajectories below illustrate how the dynamics of $C$, $\Syr$, and $S$ compare to one another.
  \begin{itemize}
  \item $C^{\N}(19) = (19, 58, 29, 88, 44, 22, 11, 34, 17, 52, 26, 13, 40, 20, 10, 5, 16, 8, 4, 2, 1, 4, 2, 1, \dots)$
  \item $\Syr^{\N}(19) = (19, 29, 11, 17, 13, 5, 1, 1, \dots)$
  \item $S^{\N}(19) = (19, 29, \mathbf{7}, 11, 17, 13, \mathbf{3}, 5, 1, 1, \dots)$
  \end{itemize}
  Note that a $\Syr$-trajectory simply consists of the odd numbers in the corresponding $C$-trajectory; however, an $S$-trajectory is not necessarily contained in the corresponding $C$-trajectory. This is due to the $n \equiv 5 \pmod 8$ case of $S$ that maps $n$ to a number not necessarily contained in the $C$-trajectory (marked in bold above). Once one of the other two cases is encountered, the $S$-trajectory arrives back at the number $\Syr(n)$.
\end{example}
$S$ is convergent if and only if $C$ (or $T$) is convergent, and the number of steps that $S$ takes to converge is between that of $T$ and $\Syr$. In a manner similar to before, we can translate the cases of $S$ into a set $\rD_S$ of dynamic rules (shown below). Since we are working with odd numbers we use a new symbol $\oddR$ to mark the end of a string, viewed functionally as $\oddR(x) = 2x+1$.
\begin{equation*}
  \begin{array}{rcl}
    \bZ \bZ \oddR & \to & \tZ \oddR \\
    \bO \bZ \oddR & \to & \oddR \\
    \bO \oddR & \to & \tT \oddR \\
  \end{array}
\end{equation*}
Termination of the rewriting system $\rS = \rD_S \cup \rX$ is equivalent to the convergence of $S$. Similar to $\rT$, proving the termination of $\rS$ is currently beyond our reach, although it may potentially be an easier path to the Collatz conjecture (compared to proving $\SN(\rT)$). Failing to prove the termination of $\rS$ itself, we considered the subsystems of $\rS$ as we did for $\rT$ in \cref{sec:subsystems-T}. With matrix interpretations, termination of all but two of the 11-rule subsystems of $\rS$ was automatically proved. Despite devoting thousands of CPU hours, we were not able to find interpretations to prove that $\rS_1 = \rS \setminus \{\bZ \bZ \oddR \to \tZ \oddR\}$ or $\rS_2 = \rS \setminus \{\bO \bZ \oddR \to \oddR\}$ is terminating. By the discussion in \cref{sec:subsystems-T}, termination of subsystems of $\rS$ corresponds to questions about the trajectories of $S$. In particular, for $\rS_1$ and $\rS_2$ we have the following.
\begin{proposition}\label{prop:s-convergence}
  \hfill
  \begin{enumerate}
  \item $\SN(\rS_1)$ if and only if all nonconvergent $S$-trajectories contain some $n \equiv 1 \pmod 8$.
  \item\label{prop:s-convergence-5mod8} $\SN(\rS_2)$ if and only if all nonconvergent $S$-trajectories contain some $n \equiv 5 \pmod 8$.
  \end{enumerate}
\end{proposition}
We can prove automatically via matrix interpretations that $\rS_3 = \rS \setminus \{\bO \oddR \to \tT \oddR\}$ is terminating, which implies the following through an equivalence of the above form for $\rS_3$.
\begin{proposition}
  All nonconvergent $S$-trajectories contain some $n \equiv 3 \pmod 4$.
\end{proposition}

Leaving out some dynamic rules from a mixed-base rewriting system that simulates a generalized Collatz function allows us to phrase the above kinds of questions as termination problems. Recall that, although $S$ is equivalent to $C$ in terms of convergence, its trajectories are not necessarily contained in the Collatz trajectories, so we cannot substitute $C$ for $S$ in the above statements (except for the second item of \cref{prop:s-convergence}, since the left-out rule is the one that maps to a number not contained in the Collatz trajectory). In the next section, we perform the same kind of inquiry into the original Collatz trajectories.

\subsection{Collatz Trajectories Modulo \texorpdfstring{$8$}{8}}
\label{sec:collatz-mod-8}

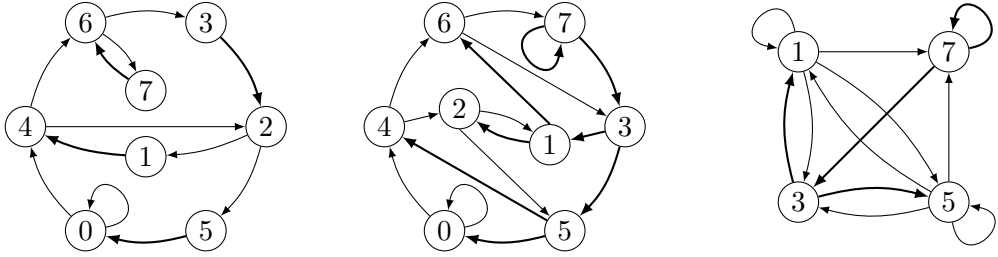
\begin{figure}[ht]
  \centering
  \begin{tikzpicture}[scale=0.8,baseline=(current bounding box.center)]
    \node[draw,circle] (4) at (-2,0) {$\!\!4\!\!$};
    \node[draw,circle] (2) at (2,0) {$\!\!2\!\!$};
    \node[draw,circle] (6) at (-1,1.73) {$\!\!6\!\!$};
    \node[draw,circle] (3) at (1,1.73) {$\!\!3\!\!$};
    \node[draw,circle] (0) at (-1,-1.73) {$\!\!0\!\!$};
    \node[draw,circle] (5) at (1,-1.73) {$\!\!5\!\!$};
    \node[draw,circle] (1) at (0,-0.5) {$\!\!1\!\!$};
    \node[draw,circle] (7) at (0,0.6) {$\!\!7\!\!$};

    \draw[-latex] (4) edge[bend left=15] (6);
    \draw[-latex] (6) edge[bend left=15] (3);
    \draw[-latex, thick] (3) edge[bend left=15] (2);
    \draw[-latex] (2) edge[bend left=15] (5);
    \draw[-latex, thick] (5) edge[bend left=15] (0);
    \draw[-latex] (0) edge[bend left=15] (4);
    \draw[-latex] (4) edge[bend left=00] (2);
    \draw[-latex] (2) edge[bend left=10] (1);
    \draw[-latex, thick] (1) edge[bend left=10] (4);
    \draw[-latex] (6) edge[bend left=15] (7);
    \draw[-latex, thick] (7) edge[bend left=15] (6);
    \draw[latex-] (0) to [out=80,in=10,looseness=6] (0);
  \end{tikzpicture}
  \hspace{2em}
  \begin{tikzpicture}[scale=0.8,baseline=(current bounding box.center)]
    \node[draw,circle] (4) at (-2,0) {$\!\!4\!\!$};
    \node[draw,circle] (3) at (2,0) {$\!\!3\!\!$};
    \node[draw,circle] (6) at (-1,1.73) {$\!\!6\!\!$};
    \node[draw,circle] (7) at (1,1.73) {$\!\!7\!\!$};
    \node[draw,circle] (0) at (-1,-1.73) {$\!\!0\!\!$};
    \node[draw,circle] (5) at (1,-1.73) {$\!\!5\!\!$};
    \node[draw,circle] (1) at (0.75,-0.3) {$\!\!1\!\!$};
    \node[draw,circle] (2) at (-0.75,0.3) {$\!\!2\!\!$};

    \draw[-latex] (4) edge[bend left=15] (6);
    \draw[-latex] (6) edge (3);
    \draw[-latex, thick] (3) edge (1);
    \draw[-latex, thick] (5) edge (4);
    \draw[-latex, thick] (5) edge[bend left=15] (0);
    \draw[-latex] (0) edge[bend left=15] (4);
    \draw[-latex] (4) edge[bend left=00] (2);
    \draw[-latex] (2) edge[bend left=15] (1);
    \draw[-latex] (2.south) to (5);
    \draw[-latex, thick] (1.north) to (6);
    \draw[-latex, thick] (1) edge[bend left=15] (2);
    \draw[-latex] (6) edge[bend left=15] (7);
    \draw[-latex, thick] (7) edge[bend left=15] (3);
    \draw[-latex, thick] (3) edge[bend left=15] (5);
    \draw[latex-] (0) to [out=80,in=10,looseness=6] (0);
    \draw[latex-, thick] (7) to [out=260,in=190,looseness=6] (7);
  \end{tikzpicture}
  \hspace{2em}
  \begin{tikzpicture}[scale=0.8,baseline=(current bounding box.center)]
    \node[draw,circle] (1) at (0,2.5) {$\!\!1\!\!$};
    \node[draw,circle] (3) at (0,0) {$\!\!3\!\!$};
    \node[draw,circle] (5) at (2.5,0) {$\!\!5\!\!$};
    \node[draw,circle] (7) at (2.5,2.5) {$\!\!7\!\!$};

    \draw[-latex] (1) edge[bend left=15] (5);
    \draw[-latex] (5) edge[bend left=15] (1);
    \draw[-latex] (5) edge (7);
    \draw[-latex] (1) edge (7);
    \draw[-latex, thick] (7) edge (3);
    \draw[-latex] (1) edge[bend left=15] (3);
    \draw[-latex, thick] (3) edge[bend left=15] (1);
    \draw[-latex] (5) edge[bend left=15] (3);
    \draw[-latex, thick] (3) edge[bend left=15] (5);
    \draw[latex-] (1) to [out=170,in=100,looseness=6] (1);
    \draw[latex-, thick] (7) to [out=80,in=10,looseness=6] (7);
    \draw[latex-] (5) to [out=350,in=280,looseness=6] (5);
  \end{tikzpicture}
  \caption{Transition graphs of the iterates in the Collatz trajectories across residue classes modulo $8$ for the functions $C$ (left), $T$ (middle), $S$ (right). For each function $f$, the edge $u \to v$ is part of its transition graph if and only if there exists some $n \equiv u \pmod 8$ such that $f(n) \equiv v \pmod 8$. Bold edges indicate transitions where $f(n) > n$.}
  \label{fig:modulo}
\end{figure}

Our initial motivation for the previous section was to find rewriting systems with reduced derivational complexity. This inquiry led us to surprisingly difficult Collatz-like problems, which we believe may be on the border of the reach of automated methods. In this section we explore more problems of a similar form and present the unsolved ones as potentially reachable targets for automated termination proving.

Let $m$ be a power of $2$. Given $k \in \{0, 1, \dots, m-1\}$, is it the case that all nonconvergent Collatz trajectories contain some $n \equiv k \pmod m$? For several values of $k$ this can be proved by inspecting the transitions of the iterates in the Collatz trajectories across residue classes modulo $m$ (shown on \cref{fig:modulo} for $m = 8$). As an example, consider the case $k = 6$.

\begin{proposition}
  All nonconvergent $C$-trajectories contain some $n \equiv 6 \pmod 8$.
\end{proposition}
\begin{proof}
  Let $t$ be a $C$-trajectory. The path that $t$ takes over the transition graph of $C$ in \cref{fig:modulo} is a concatenation of the simple cycles of the graph. If $t$ does not contain any $n \equiv 6 \pmod 8$, then the cycles containing $6$ cannot occur in the path, which leaves eight cycles (listed below). By composing the functions applied at each transition in the graph, we can compute the mapping (call it $c$) applied to an iterate after traversing a cycle:

  \begin{center}
    \renewcommand{\arraystretch}{1.2}
    \begin{tabular}{@{}p{6em}ll}
      \toprule
      Cycle & $c(n)$ \\
      \midrule
      $(0, 4, 2, 5, 0)$ & $(3n+8)/8$ \\
      $(0, 0)$ & $n/2$ \\
      $(1, 4, 2, 1)$ & $(3n+1)/4$ \\
      $(2, 5, 0, 4, 2)$ & $(3n+2)/8$ \\
      $(2, 1, 4, 2)$ & $(3n+2)/4$ \\
      $(4, 2, 5, 0, 4)$ & $(3n+4)/8$ \\
      $(4, 2, 1, 4)$ & $(3n+4)/4$ \\
      $(5, 0, 4, 2, 5)$ & $(3n+1)/8$ \\
      \bottomrule
    \end{tabular}
  \end{center}

  For $t$ to be nonconvergent, it needs to contain infinitely many occurrences of cycles that increase the values of the iterates. All of the above cycles decrease the values of the iterates, so each iterate of $t$ eventually reaches a smaller number, implying $t$ is convergent.
\end{proof}

Propositions of the above kind can also be formulated as termination of some rewriting systems. With this approach we found automated proofs for several cases.

We start by studying the smallest interesting case $m=4$. Consider the following two sets of dynamic rules to simulate the Collatz function.
\begin{equation*}
  \rD_{C}^{\mathrm{Even}/4} = \left\{
    \begin{array}{rcl}
      \bZ \bZ \btR & \to & \bZ \btR \\
      \bO \bZ \btR & \to & \bO \btR \\
      \bO \btR & \to & \tT \bZ \btR
    \end{array}
  \right\}
  \qquad
  \rD_{T}^{\mathrm{Odd}/4} = \left\{
    \begin{array}{rcl}
      \bZ \btR & \to & \btR \\
      \bZ \bO \btR & \to & \tO \bZ \btR \\
      \bO \bO \btR & \to & \tT \bO \btR
    \end{array}
  \right\}
\end{equation*}
It is straightforward to see that the SRS $\rC^{\mathrm{Even}/4} = \rD_{C}^{\mathrm{Even}/4} \cup \rX$ simulates the function $C$ and the SRS $\rT^{\mathrm{Odd}/4} = \rD_{T}^{\mathrm{Odd}/4} \cup \rX$ simulates the function $T$. The first two rules of $\rD_{C}^{\mathrm{Even}/4}$ express the application of $C$ to numbers of the form $4n$ and $4n+2$. Similarly, the last two rules of $\rD_{T}^{\mathrm{Odd}/4}$ express the application of $T$ to numbers of the form $4n+1$ and $4n+3$. Observing that a $T$-trajectory contains the same odd numbers as the corresponding $C$-trajectory, we have the following equivalences.
\begin{proposition}\label{prop:mod4-equiv}
  \hfill
  \begin{enumerate}
  \item $\SN(\rC^{\mathrm{Even}/4} \setminus \{\bZ \bZ \btR \to \bZ \btR\})$ if and only if all nonconvergent $C$-trajectories contain some $n \equiv 0 \pmod 4$.
  \item $\SN(\rC^{\mathrm{Even}/4} \setminus \{\bO \bZ \btR \to \bO \btR\})$ if and only if all nonconvergent $C$-trajectories contain some $n \equiv 2 \pmod 4$.
  \item $\SN(\rT^{\mathrm{Odd}/4} \setminus \{\bZ \bO \btR \to  \tO \bZ \btR\})$ if and only if all nonconvergent $C$-trajectories contain some $n \equiv 1 \pmod 4$.
  \item $\SN(\rT^{\mathrm{Odd}/4} \setminus \{\bO \bO \btR \to \tT \bO \btR\})$ if and only if all nonconvergent $C$-trajectories contain some $n \equiv 3 \pmod 4$.
  \end{enumerate}
\end{proposition}
We were able to prove termination of the above four systems via matrix interpretations, thus establishing the following.
\begin{theorem}\label{thm:collatz-mod4}
  If there exists a nonconvergent Collatz trajectory, it cannot avoid any of the residue classes modulo $4$.
\end{theorem}

Next, we look at the case $m=8$. It can be studied via the analogously defined rewriting systems:
\begin{equation*}
  \rD_{C}^{\mathrm{Even}/8} = \left\{
    \begin{array}{rcl}
      \bZ \bZ \bZ \btR & \to & \bZ \bZ \btR \\
      \bZ \bO \bZ \btR & \to & \bZ \bO \btR \\
      \bO \bZ \bZ \btR & \to & \bO \bZ \btR \\
      \bO \bO \bZ \btR & \to & \bO \bO \btR \\
      \bO \btR & \to & \tT \bZ \btR
    \end{array}
  \right\}
  \qquad
  \rD_{T}^{\mathrm{Odd}/8} = \left\{
    \begin{array}{rcl}
      \bZ \btR & \to & \btR \\
      \bZ \bZ \bO \btR & \to & \tZ \bO \bZ \btR \\
      \bZ \bO \bO \btR & \to & \tO \bZ \bO \btR \\
      \bO \bZ \bO \btR & \to & \tT \bZ \bZ \btR \\
      \bO \bO \bO \btR & \to & \tT \bO \bO \btR
    \end{array}
  \right\}
\end{equation*}
The SRS $\rC^{\mathrm{Even}/8} = \rD_{C}^{\mathrm{Even}/8} \cup \rX$ simulates the function $C$ and the SRS $\rT^{\mathrm{Odd}/8} = \rD_{T}^{\mathrm{Odd}/8} \cup \rX$ simulates the function $T$. We do not formally state the analogue of \cref{prop:mod4-equiv} --- terminations of the eight subsystems obtained from $\rC^{\mathrm{Even}/8}$ or $\rT^{\mathrm{Odd}/8}$ are each equivalent to the question of whether all nonconvergent Collatz trajectories have to encounter the corresponding residue class modulo $8$. Through automated proofs, we were able to establish the following.
\begin{theorem}\label{thm:collatz-2346mod8}
  If there exists a nonconvergent Collatz trajectory, it cannot avoid the residue classes of $2$, $3$, $4$, $6$ modulo $8$.
\end{theorem}
\cref{thm:collatz-mod4,thm:collatz-2346mod8} are also not difficult to prove by hand; however, it remains open whether an analogue of \cref{thm:collatz-2346mod8} holds for the residue classes of $0$, $1$, $5$, $7$ modulo $8$, that is, we have neither manual nor automated proofs for these cases. They may serve well as easier\footnote{See Denend~\cite[Section~4]{Den18} for a heuristic prediction of the relative easiness of those cases. Denend uses an empirical measure based on the number of odd elements occurring in a trajectory until either a certain residue class or $1$ is encountered, which resembles the \emph{ones-ratio} and the \emph{stopping time ratio} defined by Applegate~and~Lagarias~\cite{AL03}.} Collatz-like open problems to test our abilities against, so we offer a \$500 reward for a resolution of any of the following conjectures.\footnote{The case $0$ is omitted as its answer follows from the case $5$ by \cref{fig:modulo}.}
\begin{conjecture}
  All nonconvergent $C$-trajectories contain some $n \equiv 1 \pmod 8$.
\end{conjecture}
\begin{conjecture}
  All nonconvergent $C$-trajectories contain some $n \equiv 5 \pmod 8$.
\end{conjecture}
\begin{conjecture}
  All nonconvergent $C$-trajectories contain some $n \equiv 7 \pmod 8$.
\end{conjecture}
Even the following easier statement remains unproved, although we were able to prove (both manually and automatically) its versions with $\{1,5\}$ and $\{1,7\}$ as the residue classes.
\begin{conjecture}\label{conj:5or7mod8}
  All nonconvergent $C$-trajectories contain either some $n \equiv 5 \pmod 8$ or some $n \equiv 7 \pmod 8$.
\end{conjecture}

We can alternatively formulate each of the above conjectures as the convergence problem of some partial generalized Collatz function. For instance, \cref{conj:5or7mod8} is equivalent to the convergence of the function $H \colon \Np_\bot \to \Np_\bot$ defined as follows.
\begin{equation}\label{eq:fn-5or7mod8}
  H(n) =
  \begin{cases}
    \frac{3n}{4} & \text{if } n \equiv 0 \pmod 4 \\
    \frac{9n+1}{8} & \text{if } n \equiv 7 \pmod 8 \\
    \bot & \text{otherwise}
  \end{cases}
\end{equation}
\begin{proposition}
  $H$ is convergent if and only if all nonconvergent $C$-trajectories contain either some $n \equiv 5 \pmod 8$ or some $n \equiv 7 \pmod 8$.
\end{proposition}
\begin{proof}
  Recall from \cref{fig:modulo} the transition graph of the iterates in $C$-trajectories across residue classes modulo $8$. Observe that if a nonconvergent trajectory avoids the residue classes of $5$ and $7 \pmod 8$, then it has to visit the residue class of $1 \pmod 8$, so we may assume without loss of generality that each nonconvergent trajectory starts with some $n \equiv 1 \pmod 8$.

  Then a nonconvergent trajectory that avoids the residue classes of $5$ and $7 \pmod 8$ can traverse only the cycle of classes $(1, 4, 2, 1)$ or $(1, 4, 6, 3, 2, 1)$. The former, which is traversed only if $n \equiv 1 \pmod{32}$, effects the map $n \mapsto (3n+1)/4$. The latter, which is traversed only if $n \equiv 57 \pmod{64}$, effects the map $n \mapsto (9n+7)/8$. If $n$ belongs to any other residue class, either the trajectory is convergent, or it contains some $k \equiv 5 \pmod 8$ or some $k \equiv 7 \pmod 8$. As a result, the function $G \colon \Np_\bot \setminus \{1\} \to \Np_\bot \setminus \{1\}$ defined as
  \begin{equation*}
    G(n) =
    \begin{cases}
      \frac{3n+1}{4} & \text{if } n \equiv \phantom{5}1 \pmod{32} \\
      \frac{9n+7}{8} & \text{if } n \equiv 57 \pmod{64} \\
      \bot & \text{otherwise}
    \end{cases}
  \end{equation*}
  is nonconvergent if and only if there exists a nonconvergent $C$-trajectory that avoids the residue classes of $5$ and $7 \pmod 8$. Moreover, $H$ as defined in~\cref{eq:fn-5or7mod8} is convergent if and only if $G$ is convergent, because the bijection $n \mapsto (n-1)/8$ maps a $G$-trajectory to an $H$-trajectory.
\end{proof}

\subsection{SAT Solving Considerations}
\label{sec:sat-solving}

\paragraph{Phase-saving heuristic.} Most top-tier termination tools, such as \AProVE{}, \Matchbox{}, and \TTTT{}, use the SAT solver \MiniSat~\cite{ES04} to search for matrix interpretations. This choice is somewhat surprising as \MiniSat{} has not been updated since 2008 and the performance of SAT solvers has improved significantly in the last decade. The use of \MiniSat{} in these provers is motivated by its observed effectiveness in finding interpretations. We investigated the reason for this, which turned out to be a heuristic that \MiniSat{} disables in its default configuration.

Most SAT solvers sort variables by the frequency of their occurrences in recent conflict clauses, selecting the highest ranked one as the variable to branch on~\cite{MMZ+01}. This procedure is similar in \MiniSat{} and modern SAT solvers. However, once a variable is picked, the solvers differ in the truth value assignment that they explore first. \MiniSat{} uses negative branching~\cite{ES04}, which explores the ``false'' branch first for all decision variables. The empirical effectiveness of this heuristic likely stems from its interaction with the common choices of propositional encodings used when translating problems into SAT instances. Modern SAT solvers use phase-saving~\cite{PD07}, which first explores the branch corresponding to the truth value to which the variable was forced most recently during unit propagation. We offer the following explanation for the effectiveness of negative branching.

We use the \emph{order encoding}~\cite{TTKB09,PJ11} when translating the integer constraints from \cref{sec:natural-arctic-matrix-interpretations} into CNF formulas. In this encoding, each variable $X$ that takes values in a finite domain $A = \{0,1,\dots,n+1\}$ is represented by a corresponding list of $n+1$ Boolean variables $(x_0, x_1, \dots, x_n)$. Each Boolean variable $x_i$ indicates $X > i$. As an example, $X = 2$ has the corresponding assignment $(1, 1, 0, 0, \dots, 0)$. This encoding gives a natural way of representing inequality constraints. For instance, the constraint $a < X \leq b$ can be enforced simply by setting $x_a = 1$ and $x_b = 0$. As another example, for variables $X$, $Y$ the inequality $X \leq Y$ can be expressed by the conjunction $\bigwedge_{i=0}^{n} (\neg x_{i} \vee y_{i})$. Consequently, when using the order encoding, representations that correspond to smaller numbers require setting a larger fraction of the Boolean variables to false. For instance, the representation of zero or minus infinity assigns all variables to false. Moreover, the matrices and the vectors that occur in the interpretations are often sparse and contain relatively few large numbers. As a result, negative branching possibly directs the search towards a satisfying assignment more often.

Additionally, in contrast to phase-saving, negative branching tends to guide the solver to a potentially different part of the search space after a restart (which is performed frequently in SAT solvers). The increased diversity of the search space (across restarts) that is explored due to negative branching is also possibly helpful for the kinds of the problems discussed in this paper.

One can enable negative branching in most SAT solvers, which in our case improves solver performance for some of the formulas that encode the existence of interpretations.

\paragraph{Heavy-tailed behavior of SAT solvers.} Combinatorial search algorithms can have a tendency to show a large variance in runtime across different, randomly selected initial conditions of the search. For instance, Gomes~et~al.~\cite{GSCK00} observed that DPLL solvers~\cite{DP60,DLL62} had a substantial fraction of very short runs when dealing with relatively hard instances. This suggests that trying to solve such instances using several parallel runs is a reasonable strategy, since some of the runs will have a nonnegligible chance of finishing early.

Modern SAT solvers incorporate heuristics that alleviate the heavy-tailed behavior to some extent; however, the variability in runtime (especially on satisfiable instances) is not completely eliminated. In our termination prover, we run multiple instances of the SAT solver (specifically, \CaDiCaL{}~\cite{Bie19}) in parallel on the formula that encodes the existence of an interpretation. We introduce additional randomness to the procedure by running each parallel instance with a different shuffling of the clauses in the formula. Ideally, solver performance would be invariant to the ordering of the clauses; however, the clause order indirectly affects the variable selection heuristic, which ends up influencing runtime.

\paragraph{Experiments with solver configurations.} \cref{table:sat-solving} shows the difference in runtime for both \MiniSat{} and \CaDiCaL{} due to negative branching and parallelism with a few small-scale experiments. From all the automated proofs in this paper, we selected the steps where the SAT solving phase took longer than 2 seconds, and reran the prover for these specific steps with different configurations of those SAT solvers. When running only a single instance of a solver, we used the nonshuffled formula. In the parallel case, we ran 8 solver instances with one of the instances receiving the nonshuffled formula and the other 7 receiving random shuffles. Running multiple instances of a solver in parallel is a considerable improvement in all cases where the allotted time is enough to find an interpretation. Negative branching has an adverse effect in some cases, although in two of the most difficult instances ($\rF$ and $\rC^{\mathrm{Even}/8} \setminus \{\bO \bO \bZ \btR \to \bO \bO \btR\}$) it makes the difference between finding an interpretation relatively quickly and timing out. It also appears that, at least for the more difficult cases, replacing \MiniSat{} by \CaDiCaL{} can be beneficial. These experiments are not extensive enough to derive broad conclusions; however, they suggest that running multiple instances of a modern SAT solver with negative branching enabled in some and disabled in others can be a reasonably diverse strategy when searching for matrix interpretations.
\begin{table}[t]
  \centering
  \caption{Difference in runtime due to negative branching and parallelism. Each row corresponds to a single step of a relative termination proof occurring in the previous sections of this paper. The leftmost columns show the original system being proved as terminating, along with the parameters of the interpretations being searched for ($D$ for dimension, $V$ for the number of different values each coefficient may take) that remove at least a single rule from the system. For each solver, each of the remaining columns shows the median time the experiment takes over 25 repetitions. Each run times out after 120 seconds, and runs that time out are counted as having taken double the allowed time.}
  {\renewcommand{\arraystretch}{1.2}
    \setlength{\tabcolsep}{3.4pt}
    \begin{tabular}{l rrr c Q{4em}Q{4em} Q{4em}Q{4em}}
      \toprule
      &&&&& \multicolumn{2}{c}{Phase-saving} & \multicolumn{2}{c}{Negative branching} \\
      \cmidrule(lr){6-7} \cmidrule(lr){8-9}
      Problem & Interp. & $D$ & $V$ & Solver & \centering Single & \centering Parallel & \centering Single & \centering Parallel \tabularnewline
      \midrule
      $\rF$ & Arctic & $5$ & $8$ & \multirow{8}{*}{\rotatebox[origin=c]{90}{\MiniSat{}}} & 240.00s & 240.00s & 85.60s & \textbf{5.47s} \\
      $\rT \setminus \{\bZ \btR \to \btR\}$ & Arctic & $3$ & $5$ && 0.62s & \textbf{0.09s} & 67.69s & 15.74s \\
      $\rT \setminus \{\bZ \tZ \to \tZ \bZ\}$ & Arctic & $3$ & $4$ && 3.18s & 0.91s & 2.19s & 0.91s \\
      $\rT \setminus \{\bO \tT \to \tT \bO\}$ & Natural & $4$ & $4$ && 240.00s & 63.71s & 240.00s & 87.85s \\
      $\rT \setminus \{\btL \tT \to \btL \bZ \bO\}$ & Natural & $4$ & $4$ && 240.00s & 240.00s & 69.73s & 5.25s \\
      $\rT \setminus \{\btL \tT \to \btL \bZ \bO\}$ & Arctic & $4$ & $3$ && 1.39s & \textbf{0.18s} & 7.65s & 1.71s \\
      $\rC^{\mathrm{Even}/8} \setminus \{\bO \bO \bZ \btR \to \bO \bO \btR\}$ & Natural & $3$ & $11$ && 240.00s & 240.00s & 240.00s & 240.00s \\
      $\rT^{\mathrm{Odd}/8} \setminus \{\bZ \bO \bO \btR \to \tO \bZ \bO \btR\}$ & Arctic & $3$ & $12$ && 2.35s & 0.36s & 117.37s & 0.88s \tabularnewline
      \midrule
      $\rF$ & Arctic & $5$ & $8$ & \multirow{8}{*}{\rotatebox[origin=c]{90}{\CaDiCaL{}}} & 240.00s & 240.00s & 44.45s & 9.22s \\
      $\rT \setminus \{\bZ \btR \to \btR\}$ & Arctic & $3$ & $5$ && 1.52s & 0.13s & 29.95s & 13.12s \\
      $\rT \setminus \{\bZ \tZ \to \tZ \bZ\}$ & Arctic & $3$ & $4$ && 3.75s & \textbf{0.83s} & 3.27s & 1.71s \\
      $\rT \setminus \{\bO \tT \to \tT \bO\}$ & Natural & $4$ & $4$ && 75.78s & 19.12s & 29.62s & \textbf{8.13s} \\
      $\rT \setminus \{\btL \tT \to \btL \bZ \bO\}$ & Natural & $4$ & $4$ && 75.05s & \textbf{5.22s} & 24.31s & 6.43s \\
      $\rT \setminus \{\btL \tT \to \btL \bZ \bO\}$ & Arctic & $4$ & $3$ && 3.33s & 0.52s & 11.55s & 3.84s \\
      $\rC^{\mathrm{Even}/8} \setminus \{\bO \bO \bZ \btR \to \bO \bO \btR\}$ & Natural & $3$ & $11$ && 240.00s & 240.00s & 240.00s & \textbf{79.05s} \\
      $\rT^{\mathrm{Odd}/8} \setminus \{\bZ \bO \bO \btR \to \tO \bZ \bO \btR\}$ & Arctic & $3$ & $12$ && 1.94s & \textbf{0.33s} & 3.28s & 0.38s \\
      \bottomrule
    \end{tabular}
    \label{table:sat-solving}
  }
\end{table}

\section{Related Work}
\label{sec:related-work}
There are several previous studies of the Collatz conjecture through alternative models of computation.

\paragraph{String rewriting systems.} To our knowledge, Zantema~\cite{Zan05}, with his system $\rZ$, which we saw in \cref{sec:rewriting-unary}, was the first to attempt using an automated method and string rewriting to search for a proof of the Collatz conjecture. In addition, although we independently discovered the mixed binary--ternary system described in \cref{sec:rewriting-mixed-base}, Scollo~\cite{Sco05} had essentially the same idea, the distinction being that he adopted a functional view of the digits that is different than in~\cref{eq:interp}. Scollo was not concerned with proving termination, though, and proposed rewriting primarily as a formalism that forgoes the arithmetic interpretation of the iterates and instead emphasizes the dynamic/computational behavior.

\paragraph{Tag systems.} An $m$-tag~system~\cite{Pos43} is a computational model described by a set of production rules that map symbols to strings (or ``tags''). Given an initial string $X = x_0 x_1 \dots x_n$, the tag system finds the rule $x_0 \tagr Z$ whose LHS matches the leftmost symbol of $X$, appends the tag $Z$ to $X$, and deletes the first $m$ symbols of $XZ$; resulting in a new string $X' = x_m x_{m+1} \dots x_n Z$. The tag system repeats this transformation until reaching a string of length less than $m$, at which point it halts.

De Mol~\cite{De08} showed the existence of a small $2$-tag system with the following rules that simulates the iterated application of the Collatz function given a unary representation:
\begin{equation*}
  \begin{array}{rcl}
    \un & \tagr & \btL \btR \\
    \btL & \tagr & \un \\
    \btR & \tagr & \un \un \un
  \end{array}
\end{equation*}
This tag system halts if and only if the Collatz conjecture holds, giving yet another formulation of the problem. De Mol further extended this scheme to allow the simulation of an arbitrary generalized Collatz function with modulus $d$ by a $d$-tag system.

\paragraph{Cellular automata.} Kari~\cite{Kar12} designed 1D cellular automata that perform multiplication by $3$ and $3/2$ in base $6$, and reformulated both the Collatz conjecture and Mahler's $3/2$ problem as sets of constraints to be satisfied by the space-time diagrams of these cellular automata.

\paragraph{String arithmetic.} Kauffman~\cite{Kau95} developed a formalism to perform arithmetic, which he called \emph{string arithmetic}, and expressed the Collatz conjecture within it. This formalism works with unary representations of numbers, and uses the three symbols $\un$, $\btL$, $\btR$. Letting $\emp$ denote the empty string and $\num$ be any string representing a number, string arithmetic consists of the following bidirectional rewrite rules (or ``identities'') to convert between different strings representing the same number:
\begin{equation*}
  \begin{array}{rcl}
    \btR \btL & \saeq & \emp \\
    \un \un & \saeq & \btL \un \btR \\
    \un \num & \saeq & \num \un
  \end{array}
\end{equation*}
Along with the above identities, Kauffman's encoding of the Collatz function uses the following two rules:
\begin{equation*}
  \begin{array}{lcl}
    \btL \num \btR & \to & \num \\
    \btL \num \btR \un & \to & \btL \num \un \btR \num
  \end{array}
\end{equation*}
The Collatz conjecture becomes the question of whether for strings of $\un$s of all lengths there exists a rewrite sequence using the five rules above to reach the string $\un$.

\section{Conclusions}
\label{sec:conclusions}

We presented an approach to prove convergence of instances of generalized Collatz functions by translating them into string rewriting systems and applying SAT solving to find matrix interpretations that show termination. Important components of our approach include the quality of the translation into rewriting, the use of weakly monotone $\Sigma$-algebras, and dedicated heuristics in the SAT solving phase. We considered some interesting, simpler variants of the Collatz conjecture to gauge the feasibility of this approach in proving mathematically interesting statements. We observed that some variants could be solved only via natural matrix interpretations, while others required arctic matrix interpretations.

Several extensions to this work can further our understanding of the potential of rewriting techniques for answering mathematical questions. For instance, it is of interest to study the efficacy of different termination proving techniques on the problems that we considered. In particular, we found matrix interpretations to be the most successful for our purposes despite the existence of newer techniques developed for automatically proving termination of a few select challenging instances. Although matrix interpretations lead to automated proofs of several weakened variants discussed in this paper, it might still be the case that there exist no such interpretations to establish the termination of the Collatz system $\rT$ (as we have shown in the case of natural matrix interpretations for Zantema's system $\rZ$). This would be an interesting result in itself. Another issue is the matter of representation; specifically, it is worth exploring whether there exists a suitable translation of the Collatz conjecture into the termination of a term, instead of string, rewriting system since many automated termination proving techniques are generalized to term rewriting. Finally, injecting problem-specific knowledge into the rewriting systems or the termination techniques would be helpful as there exists a wealth of information about the Collatz conjecture that could simplify the search for a termination proof.


\section*{Acknowledgments}
We thank Johannes Waldmann for insightful discussions regarding arctic matrix interpretations, for pointing us to~\cite{Kar12} and the rewriting system in \cref{ex:waldmann-unary}, for responding to our challenge to solve Farkas' variant with \Matchbox, and for feedback on an early draft. We thank Carsten Fuhs and Jürgen Giesl for responding to our challenge to solve Farkas' variant with \AProVE\@. We additionally thank Carsten Fuhs for his thorough explanations of the dependency pair framework and \AProVE's strategies. We thank Florian Frohn for responding to the challenge to solve the subsystems from \cref{sec:subsystems-T} with \AProVE\@. We thank Jeffrey Lagarias for discussions regarding the problems in \cref{sec:collatz-mod-8}. We thank Luke Schaeffer and Chris Lynch for discussions on alternative rewriting systems that simulate the Collatz map. We thank Amazon Web Services for computing resources. We thank Jeremy Avigad and Jasmin Blanchette for their detailed comments on an early draft. Finally, we thank the reviewers of CADE for their comments on the preliminary version of this paper and the reviewers of the Journal of Automated Reasoning for their comments on the journal version.

This material is based upon work supported by the National Science Foundation under grant CCF-2006363.


\newcommand{\etalchar}[1]{$^{#1}$}
 \newcommand{\Proc}[1]{Proceedings of the \nth{#1}}
  \newcommand{\STOC}[1]{\Proc{#1} Symposium on Theory of Computing (STOC)}
  \newcommand{\FOCS}[1]{\Proc{#1} Symposium on Foundations of Computer Science
  (FOCS)} \newcommand{\SODA}[1]{\Proc{#1} Symposium on Discrete Algorithms
  (SODA)} \newcommand{\CCC}[1]{\Proc{#1} Computational Complexity Conference
  (CCC)} \newcommand{\ITCS}[1]{\Proc{#1} Innovations in Theoretical Computer
  Science (ITCS)} \newcommand{\ICS}[1]{\Proc{#1} Innovations in Computer
  Science (ICS)} \newcommand{\ICALP}[1]{\Proc{#1} International Colloquium on
  Automata, Languages, and Programming (ICALP)}
  \newcommand{\STACS}[1]{\Proc{#1} Symposium on Theoretical Aspects of Computer
  Science (STACS)} \newcommand{\LICS}[1]{\Proc{#1} Symposium on Logic in
  Computer Science (LICS)} \newcommand{\CSLw}[1]{\Proc{#1} International
  Workshop on Computer Science Logic (CSL)} \newcommand{\CSL}[1]{\Proc{#1}
  Conference on Computer Science Logic (CSL)} \newcommand{\IJCAR}[1]{\Proc{#1}
  International Joint Conference on Automated Reasoning (IJCAR)}
  \newcommand{\CADE}[1]{\Proc{#1} Conference on Automated Deduction (CADE)}
  \newcommand{\SAT}[1]{\Proc{#1} International Conference on Theory and
  Applications of Satisfiability Testing (SAT)} \newcommand{\RTA}[1]{\Proc{#1}
  International Conference on Rewriting Techniques and Applications (RTA)}
  \newcommand{\MFCS}[1]{\Proc{#1} International Symposium on Mathematical
  Foundations of Computer Science (MFCS)} \newcommand{\TAMC}[1]{\Proc{#1}
  International Conference on Theory and Applications of Models of Computation
  (TAMC)} \newcommand{\DLT}[1]{\Proc{#1} International Conference on
  Developments in Language Theory (DLT)} \newcommand{\TABLEAUX}[1]{\Proc{#1}
  International Conference on Automated Reasoning with Analytic Tableaux and
  Related Methods (TABLEAUX)} \newcommand{\TACAS}[1]{\Proc{#1} International
  Conference on Tools and Algorithms for the Construction and Analysis of
  Systems (TACAS)} \newcommand{\LPAR}[1]{\Proc{#1} International Conference on
  Logic for Programming, Artificial Intelligence and Reasoning (LPAR)}
  \newcommand{\HVC}[1]{\Proc{#1} Haifa Verification Conference (HVC)}
  \newcommand{\DAC}[1]{\Proc{#1} Design Automation Conference (DAC)}
  \newcommand{\DATE}{Proceedings of the Design, Automation and Test in Europe
  (DATE)} \newcommand{\ISAIM}[1]{\Proc{#1} International Symposium on
  Artificial Intelligence and Mathematics (ISAIM)}
  \newcommand{\SOFSEM}[1]{\Proc{#1} International Conference on Current Trends
  in Theory and Practice of Computer Science (SOFSEM)}
  \newcommand{\AAAI}[1]{\Proc{#1} AAAI Conference on Artificial Intelligence
  (AAAI)} \newcommand{\ICML}[1]{\Proc{#1} International Conference on Machine
  Learning (ICML)} \newcommand{\NeurIPS}[1]{\Proc{#1} Conference on Neural
  Information Processing Systems (NeurIPS)}

\appendix
\section{Alternative Proof of \texorpdfstring{\cref{lem:blockwise-termination}}{Lemma}}
\label{sec:alternative-proof-blockwise-termination}

Let us recall the statement of the result before proceeding with the proof.

\blockwise*
\begin{proof}
  Assume there exists some string $X \in \{\bZ, \bO, \tZ, \tO, \tT, \btL, \btR\}^*$ that admits an infinite rewrite sequence for $\rT$. View $X$ as split into blocks delimited by $\btL$ or $\btR$, i.e., for some $k \in \Np$, write
  \begin{equation*}
    X = N_1 \delim_1 N_2 \delim_2 \dots \delim_{k-1} N_k,
  \end{equation*}
  where $\delim_i \in \{\btL, \btR\}$ for each $1 \leq i < k$, and $N_i \in \{\bZ, \bO, \tZ, \tO, \tT\}^*$ for each $1 \leq i \leq k$. Any string containing $X$ also admits an infinite rewrite sequence, so instead consider $Z = \delim_0 X \delim_k$, written as
  \begin{equation*}
    Z = \delim_0 N_1 \delim_1 N_2 \delim_2 \dots \delim_{k-1} N_k \delim_k,
  \end{equation*}
  where $\delim_0, \delim_k \in \{\btL, \btR\}$.

  As $\rT$ is not terminating on $Z$, there exists some rewrite rule applicable to it. Consider an arbitrary block of $Z$ along with its neighbors (to which some rewrite rule of $\rT$ applies), written as
  \begin{equation*}
    V = P c Q d R,
  \end{equation*}
  where $c, d \in \{\btL, \btR\}$ and $P, Q, R \in \{\bZ, \bO, \tZ, \tO, \tT\}^*$. Due to the shapes of the rules of $\rT$, the string $V$ can be rewritten only into one of the following:
  \begin{equation*}
    \begin{array}{rclrrcl}
      V_1 &=& \widetilde{P} c Q d R & \qquad \text{ such that } & P c & \to_\rT & \widetilde{P} c \\[0.25em]
      V_2 &=& P c \widetilde{Q} d R & \qquad \text{ such that } & c Q d & \to_\rT & c \widetilde{Q} d \\[0.25em]
      V_3 &=& P c Q d \widetilde{R} & \qquad \text{ such that } & d R & \to_\rT & d \widetilde{R}
    \end{array}
  \end{equation*}
  In any case, the delimiters are unchanged, and only a single block is affected by the rewrite. This means that any rewrite sequence that starts from $Z$ consists of several sequences that each operate entirely on some block of $Z$. Thus, since $Z$ has finitely many blocks, there exists some block that admits an infinite rewrite sequence. In particular, there exists a string $cNd$, where $c,d \in \{\btL, \btR\}$ and $N \in \{\bZ, \bO, \tZ, \tO, \tT\}^*$, that can be rewritten indefinitely. We claim that this requires $c = \btL$ and $d = \btR$, so the string $cNd$ is of the canonical form $\btL (\bZ \vert \bO \vert \tZ \vert \tO \vert \tT)^* \btR$. As shown below, the other cases are all terminating.
  \begin{enumerate}[(i)]
  \item $\btR N \btR$: This string does not contain $\btL$, so it can be rewritten using only the rules in $\rT \setminus \rB$, but we know from \cref{lem:A-termination} that $\SN(\rT \setminus \rB)$, so there can be no infinite rewrite sequence.
  \item $\btL N \btL$: This string does not contain $\btR$, so it can be rewritten using only the rules in $\rT \setminus \rD_T$, but we know from \cref{lem:A-termination} that $\SN(\rT \setminus \rD_T)$, so there can be no infinite rewrite sequence.
  \item $\btR N \btL$: The SRS $\rT$ contains no rules of the form $s \btL \to t \btL$ or $\btR s \to \btR t$ for any $s, t$, so this string can be rewritten using only the rules in $\rT \setminus (\rB \cup \rD_T)$. Since each subset of a terminating SRS is terminating, by \cref{lem:A-termination} we conclude $\SN(\rT \setminus (\rB \cup \rD_T))$, so there can be no infinite rewrite sequence.
  \end{enumerate}
  This proves the contrapositive of the lemma statement.
\end{proof}

\section{Remaining Part of the Automated Proof for Farkas' Variant}
\label{sec:farkas-interpretations}

With the interpretations from~\cref{eq:farkas-interpretations}, the rules of $\rev{\rF}$ satisfy the following relations:

\begin{NiceMatrixBlock}[auto-columns-width]
  \begin{longtable*}{IJK}
    {[}\btR \tO](\vx) =
    \begin{bNiceMatrix}
      1 & \fademinfty & \fademinfty & \fademinfty & \fademinfty \\
      \fademinfty & \fademinfty & \fademinfty & \fademinfty & \fademinfty \\
      \fademinfty & \fademinfty & \fademinfty & \fademinfty & \fademinfty \\
      \fademinfty & \fademinfty & \fademinfty & \fademinfty & \fademinfty \\
      \fademinfty & \fademinfty & \fademinfty & \fademinfty & \fademinfty
    \end{bNiceMatrix}
    \vx \oplus
    \begin{bNiceMatrix}
      \fademinfty \\
      \fademinfty \\
      \fademinfty \\
      \fademinfty \\
      \fademinfty
    \end{bNiceMatrix}
    & \gg &
    \begin{bNiceMatrix}
      0 & \fademinfty & \fademinfty & \fademinfty & \fademinfty \\
      \fademinfty & \fademinfty & \fademinfty & \fademinfty & \fademinfty \\
      \fademinfty & \fademinfty & \fademinfty & \fademinfty & \fademinfty \\
      \fademinfty & \fademinfty & \fademinfty & \fademinfty & \fademinfty \\
      \fademinfty & \fademinfty & \fademinfty & \fademinfty & \fademinfty
    \end{bNiceMatrix}
    \vx \oplus
    \begin{bNiceMatrix}
      \fademinfty \\
      \fademinfty \\
      \fademinfty \\
      \fademinfty \\
      \fademinfty
    \end{bNiceMatrix}
    = {[}\btR](\vx)
    \\[3.875em]
    {[}\btR \bZ \tZ](\vx) =
    \begin{bNiceMatrix}
      2 & 5 & 2 & \fademinfty & \fademinfty \\
      \fademinfty & \fademinfty & \fademinfty & \fademinfty & \fademinfty \\
      \fademinfty & \fademinfty & \fademinfty & \fademinfty & \fademinfty \\
      \fademinfty & \fademinfty & \fademinfty & \fademinfty & \fademinfty \\
      \fademinfty & \fademinfty & \fademinfty & \fademinfty & \fademinfty
    \end{bNiceMatrix}
    \vx \oplus
    \begin{bNiceMatrix}
      0 \\
      \fademinfty \\
      \fademinfty \\
      \fademinfty \\
      \fademinfty
    \end{bNiceMatrix}
    & \gg &
    \begin{bNiceMatrix}
      0 & 4 & 0 & \fademinfty & \fademinfty \\
      \fademinfty & \fademinfty & \fademinfty & \fademinfty & \fademinfty \\
      \fademinfty & \fademinfty & \fademinfty & \fademinfty & \fademinfty \\
      \fademinfty & \fademinfty & \fademinfty & \fademinfty & \fademinfty \\
      \fademinfty & \fademinfty & \fademinfty & \fademinfty & \fademinfty
    \end{bNiceMatrix}
    \vx \oplus
    \begin{bNiceMatrix}
      \fademinfty \\
      \fademinfty \\
      \fademinfty \\
      \fademinfty \\
      \fademinfty
    \end{bNiceMatrix}
    = {[}\btR \tZ](\vx)
    \\[3.875em]
    {[}\btR \bZ \tO](\vx) =
    \begin{bNiceMatrix}
      2 & \fademinfty & \fademinfty & \fademinfty & \fademinfty \\
      \fademinfty & \fademinfty & \fademinfty & \fademinfty & \fademinfty \\
      \fademinfty & \fademinfty & \fademinfty & \fademinfty & \fademinfty \\
      \fademinfty & \fademinfty & \fademinfty & \fademinfty & \fademinfty \\
      \fademinfty & \fademinfty & \fademinfty & \fademinfty & \fademinfty
    \end{bNiceMatrix}
    \vx \oplus
    \begin{bNiceMatrix}
      0 \\
      \fademinfty \\
      \fademinfty \\
      \fademinfty \\
      \fademinfty
    \end{bNiceMatrix}
    & \gg &
    \begin{bNiceMatrix}
      1 & \fademinfty & \fademinfty & \fademinfty & \fademinfty \\
      \fademinfty & \fademinfty & \fademinfty & \fademinfty & \fademinfty \\
      \fademinfty & \fademinfty & \fademinfty & \fademinfty & \fademinfty \\
      \fademinfty & \fademinfty & \fademinfty & \fademinfty & \fademinfty \\
      \fademinfty & \fademinfty & \fademinfty & \fademinfty & \fademinfty
    \end{bNiceMatrix}
    \vx \oplus
    \begin{bNiceMatrix}
      \fademinfty \\
      \fademinfty \\
      \fademinfty \\
      \fademinfty \\
      \fademinfty
    \end{bNiceMatrix}
    = {[}\btR \tO](\vx)
    \\[3.875em]
    {[}\btR \bO \tO](\vx) =
    \begin{bNiceMatrix}
      2 & 5 & 2 & \fademinfty & \fademinfty \\
      \fademinfty & \fademinfty & \fademinfty & \fademinfty & \fademinfty \\
      \fademinfty & \fademinfty & \fademinfty & \fademinfty & \fademinfty \\
      \fademinfty & \fademinfty & \fademinfty & \fademinfty & \fademinfty \\
      \fademinfty & \fademinfty & \fademinfty & \fademinfty & \fademinfty
    \end{bNiceMatrix}
    \vx \oplus
    \begin{bNiceMatrix}
      0 \\
      \fademinfty \\
      \fademinfty \\
      \fademinfty \\
      \fademinfty
    \end{bNiceMatrix}
    & \gg &
    \begin{bNiceMatrix}
      1 & 4 & 0 & \fademinfty & \fademinfty \\
      \fademinfty & \fademinfty & \fademinfty & \fademinfty & \fademinfty \\
      \fademinfty & \fademinfty & \fademinfty & \fademinfty & \fademinfty \\
      \fademinfty & \fademinfty & \fademinfty & \fademinfty & \fademinfty \\
      \fademinfty & \fademinfty & \fademinfty & \fademinfty & \fademinfty
    \end{bNiceMatrix}
    \vx \oplus
    \begin{bNiceMatrix}
      \fademinfty \\
      \fademinfty \\
      \fademinfty \\
      \fademinfty \\
      \fademinfty
    \end{bNiceMatrix}
    = {[}\btR \tT \tO](\vx)
    \\[3.875em]
    {[}\btR \bO \tT](\vx) =
    \begin{bNiceMatrix}
      2 & \fademinfty & 2 & \fademinfty & 2 \\
      \fademinfty & \fademinfty & \fademinfty & \fademinfty & \fademinfty \\
      \fademinfty & \fademinfty & \fademinfty & \fademinfty & \fademinfty \\
      \fademinfty & \fademinfty & \fademinfty & \fademinfty & \fademinfty \\
      \fademinfty & \fademinfty & \fademinfty & \fademinfty & \fademinfty
    \end{bNiceMatrix}
    \vx \oplus
    \begin{bNiceMatrix}
      0 \\
      \fademinfty \\
      \fademinfty \\
      \fademinfty \\
      \fademinfty
    \end{bNiceMatrix}
    & \gg &
    \begin{bNiceMatrix}
      0 & \fademinfty & 1 & \fademinfty & 0 \\
      \fademinfty & \fademinfty & \fademinfty & \fademinfty & \fademinfty \\
      \fademinfty & \fademinfty & \fademinfty & \fademinfty & \fademinfty \\
      \fademinfty & \fademinfty & \fademinfty & \fademinfty & \fademinfty \\
      \fademinfty & \fademinfty & \fademinfty & \fademinfty & \fademinfty
    \end{bNiceMatrix}
    \vx \oplus
    \begin{bNiceMatrix}
      \fademinfty \\
      \fademinfty \\
      \fademinfty \\
      \fademinfty \\
      \fademinfty
    \end{bNiceMatrix}
    = {[}\btR \tT \tT](\vx)
    \\[3.875em]
    {[}\tZ \bZ](\vx) =
    \begin{bNiceMatrix}
      2 & 6 & 4 & 2 & \fademinfty \\
      \fademinfty & 6 & 4 & \fademinfty & \fademinfty \\
      2 & 6 & 4 & \fademinfty & \fademinfty \\
      2 & 5 & 3 & 2 & \fademinfty \\
      \fademinfty & \fademinfty & \fademinfty & \fademinfty & \fademinfty
    \end{bNiceMatrix}
    \vx \oplus
    \begin{bNiceMatrix}
      0 \\
      \fademinfty \\
      \fademinfty \\
      0 \\
      \fademinfty
    \end{bNiceMatrix}
    & \geq &
    \begin{bNiceMatrix}
      2 & 5 & 2 & \fademinfty & \fademinfty \\
      \fademinfty & 6 & 0 & \fademinfty & \fademinfty \\
      2 & 6 & 2 & \fademinfty & \fademinfty \\
      \fademinfty & \fademinfty & \fademinfty & \fademinfty & \fademinfty \\
      \fademinfty & \fademinfty & \fademinfty & \fademinfty & \fademinfty
    \end{bNiceMatrix}
    \vx \oplus
    \begin{bNiceMatrix}
      0 \\
      \fademinfty \\
      \fademinfty \\
      \fademinfty \\
      \fademinfty
    \end{bNiceMatrix}
    = {[}\bZ \tZ](\vx)
    \\[3.875em]
    {[}\tO \bZ](\vx) =
    \begin{bNiceMatrix}
      \fademinfty & \fademinfty & \fademinfty & 3 & \fademinfty \\
      2 & 6 & 4 & \fademinfty & \fademinfty \\
      2 & 6 & 4 & \fademinfty & \fademinfty \\
      \fademinfty & \fademinfty & \fademinfty & 2 & \fademinfty \\
      2 & 5 & 3 & 2 & \fademinfty
    \end{bNiceMatrix}
    \vx \oplus
    \begin{bNiceMatrix}
      1 \\
      \fademinfty \\
      \fademinfty \\
      0 \\
      0
    \end{bNiceMatrix}
    & \geq &
    \begin{bNiceMatrix}
      \fademinfty & \fademinfty & \fademinfty & \fademinfty & \fademinfty \\
      0 & 6 & 0 & \fademinfty & \fademinfty \\
      2 & 6 & 2 & \fademinfty & \fademinfty \\
      \fademinfty & \fademinfty & \fademinfty & \fademinfty & \fademinfty \\
      \fademinfty & \fademinfty & \fademinfty & \fademinfty & \fademinfty
    \end{bNiceMatrix}
    \vx \oplus
    \begin{bNiceMatrix}
      0 \\
      \fademinfty \\
      \fademinfty \\
      \fademinfty \\
      \fademinfty
    \end{bNiceMatrix}
    = {[}\bO \tZ](\vx)
    \\[3.875em]
    {[}\tT \bZ](\vx) =
    \begin{bNiceMatrix}
      2 & \fademinfty & \fademinfty & 2 & \fademinfty \\
      \fademinfty & 6 & 4 & \fademinfty & \fademinfty \\
      3 & \fademinfty & \fademinfty & 2 & \fademinfty \\
      \fademinfty & \fademinfty & \fademinfty & \fademinfty & \fademinfty \\
      2 & \fademinfty & \fademinfty & 2 & \fademinfty
    \end{bNiceMatrix}
    \vx \oplus
    \begin{bNiceMatrix}
      0 \\
      \fademinfty \\
      0 \\
      \fademinfty \\
      0
    \end{bNiceMatrix}
    & \geq &
    \begin{bNiceMatrix}
      2 & \fademinfty & \fademinfty & \fademinfty & \fademinfty \\
      \fademinfty & 6 & 2 & \fademinfty & \fademinfty \\
      3 & \fademinfty & \fademinfty & \fademinfty & \fademinfty \\
      \fademinfty & \fademinfty & \fademinfty & \fademinfty & \fademinfty \\
      \fademinfty & \fademinfty & \fademinfty & \fademinfty & \fademinfty
    \end{bNiceMatrix}
    \vx \oplus
    \begin{bNiceMatrix}
      0 \\
      \fademinfty \\
      \fademinfty \\
      \fademinfty \\
      \fademinfty
    \end{bNiceMatrix}
    = {[}\bZ \tO](\vx)
    \\[3.875em]
    {[}\tZ \bO](\vx) =
    \begin{bNiceMatrix}
      4 & 6 & 4 & \fademinfty & 4 \\
      4 & 6 & 4 & \fademinfty & 4 \\
      4 & 6 & 4 & \fademinfty & 4 \\
      3 & 5 & 3 & \fademinfty & 3 \\
      \fademinfty & \fademinfty & \fademinfty & \fademinfty & \fademinfty
    \end{bNiceMatrix}
    \vx \oplus
    \begin{bNiceMatrix}
      0 \\
      \fademinfty \\
      \fademinfty \\
      0 \\
      \fademinfty
    \end{bNiceMatrix}
    & \geq &
    \begin{bNiceMatrix}
      2 & 5 & 2 & \fademinfty & \fademinfty \\
      1 & 6 & 2 & \fademinfty & \fademinfty \\
      3 & 6 & 2 & \fademinfty & \fademinfty \\
      \fademinfty & \fademinfty & \fademinfty & \fademinfty & \fademinfty \\
      \fademinfty & \fademinfty & \fademinfty & \fademinfty & \fademinfty
    \end{bNiceMatrix}
    \vx \oplus
    \begin{bNiceMatrix}
      0 \\
      \fademinfty \\
      \fademinfty \\
      \fademinfty \\
      \fademinfty
    \end{bNiceMatrix}
    = {[}\bO \tO](\vx)
    \\[3.875em]
    {[}\tO \bO](\vx) =
    \begin{bNiceMatrix}
      \fademinfty & \fademinfty & \fademinfty & \fademinfty & 3 \\
      4 & 6 & 4 & \fademinfty & 4 \\
      4 & 6 & 4 & \fademinfty & 4 \\
      \fademinfty & \fademinfty & \fademinfty & \fademinfty & 2 \\
      3 & 5 & 3 & \fademinfty & 3
    \end{bNiceMatrix}
    \vx \oplus
    \begin{bNiceMatrix}
      1 \\
      \fademinfty \\
      \fademinfty \\
      0 \\
      0
    \end{bNiceMatrix}
    & \geq &
    \begin{bNiceMatrix}
      \fademinfty & \fademinfty & \fademinfty & \fademinfty & \fademinfty \\
      0 & 6 & 1 & \fademinfty & 0 \\
      2 & \fademinfty & 2 & \fademinfty & \fademinfty \\
      \fademinfty & \fademinfty & \fademinfty & \fademinfty & \fademinfty \\
      \fademinfty & \fademinfty & \fademinfty & \fademinfty & \fademinfty
    \end{bNiceMatrix}
    \vx \oplus
    \begin{bNiceMatrix}
      0 \\
      \fademinfty \\
      \fademinfty \\
      \fademinfty \\
      \fademinfty
    \end{bNiceMatrix}
    = {[}\bZ \tT](\vx)
    \\[3.875em]
    {[}\tT \bO](\vx) =
    \begin{bNiceMatrix}
      2 & \fademinfty & 2 & \fademinfty & 2 \\
      4 & 6 & 4 & \fademinfty & 4 \\
      3 & \fademinfty & 3 & \fademinfty & 2 \\
      \fademinfty & \fademinfty & \fademinfty & \fademinfty & \fademinfty \\
      2 & \fademinfty & 2 & \fademinfty & 2
    \end{bNiceMatrix}
    \vx \oplus
    \begin{bNiceMatrix}
      0 \\
      \fademinfty \\
      0 \\
      \fademinfty \\
      0
    \end{bNiceMatrix}
    & \geq &
    \begin{bNiceMatrix}
      2 & \fademinfty & 2 & \fademinfty & 2 \\
      0 & 6 & 1 & \fademinfty & 0 \\
      2 & \fademinfty & 3 & \fademinfty & 2 \\
      \fademinfty & \fademinfty & \fademinfty & \fademinfty & \fademinfty \\
      \fademinfty & \fademinfty & \fademinfty & \fademinfty & \fademinfty
    \end{bNiceMatrix}
    \vx \oplus
    \begin{bNiceMatrix}
      0 \\
      \fademinfty \\
      \fademinfty \\
      \fademinfty \\
      \fademinfty
    \end{bNiceMatrix}
    = {[}\bO \tT](\vx)
    \\[3.875em]
    {[}\tZ \btL](\vx) =
    \begin{bNiceMatrix}
      \fademinfty & \fademinfty & \fademinfty & \fademinfty & \fademinfty \\
      \fademinfty & \fademinfty & \fademinfty & \fademinfty & \fademinfty \\
      \fademinfty & \fademinfty & \fademinfty & \fademinfty & \fademinfty \\
      \fademinfty & \fademinfty & \fademinfty & \fademinfty & \fademinfty \\
      \fademinfty & \fademinfty & \fademinfty & \fademinfty & \fademinfty
    \end{bNiceMatrix}
    \vx \oplus
    \begin{bNiceMatrix}
      6 \\
      6 \\
      6 \\
      5 \\
      \fademinfty
    \end{bNiceMatrix}
    & \geq &
    \begin{bNiceMatrix}
      \fademinfty & \fademinfty & \fademinfty & \fademinfty & \fademinfty \\
      \fademinfty & \fademinfty & \fademinfty & \fademinfty & \fademinfty \\
      \fademinfty & \fademinfty & \fademinfty & \fademinfty & \fademinfty \\
      \fademinfty & \fademinfty & \fademinfty & \fademinfty & \fademinfty \\
      \fademinfty & \fademinfty & \fademinfty & \fademinfty & \fademinfty
    \end{bNiceMatrix}
    \vx \oplus
    \begin{bNiceMatrix}
      6 \\
      4 \\
      2 \\
      \fademinfty \\
      \fademinfty
    \end{bNiceMatrix}
    = {[}\bO \btL](\vx)
    \\[3.875em]
    {[}\tO \btL](\vx) =
    \begin{bNiceMatrix}
      \fademinfty & \fademinfty & \fademinfty & \fademinfty & \fademinfty \\
      \fademinfty & \fademinfty & \fademinfty & \fademinfty & \fademinfty \\
      \fademinfty & \fademinfty & \fademinfty & \fademinfty & \fademinfty \\
      \fademinfty & \fademinfty & \fademinfty & \fademinfty & \fademinfty \\
      \fademinfty & \fademinfty & \fademinfty & \fademinfty & \fademinfty
    \end{bNiceMatrix}
    \vx \oplus
    \begin{bNiceMatrix}
      1 \\
      6 \\
      6 \\
      0 \\
      5
    \end{bNiceMatrix}
    & \geq &
    \begin{bNiceMatrix}
      \fademinfty & \fademinfty & \fademinfty & \fademinfty & \fademinfty \\
      \fademinfty & \fademinfty & \fademinfty & \fademinfty & \fademinfty \\
      \fademinfty & \fademinfty & \fademinfty & \fademinfty & \fademinfty \\
      \fademinfty & \fademinfty & \fademinfty & \fademinfty & \fademinfty \\
      \fademinfty & \fademinfty & \fademinfty & \fademinfty & \fademinfty
    \end{bNiceMatrix}
    \vx \oplus
    \begin{bNiceMatrix}
      0 \\
      6 \\
      2 \\
      \fademinfty \\
      \fademinfty
    \end{bNiceMatrix}
    = {[}\bZ \bZ \btL](\vx)
    \\[3.875em]
    {[}\tT \btL](\vx) =
    \begin{bNiceMatrix}
      \fademinfty & \fademinfty & \fademinfty & \fademinfty & \fademinfty \\
      \fademinfty & \fademinfty & \fademinfty & \fademinfty & \fademinfty \\
      \fademinfty & \fademinfty & \fademinfty & \fademinfty & \fademinfty \\
      \fademinfty & \fademinfty & \fademinfty & \fademinfty & \fademinfty \\
      \fademinfty & \fademinfty & \fademinfty & \fademinfty & \fademinfty
    \end{bNiceMatrix}
    \vx \oplus
    \begin{bNiceMatrix}
      0 \\
      6 \\
      4 \\
      \fademinfty \\
      4
    \end{bNiceMatrix}
    & \geq &
    \begin{bNiceMatrix}
      \fademinfty & \fademinfty & \fademinfty & \fademinfty & \fademinfty \\
      \fademinfty & \fademinfty & \fademinfty & \fademinfty & \fademinfty \\
      \fademinfty & \fademinfty & \fademinfty & \fademinfty & \fademinfty \\
      \fademinfty & \fademinfty & \fademinfty & \fademinfty & \fademinfty \\
      \fademinfty & \fademinfty & \fademinfty & \fademinfty & \fademinfty
    \end{bNiceMatrix}
    \vx \oplus
    \begin{bNiceMatrix}
      0 \\
      6 \\
      4 \\
      \fademinfty \\
      \fademinfty
    \end{bNiceMatrix}
    = {[}\bO \bZ \btL](\vx)
  \end{longtable*}
\end{NiceMatrixBlock}

\section{More Problems to Approach via Rewriting}
\label{sec:more-rewriting}

We adapt the mixed base scheme for several other open problems and reformulate them as termination of relatively small rewriting systems.

\subsection{Mahler's \texorpdfstring{$3/2$}{3/2} Problem}
\label{sec:mahler-32-problem}
\begin{definition}
  Let $\xi \in \R_{>0}$ be a positive real number. It is called a \emph{$Z$-number} if for all $k \in \N$ we have $\fracp\left(\xi \left(\frac{3}{2}\right)^k \right) < \frac{1}{2}$, where $\fracp(\cdot)$ denotes the fractional part of the number.
\end{definition}

Mahler~\cite{Mah68} conjectured that there are no $Z$-numbers. Moreover, he considered a generalized Collatz function $M \colon \Np \to \Np$, defined as follows.
\begin{equation*}
  M(n) =
  \begin{cases}
    \frac{3n}{2} & \text{if } n \equiv 0 \pmod 2 \\
    \frac{3n+1}{2} & \text{if } n \equiv 1 \pmod 2
  \end{cases}
\end{equation*}
He related the behaviors of $M$-trajectories to the existence of $Z$-numbers:
\begin{theorem}\label{thm:mahler}
  For $n \in \Np$, if a $Z$-number exists in the interval $[n, n+1)$, then there is no $k \in \N$ for which $M^k(n) \equiv 3 \pmod 4$.
\end{theorem}
In order to formulate this as a termination problem, we split the odd case of $M$ into two and leave one undefined:
\begin{equation*}
  M'(n) =
  \begin{cases}
    \frac{3n}{2} & \text{if } n \equiv 0 \pmod 2 \\
    \frac{3n+1}{2} & \text{if } n \equiv 1 \pmod 4 \\
    \bot & \text{if } n \equiv 3 \pmod 4
  \end{cases}
\end{equation*}
Given $n \in \Np$, if the trajectory ${M'}^{\N}(n)$ contains $\bot$ then by the contrapositive of \cref{thm:mahler} there is no $Z$-number in the interval $[n, n+1)$. When proving nonexistence, we may assume without loss of generality that a $Z$-number is at least $1$, since if $\xi \in (0,1)$ is a $Z$-number then $\xi\left(\frac{3}{2}\right)^m \in [1,\infty)$ is another one for sufficiently large $m \in \N$. Thus, the nonexistence of $Z$-numbers can be established by proving that $M'$ is convergent, which is equivalent to the termination of the following rewriting system $\rM$. In order to ensure termination at the case $n \equiv 3 \pmod 4$, there is no rule with the LHS $\bO \bO \btR$.
\begin{equation*}
  \begin{array}[t]{rcl}
    \bZ \btR & \to & \tZ \btR \\
    \bZ \bO \btR & \to & \tO \tZ \btR \\
  \end{array}
  \qquad
  \begin{array}[t]{rcl}
    \bZ \tZ & \to & \tZ \bZ \\
    \bZ \tO & \to & \tZ \bO \\
    \bZ \tT & \to & \tO \bZ
  \end{array}
  \qquad
  \begin{array}[t]{rcl}
    \bO \tZ & \to & \tO \bO \\
    \bO \tO & \to & \tT \bZ \\
    \bO \tT & \to & \tT \bO
  \end{array}
  \qquad
  \begin{array}[t]{rcl}
    \btL \tZ & \to & \btL \bO \\
    \btL \tO & \to & \btL \bZ \bZ \\
    \btL \tT & \to & \btL \bZ \bO
  \end{array}
\end{equation*}

\subsection{Halting Problem for Busy Beavers}
\label{sec:halting-problem-busy-beavers}

The busy beaver problem~\cite{Rad62} concerns finding binary-alphabet Turing machines with $n$ states that, when given an input tape of all $0$s, write the largest number of $1$s on the tape upon halting. For each $n$, a machine that achieves this is called an \emph{$n$-state Busy Beaver} (or BB-$n$ for short). Note that this definition only requires the machines to halt on all-$0$ inputs, leaving the behavior on other inputs unspecified and allowing them not to halt in general.

Michel~\cite{Mic15} observed that for $n \in \{2, 3, 4\}$, the busy beaver machines are all \emph{total Turing machines}, i.e., they halt on all inputs, and moreover proved that they all simulate some generalized Collatz function. It is an open problem whether all busy beavers are total Turing machines. In particular, it is unknown whether the current BB-$5$ candidate is total. Michel showed that the BB-$5$ candidate simulates the following generalized Collatz function.
\begin{equation*}
  B(n) =
  \begin{cases}
    \frac{5n+18}{3} & \text{if } n \equiv 0 \pmod 3 \\
    \frac{5n+22}{3} & \text{if } n \equiv 1 \pmod 3 \\
    \bot & \text{if } n \equiv 2 \pmod 3
  \end{cases}
\end{equation*}

Convergence of the above function can be studied via the termination of a corresponding rewriting system, although the translation is not as straightforward as the other generalized Collatz functions we have in this paper. Specifically, the mixed base scheme results in a compact rewriting system only when all the mappings $an + b \mapsto cn + d$ applied by the generalized Collatz function at hand satisfy $d < c$. This is not the case for the function $B$ above since it requires us to express $3n \mapsto 5n + 6$ and $3n + 1 \mapsto 5n + 9$. These mappings are expressed more easily by the help of unary representations, so we can have the following hybrid rewriting system $\rB\rB'$ over the alphabet $\{\inc, \pZ, \tZ, \btL, \btR\}$ to simulate the function $B$, where the new symbol $\pZ$ has the functional view $\pZ(x) = 5x$.
\begin{equation*}
  \begin{array}[t]{rcl}
    \tZ \btR & \to & \inc \pZ \inc \btR \\
    \tZ \inc \btR & \to & \inc \pZ \inc \inc \inc \inc \btR
  \end{array}
  \qquad
  \begin{array}[t]{rcl}
    \tZ \inc \inc \inc & \to & \inc \tZ \\
    \inc \pZ & \to & \pZ \inc \inc \inc \inc \inc \\
    \tZ \pZ & \to & \pZ \tZ
  \end{array}
  \qquad
  \begin{array}[t]{rcl}
    \btL \inc \inc & \to & \btL \tZ \\
    \btL \pZ & \to & \btL \tZ \inc \inc \\
  \end{array}
\end{equation*}
We did not succeed in proving the termination of this rewriting system, which is not too surprising since the hybrid scheme does not even pass the test of proving the convergence of $W$. Alternatively, we can split the cases of $B$ further for the congruence classes modulo $3^k$ with $k > 1$ until the iterated application of $B$ can either be expressed as $an + b \mapsto cn + d$ with $d < c$, or it reaches $\bot$. After we perform this procedure while ensuring that the number of newly introduced cases are kept to a minimum, we end up with the following mappings to simulate $B$ in an accelerated manner (similar to the idea in \cref{sec:odd-trajectories}).
\begin{equation}
  \label{eq:bb5-mappings}
  \begin{array}{rcl c rcl}
    9n && & \; \; \mapsto \; \; & 25n &+& 16 \\
    9n &+& 1 & \; \; \mapsto \; \; & 25n &+& 21 \\
    27n &+& 6 & \; \; \mapsto \; \; & 125n &+& 64 \\
    27n &+& 7 & \; \; \mapsto \; \; & 125n &+& 71 \\
    27n &+& 16 & \; \; \mapsto \; \; & 125n &+& 114 \\
    81n &+& 51 & \; \; \mapsto \; \; & 625n &+& 459 \\
    243n &+& 78 & \; \; \mapsto \; \; & 3125n &+& 1116 \\
    243n &+& 159 & \; \; \mapsto \; \; & 3125n &+& 2159
  \end{array}
\end{equation}
Consider the symbols $\{\tZ, \tO, \tT, \pZ, \pO, \pT, \pTh, \pF, \btL, \btR\}$ with the following functional views.
\begin{equation*}
  \begin{array}{lcl}
    \tZ(x) & = & 3x \\
    \tO(x) & = & 3x + 1 \\
    \tT(x) & = & 3x + 2
  \end{array}
  \qquad
  \begin{array}{lcl}
    \pZ(x) & = & 5x \\
    \pO(x) & = & 5x + 1 \\
    \pT(x) & = & 5x + 2 \\
    \pTh(x) & = & 5x + 3 \\
    \pF(x) & = & 5x + 4
  \end{array}
  \qquad
  \begin{array}{lcl}
    \btL(x) & = & 0 \\
    \btR(x) & = & x
  \end{array}
\end{equation*}
With these symbols, we have the following mixed $\{3,5\}$-ary (ternary--quinary) system $\rB\rB$ the termination of which is equivalent to the convergence of $B$. The dynamic rules on the left implement the mappings from~\cref{eq:bb5-mappings} and the rest are auxiliary rules that push ternary symbols towards the rightmost end of the string while preserving its value. We were unable to prove the termination of this system, so we include it as yet another challenge for automated termination proving.
\begin{equation*}
  \begin{array}[t]{rcl}
    \tZ \tZ \btR & \to & \pTh \pO \btR \\
    \tZ \tO \btR & \to & \pF \pO \btR \\
    \tZ \tT \tZ \btR & \to & \pT \pT \pF \btR \\
    \tZ \tT \tO \btR & \to & \pT \pF \pO \btR \\
    \tO \tT \tO \btR & \to & \pF \pT \pF \btR \\
    \tO \tT \tT \tZ \btR & \to & \pTh \pTh \pO \pF \btR \\
    \tZ \tT \tT \tT \tZ \btR & \to & \pO \pTh \pF \pTh \pO \btR \\
    \tO \tT \tT \tT \tZ \btR & \to & \pTh \pT \pO \pO \pF \btR
  \end{array}
  \qquad
  \begin{array}[t]{rcl}
    \tZ \pZ & \to & \pZ \tZ \\
    \tZ \pO & \to & \pZ \tO \\
    \tZ \pT & \to & \pZ \tT \\
    \tZ \pTh & \to & \pO \tZ \\
    \tZ \pF & \to & \pO \tO
  \end{array}
  \qquad
  \begin{array}[t]{rcl}
    \tO \pZ & \to & \pO \tT \\
    \tO \pO & \to & \pT \tZ \\
    \tO \pT & \to & \pT \tO \\
    \tO \pTh & \to & \pT \tT \\
    \tO \pF & \to & \pTh \tZ
  \end{array}
  \qquad
  \begin{array}[t]{rcl}
    \tT \pZ & \to & \pTh \tO \\
    \tT \pO & \to & \pTh \tT \\
    \tT \pT & \to & \pF \tZ \\
    \tT \pTh & \to & \pF \tO \\
    \tT \pF & \to & \pF \tT
  \end{array}
  \qquad
  \begin{array}[t]{rcl}
    \btL \pZ & \to & \btL \tZ \\
    \btL \pO & \to & \btL \tO \\
    \btL \pT & \to & \btL \tT \\
    \btL \pTh & \to & \btL \tO \tZ \\
    \btL \pF & \to & \btL \tO \tO
  \end{array}
\end{equation*}

\subsection{Ternary Expansions of \texorpdfstring{$2^n$}{2n}}
\label{sec:ternary-2n}

Erd\H{o}s~\cite{Erd79} asked: When does the ternary expansion of $2^n$ omit the digit $2$? This is the case for $2^0 = (1)_3$, $2^2 = (11)_3$, and $2^8 = (100111)_3$. He conjectured that it does not happen for $n > 8$. This conjecture can be reformulated as the statement that the following rewriting system $\rE$, where $\inv{\rX} = \{r \to \ell \mid \ell \to r \in \rX\}$, is terminating on all initial strings of the form $\btL \bZ^8 \bZ^+ \btR$.
\begin{equation*}
  \begin{array}{rcl}
    \tZ \btR & \to & \btR \\
    \tO \btR & \to & \btR \\
    \btL \btR & \to & \btL \btR
  \end{array}
  \qquad
  \inv{\rX} = \left\{
    \begin{array}{rcl}
      \tZ \bZ & \to & \bZ \tZ \\
      \tZ \bO & \to & \bZ \tO \\
      \tO \bZ & \to & \bZ \tT
    \end{array}
    \qquad
    \begin{array}{rcl}
      \tO \bO & \to & \bO \tZ \\
      \tT \bZ & \to & \bO \tO \\
      \tT \bO & \to & \bO \tT
    \end{array}
    \qquad
    \begin{array}{rcl}
      \btL \bO & \to & \btL \tZ \\
      \btL \bZ \bZ & \to & \btL \tO \\
      \btL \bZ \bO & \to & \btL \tT
    \end{array}
  \right\}
\end{equation*}
Given a string that corresponds to the binary representation of a power of $2$, the inverted system $\inv{\rX}$ essentially rewrites the string into ternary by pushing ternary symbols to the right without altering the value that the string represents. The two rules $\{\tZ \btR \to \btR,\ \tO \btR \to \btR\}$ remove the occurrences of the ternary digits $\tZ$ and $\tO$ (but not $\tT$). If the ternary expansion does not contain the digit $\tT$ then all digits will be removed, resulting in the string $\btL \btR$ that can then be rewritten to itself indefinitely.

This problem, as described, is an instance of ``local termination''~\cite{WdE10} since it is concerned with termination on not all possible strings but a subset of them. We have not performed experiments with this system or local termination yet and we leave this for future work.


\end{document}